%% file: vdpoly7Pankaj.tex
\documentclass[11pt]{article}

\usepackage{amsmath,amssymb,amsthm,color,graphicx}
\usepackage{euscript}
\usepackage{mathpazo}

\usepackage[top=1.25in,bottom=1.25in,left=1.25in,right=1.25in]{geometry}

\usepackage{microtype}
\usepackage{fixltx2e}
\usepackage{fullpage}
\usepackage{xspace}

\newcommand{\lemlab}[1]{\label{lemma:#1}}

\newcommand{\lemref}[1]{Lemma~\ref{lemma:#1}}

\newtheorem{theorem}{Theorem}[section]

\newtheorem{cor}[theorem]{Corollary}
\newtheorem{lemma}[theorem]{Lemma}

\graphicspath{{Figs/}}

\def \reals{{\mathbb R}}
\def \sphere{{\mathbb S}}

\def\I{\EuScript{I}}

\def\bd{{\partial}}
\def\eps{{\varepsilon}}
\def\poly{\diamond}

\newcommand{\ignore}[1]{}

\def\bisect{b}
\def\distfn{\varphi}
\def\etal{\textsl{et~al.}}

\def\conv{\mathop{\mathrm{conv}}}
\def\DT{\mathop{\mathrm{DT}}}
\def\VD{\mathop{\mathrm{VD}}}
\def\Vor{{\mathop{\mathrm{Vor}}}}

\def\Generic{Generic\xspace}
\def\generic{generic\xspace}
\def\Degenerate{Singular\xspace}
\def\degenerate{singular\xspace}

\makeatletter
\long\def\@makecaption#1#2{
   \vskip 10pt
   \setbox\@tempboxa\hbox{{\footnotesize \textbf{#1.} #2}}
   \ifdim \wd\@tempboxa >\hsize         
       {\footnotesize \textbf{#1.} #2\par}
     \else                              
       \hbox to\hsize{\hfil\box\@tempboxa\hfil}
   \fi}
\makeatother

\begin{document}

\begin{titlepage}

\title{Kinetic Voronoi Diagrams and Delaunay Triangulations under Polygonal Distance Functions\thanks{%
Work by P.A. and M.S. was supported by Grant 2012/229 from the U.S.-Israel
Binational Science Foundation.  Work by P.A. was also supported by NSF under
grants CCF-09-40671, CCF-10-12254, and CCF-11-61359, by an ARO contract
W911NF-13-P-0018, and by an ERDC contract W9132V-11-C-0003.
Work by Haim Kaplan has been supported by grant 822/10 from the
      Israel Science Foundation, grant 1161/2011 from the
      German-Israeli Science Foundation, and by the Israeli Centers
      for Research Excellence (I-CORE) program (center no.~4/11). %
Work by N.R. was partially supported by Grants 975/06 and 338/09 from the
Israel Science Fund, by Minerva Fellowship Program of the Max Planck Society, by the Fondation Sciences Math\'{e}matiques de Paris (FSMP), and by a public grant overseen by the French National Research Agency (ANR) as part of the "Investissements d'Avenir" program (reference: ANR-10-LABX-0098).
Work by Micha Sharir has also been supported by NSF Grant
      CCF-08-30272, by Grants 338/09 and 892/13 from the Israel
      Science Foundation, by the Israeli Centers for
      Research Excellence (I-CORE) program (center no.~4/11), and by
      the Hermann Minkowski--MINERVA Center for Geometry at Tel Aviv
      University.
}}

\author{Pankaj K. Agarwal\thanks{%
Department of Computer Science, Duke University, Durham, NC
27708-0129, USA; {\tt pankaj@cs.duke.edu}.}
\and
Haim Kaplan\thanks{%
School of Computer Science, Tel Aviv University, Tel~Aviv 69978, Israel;
{\tt haimk@tau.ac.il}.}
\and
Natan Rubin\thanks{%
Jussieu Institute of Mathematics, Pierre and Marie Curie University and Paris Diderot University, UMR 7586 du CNRS, Paris 75005, France;
{\tt rubinnat.ac@gmail.com}.}
\and
Micha Sharir\thanks{%
School of Computer Science, Tel Aviv University, Tel~Aviv 69978, Israel;
{\tt michas@tau.ac.il}.}
}

\maketitle

\begin{abstract}
Let $P$ be a set of $n$ points and $Q$ a convex $k$-gon in $\reals^2$.
We analyze in detail the topological (or discrete) changes in the structure
of the Voronoi diagram and the Delaunay triangulation of $P$, under the convex
distance function defined by $Q$, as the points of $P$ move along prespecified
continuous trajectories. Assuming that each point of $P$ moves along an algebraic
trajectory of bounded degree, we establish an upper bound of $O(k^4n\lambda_r(n))$
on the number of topological changes experienced by the diagrams throughout the motion;
here $\lambda_r(n)$ is the maximum length of an $(n,r)$-Davenport-Schinzel
sequence, and $r$ is a constant depending on the algebraic degree of the
motion of the points. Finally, we describe an algorithm for
efficiently maintaining the above structures, using the kinetic data structure (KDS)
framework.
\end{abstract}

\end{titlepage}

\section{Introduction}

Let $P$ be a set of $n$ points in $\reals^2$, and let $Q$ be a compact convex
(not necessarily polygonal) set in $\reals^2$ with nonempty interior and
with the origin lying in its interior.
For an ordered pair of points $x,y\in \reals^2$, the \emph{$Q$-distance}
from $x$ to $y$ is defined as
$$
d_Q(x,y) = \min\{\lambda \mid y\in x+\lambda Q\};
$$
$d_Q$ is a metric if and only if $Q$ is centrally symmetric
with respect to the origin (otherwise $d_Q$ need not be symmetric).
For a point of $P$, the \emph{$Q$-Voronoi cell} of $p$ is defined as
$$
\Vor_Q(p) = \{ x\in\reals^2 \mid d_Q(x,p) \le d_Q(x,p')\, \forall p'
\in P \}.
$$
If the points of $P$ are in general position with respect to $Q$
(see Section~\ref{sec:static} for the definition), the Voronoi cells
of points in $P$ are nonempty, have pairwise-disjoint interiors, and
partition the plane (see Figure~\ref{fig:VDT}(b)). The planar
subdivision induced by these Voronoi cells is referred to as the
\emph{$Q$-Voronoi diagram} of $P$ and we denote it as $\VD_Q(P)$.



\begin{figure}[htb]
\centering
\begin{tabular}{ccc}
\includegraphics[scale=0.45]{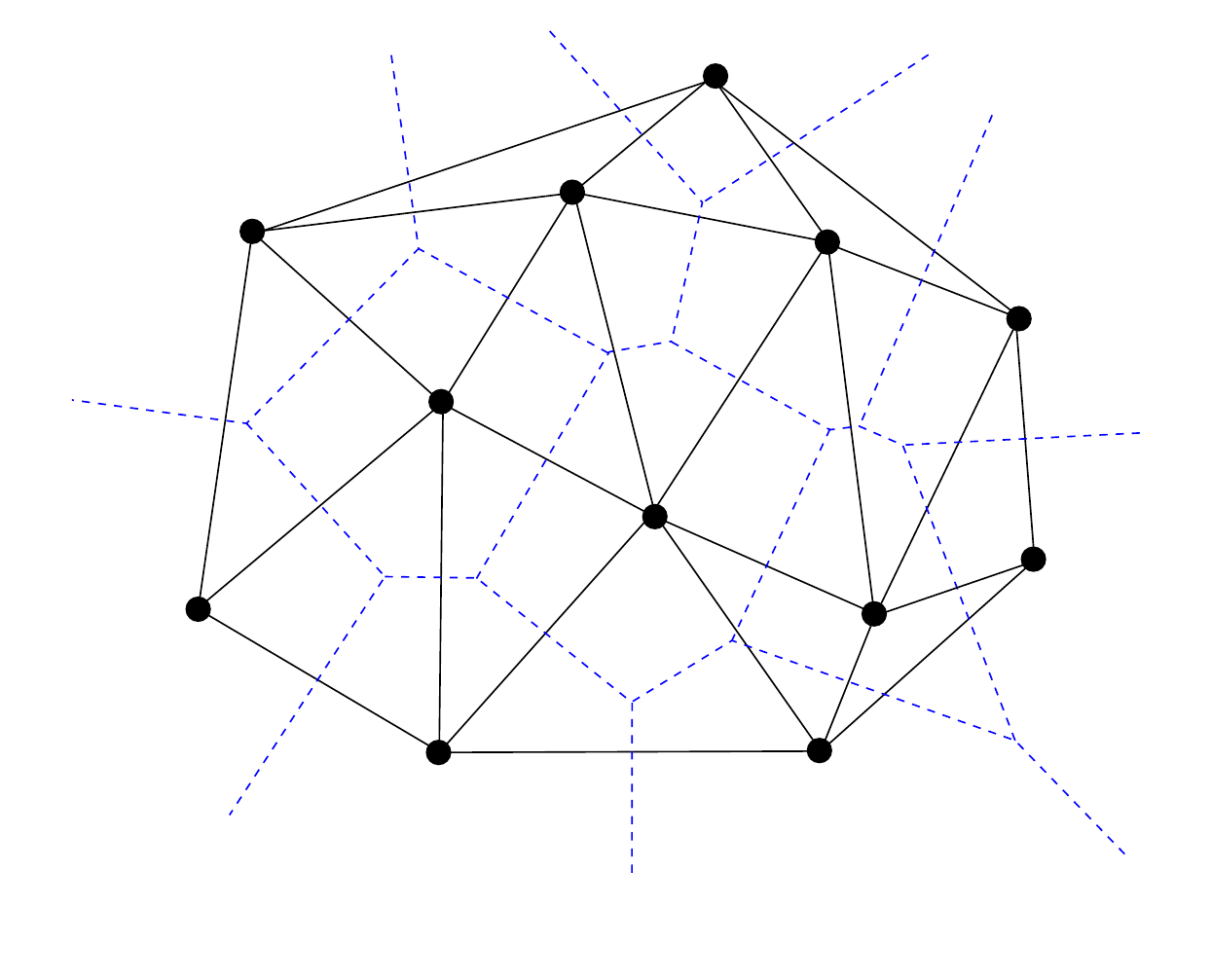}&
\includegraphics[scale=0.45]{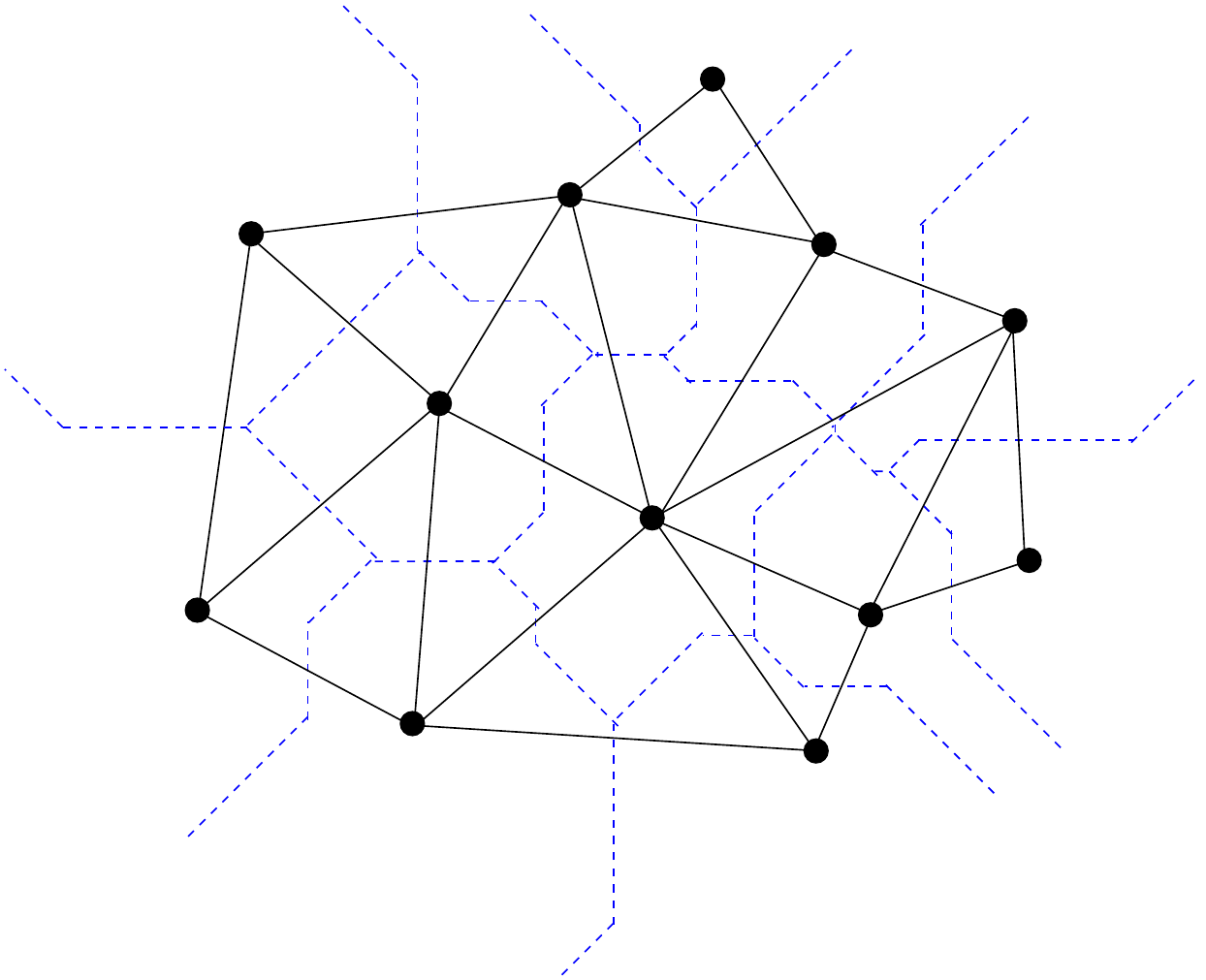}&
\hspace*{5mm}\includegraphics[scale=0.45]{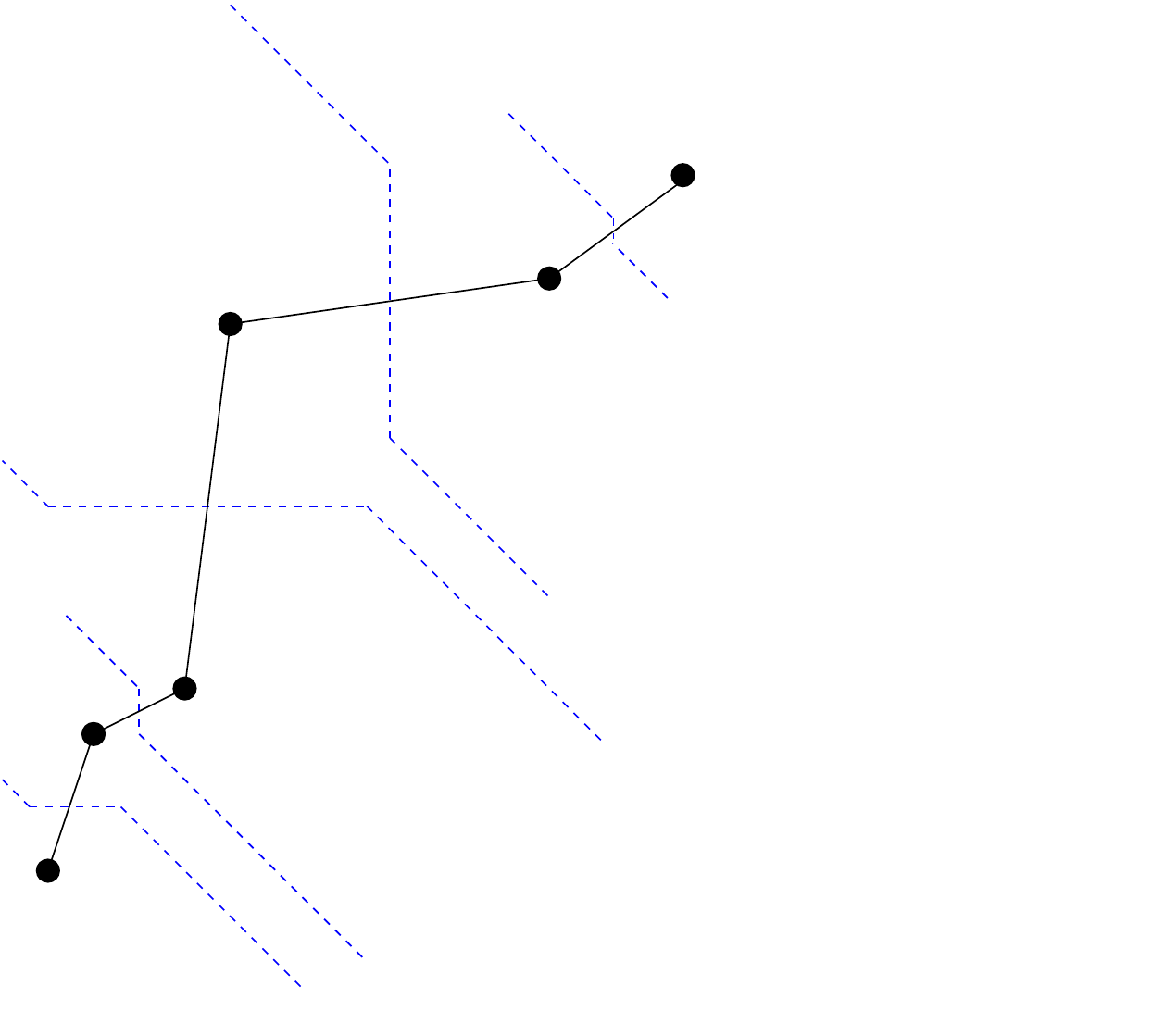}\\
\small (a) & \small (b) & \small (c)
\end{tabular}
\caption{(a) The Euclidean Voronoi diagram (dotted) and Delaunay
triangulation (solid). (b) $\VD_Q(P)$ and $\DT_Q(P)$ for an
axis-parallel square $Q$, i.e., the diagrams $\VD_Q$ and $\DT_Q$
under the $L_\infty$-metric. (c) $\VD_Q(P)$ and $\DT_Q(P)$, for the
same $Q$ as in (b), with an empty-interior support hull ($\VD_Q(P)$
has no vertices in this case).} \label{fig:VDT}
\end{figure}

The \emph{$Q$-Delaunay triangulation} of $P$, denoted by $\DT_Q(P)$,
is the dual structure of $\VD_Q(P)$. Namely, a pair of points
$p,q\in P$ are connected by an edge in $\DT_Q(P)$ if and only if the
boundaries of their respective $Q$-Voronoi cells $\Vor_Q(p)$ and
$\Vor_Q(q)$ share a $Q$-Voronoi edge, given by
$$
e_{pq}=\{x\in\reals^2 \mid d_Q(x,p)=d_Q(x,q) \le d_Q(x,p')\; \forall p'\in P \} .
$$
$\DT_Q(P)$ can be defined directly as well: it is composed of all
edges $pq$, with $p,q\in P$, for which there exists a homothetic
placement of $Q$ whose boundary touches $p$ and $q$ and whose
interior contains no other
points of $P$.\footnote{%
  We remark that $\Vor_Q(p)$ is often defined in the literature as the set
  $\Vor(p)=\{x \in \reals^2 \mid d_Q(p,x) \le d_Q(q,x) \; \forall q \in P\}$~\cite{AKL,CD,WWW}.
  If $Q$ is not centrally symmetric, then this definition of $\Vor(p)$ is not the same
  as the one given above.
  Furthermore, under this definition, $pq$ is an edge of $\DT_Q(P)$ if there exists
  a $P$-empty homothetic placement of $-Q$ (and not of $Q$) whose boundary touches $p$ and $q$.}
Placements of $Q$ with this latter property are called
\emph{$P$-empty}. If $Q$ is a circular disk then $\DT_Q(P)$ (resp.,
$\VD_Q(P)$) is the well-known Euclidean Delaunay triangulation
(resp., Voronoi diagram) of $P$.

If $P$ is in general position with respect to $Q$, then $\DT_Q(P)$
is spanned by so called $Q$-Delaunay triangles. Each of these
triangles $\triangle pqr$ corresponds to the (unique) $P$-empty
homothetic placement $Q_{pqr}$ of $Q$ whose boundary touches $p,q,$
and $r$. That is, $\triangle pqr$ corresponds to a $Q$-Voronoi
vertex $v_{pqr}$ that lies at equal $Q$-distances from $p$, $q$, and
$r$, so that $v_{pqr}$ is the center of $Q_{pqr}$ (that is,
$v_{pqr}$ is the image of the origin under the homothetic mapping of
$Q$ into $Q_{pqr}$).
If $Q$ is smooth (e.g., as in the Euclidean case), then $\DT_Q(P)$
is a triangulation of the convex hull of $P$; otherwise it is a
triangulation of a simply-connected polygonal subregion of
$\conv(P)$, sometimes referred to as the \emph{support hull} of $P$
(see~\cite{WWW} and Figure~\ref{fig:VDT}~(b)). The interior of the
support hull may be empty, as shown in Figure~\ref{fig:VDT}~(c).

In many applications of Delaunay/Voronoi methods (e.g., mesh
generation and kinetic collision detection), the points in $P$ move
continuously, so these structures need to be updated efficiently as
motion takes place. Even though the motion of the points of $P$ is
continuous, the topological  structures of $\VD_Q(P)$ and
$\DT_Q(P)$ change only at discrete times when certain \emph{events}
occur.\footnote{The {\em topological structures} of $\DT_Q(P)$ and $\VD_Q(P)$ are  the graphs that they define. More specifically, the topological structure of $\DT_Q(P)$ and $\VD_Q(P)$ consists of  the set of triples of points defining the
Voronoi vertices, and the sets of Voronoi and Delaunay edges. As we will see later each Voronoi  edge is a sequence  of one or more edgelets. Each such edgelet is  defined by a pair of edges of $Q$. The sequences of pairs of edges of $Q$ defining
the edgelet structures of the Voronoi edges
 are also part of the topological structure of $\VD_Q(P)$.}
Assume that each point of $P$ moves
independently along some known trajectory.
Let $p_i(t)=(x_i(t),y_i(t))$ denote the position of point $p_i$ at time
$t$, and set $P(t)=\{p_1(t), \ldots, p_n(t)\}$.  We call the motion of $P$
\emph{algebraic} if each $x_i(t),y_i(t)$ is a polynomial function of $t$,
and the \emph{degree} of the motion of $P$ is the maximum degree of these
polynomials.\footnote{%
  This assumption can be somewhat relaxed to allow more general motions,
  as can be inferred from the analysis in the paper.}

In this paper we focus on the case when $Q$ is a convex $k$-gon and
study the resulting $Q$-Voronoi and $Q$-Delaunay structures as each
point of $P$ moves continuously along an algebraic trajectory whose
degree is bounded by a constant. Since $Q$ will be either fixed or
obvious from the context, we will use the simplified notations
$\Vor(p)$, $\VD(P)$, and $\DT(P)$ to denote $\Vor_Q(p)$, $\VD_Q(P)$,
and $\DT_Q(P)$, respectively.

\paragraph{Related work.}
There has been extensive work on studying the geometric and topological structure of Voronoi diagrams and Delauany triangulations under convex distance
functions; see e.g.\ \cite{AKL} and the references therein. In the late 1970s,
$O(n\log n)$-time algorithms were proposed for computing the Voronoi diagram of a set of $n$ points in $\reals^2$ under any $L_p$-metric~\cite{Hw,Lee,LW}.
In the mid 1980s, Chew and Drsydale~\cite{CD} and
Widmayer~\etal~\cite{WWW} showed that if $Q$ is a convex $k$-gon, $\VD(P)$ has
$O(nk)$ size and that it can be computed in $O(kn\log n)$ time.
Motivated by a motion-planning application, Leven and
Sharir~\cite{LS} studied Voronoi diagrams under a convex polygonal distance
function for the case where the input sites are convex polygons.
Efficient divide-and-conquer, sweep-line, and edge-flip based incremental
algorithms have been  proposed to compute $\DT(P)$
directly~\cite{Dr,Ma,Sk}. Several recent works study
the structure of $\VD(P)$ under a convex polyhedral
distance function in $\reals^3$~\cite{BSTY,IKLM,KS}.

One of the hardest and best-known open problems in
discrete and computational geometry is to determine the asymptotic
behavior of the maximum possible number of discrete changes
experienced by the \emph{Euclidean} Delaunay triangulation during an
algebraic motion of constant degree of the points of $P$, where the prevailing
conjecture is that this number is nearly quadratic in $n$. A near-cubic bound was proved in~\cite{gmr-vdmpp-92}.
After
almost 25 years of no real progress, two recent works by one of the
authors \cite{Rubin, RubinUnit} substantiate this conjecture, and
establish an almost tight upper bound of $O(n^{2+\eps})$, for any
$\eps>0$, for restricted motions where any four points of $P$ can
become cocircular at most \emph{two} times (in \cite{Rubin}) or at
most \emph{three} times (in \cite{RubinUnit}). In particular, the
latter result \cite{RubinUnit}, involving at most three
cocircularities of any quadruple, applies to the case of points
moving along lines at common (unit) speed. Only near-cubic bounds are
known so far for more general motions.  Chew~\cite{Chew} showed that
the number of topological changes in the Delaunay triangulation
under the $L_1$ or $L_\infty$ metric is $O(n\lambda_r(n))$, where
$\lambda_r(n)$ is the almost-linear maximum length of a Davenport-Schinzel
sequence of order $r$ on $n$ symbols, and $r$ is a constant that
depends on the algebraic degree of the motions of the points. Chew's
result also holds for any convex quadrilateral $Q$.  He focuses on bounding
the number of changes in the Delaunay triangulation and not how it changes at
each ``event,'' so his analysis omits some critical details of how the Delaunay
triangulation and the Voronoi diagram change at an event; changes in the
topological structure of $\VD(P)$ are particularly subtle.
Chew remarks, without supplying any details, that his technique can be
extended to general convex polygons.

Later, Basch~\etal~\cite{bgh-dsmd-99} introduced
the {\em kinetic data structure} (KDS in short) framework
for designing efficient algorithms for maintaining a variety of
geometric and topological structures of mobile data. Several algorithms
have been developed in this framework for kinetically maintaining various
geometric and topological structures; see~\cite{Gu}.
The crux in designing an efficient KDS is finding a set of
{\em certificates} that, on one hand, ensure the correctness of
the configuration currently being maintained,
and, on the other hand, are inexpensive to maintain as the points move.
When a certificate fails during the motion of the objects,
the KDS fixes the configuration, replaces the failing certificate(s) by
new valid ones, and computes their failure times. The failure times, called \emph{events}, are stored in a priority queue, to keep track of the next event that the KDS needs to process.
The performance of a KDS is measured by
the number of events that it processes, the time taken to process
each event, and the total space used. If these parameters are small
(in a sense that may be problem dependent and has to be made precise),
the KDS is called, respectively, \emph{efficient}, \emph{responsive},
and \emph{compact}. See~\cite{bgh-dsmd-99,Gu} for details.

Delaunay triangulations and Voronoi diagrams are well suited for the KDS framework
because they admit local certifications associated with their
individual features.  These certifications fail only at the events
when the topological structure of the diagrams changes. The resulting KDS is
compact ($O(n)$ certificates suffice) and responsive (each update takes
$O(\log n)$ time, mainly to update the event priority queue), but its
efficiency, namely, the number of events that it has to process, depends on the
number of topological changes in $\DT(P)$, so a near quadratic
bound on the number of events for the Euclidean case holds only when
each point moves along some line with unit speed (or in similar situations
when only three co-circularities can exist for any quadruple of points).
A KDS for $\DT (P)$ when $Q$ is a convex quadrilateral was presented by
Abam and de Berg~\cite{AB}, but it is not straightforward to extend their
KDS for the case where $Q$ is a general convex $k$-gon.
Furthermore, it is not clear how to use their KDS for maintaining $\VD(P)$.

\paragraph{Our contribution.}
First, we establish a few key topological properties of $\VD(P)$
and $DT(P)$ when $P$ is a set of $n$ stationary points
in $\reals^2$ and $Q$ is a convex $k$-gon (Section~\ref{sec:static}).
Although these properties follow from earlier work on this topic
(see~\cite[Chapter~7]{AKL}),
we include them here because they are important for the kinetic setting
and most of them have not been stated in earlier work in exactly the same form as here.

Next, we characterize the topological changes
that $\VD(P)$ and $\DT(P)$ can undergo when
the points of $P$ move along continuous trajectories
(Section~\ref{sec:kinetic}).
These changes occur at critical moments
when the points of $P$ are not in $Q$-general position, so that some
$O(1)$ points of $P$ are involved in a \emph{degenerate}
configuration with respect to $Q$. The most ubiquitous type of such
events is when four points of $P$ become \emph{$Q$-cocircular}, in
the sense that there exists a $P$-empty  homothetic placement of $Q$ whose
boundary touches those four points.

We provide the first comprehensive and rigorous asymptotic analysis
of the maximum number of topological changes that $\VD(P)$ and
$\DT(P)$ can undergo during the motion of the points of $P$ (Section~\ref{sec:count}).
Specifically, if $Q$ has $k$ vertices, then $\VD(P)$ and
$\DT(P)$ experience $O(k^4n\lambda_r(n))$ such changes, where
$\lambda_r(n)$ is the almost-linear maximum length of a Davenport-Schinzel
sequence of order $r$ on $n$ symbols, and $r$ is some constant that
depends on the algebraic degree of the motions of the points.
Some of these changes occur as components of so-called
\emph{\degenerate sequences}, in which several events
that affect the structure of $\VD(P)$ and $\DT(P)$ occur simultaneously, and
their collective effect might involve a massive change in the
topological structures of these diagrams. These compound effects are
a consequence of the non-strict convexity of $Q$, and their
analysis requires extra care. Nevertheless, the above near-quadratic
bound on the number of changes also holds when we count each of
the individual critical events in any such sequence separately.


Finally, we describe an efficient algorithm for maintaining
$\VD(P)$ and $\DT(P)$ during an algebraic
motion of $P$, within the standard KDS framework
(Section~\ref{sec:KDS}).
Here we assume an algebraic model of computation, in which
algebraic computations, including solving a polynomial
equation of constant degree, can be performed in an exact manner,
in constant time.
The precise sense of this assumption is that comparisons between
algebraic quantities that are defined in this manner can be performed
exactly in constant time. This is a standard model used
widely in theory \cite[Section 6.1]{SA95} and nowadays also in practice
(see, e.g., \cite{CGAL}).
This model allows us to perform in constant time the various computations
that are needed by our KDS, the most ubiquitous of which are the calculation
of the failure times of the various certificates being maintained; see
Section~\ref{sec:KDS} for details.



\paragraph{Stable Delaunay edges.}
Our study of Voronoi diagrams under a convex polygonal distance
function, to a large extent, is motivated by the notion of
\emph{stable} Delaunay edges, introduced by the authors in a
companion paper~\cite{Stable}, and defined as follows: Let $pq$ be a
Delaunay edge under the Euclidean norm, and let $\triangle pqr^+$
and $\triangle pqr^-$ be the two Delaunay triangles adjacent to
$pq$. For a fixed parameter $\alpha>0$, $pq$ is called an
\emph{$\alpha$-stable (Euclidean) Delaunay edge} if its opposite
angles in these triangles satisfy $\angle pr^+q + \angle pr^-q \leq
\pi-\alpha$. An equivalent and more useful definition, in terms of
the Voronoi diagram, is that $pq$ is $\alpha$-stable if \textit{the
equal angles at which $p$ and $q$ see their common (Euclidean)
Voronoi edge $e_{pq}$ are at least $\alpha$ each.} It is shown in
\cite{Stable} that if $pq$ is $\alpha$-stable in the Euclidean
Delaunay triangulation, then it also appears, and at least
$\alpha/8$-stable, in the $Q$-Delaunay triangulation $\DT_Q(P)$ for
any  shape $Q$ that is sufficiently close to (in terms of its
Hausdorff distance from) the unit disk. The results in this paper,
along with the aforementioned result, imply that by maintaining
$\DT_Q(P)$, where $Q$ is a regular (convex) $k$-gon, for
$k=\Theta(1/\alpha)$, we can maintain (a superset of) the stable
edges of the Euclidean Delaunay triangulation, as a subgraph of
$\DT_Q(P)$, and that we have to handle only a nearly quadratic
number of topological changes if the motion of the points of $P$ is
algebraic of degree bounded by a constant. See~\cite{Stable} for
details.


\section{The topology of $\VD(P)$}
\label{sec:static}

In this section we state and prove a few geometric and topological
properties of the  $Q$-Voronoi diagram of a set of stationary points when $Q$ is a convex polygon.

\paragraph{Some notations.}
Let $Q$ be a convex $k$-gon with vertices
$v_0,\ldots,v_{k-1}$ in clockwise order, whose interior contains the
origin. For each $0\leq i< k$, let $e_i$ denote the edge
$v_iv_{i+1}$ of $Q$, where index addition is modulo $k$ (so
$v_k=v_0$). We refer to the origin as the \emph{center} of $Q$ and
denote it by $o$. A {\em homothetic placement} (or {\em placement}
for short) $Q'$ of $Q$ is represented by a pair $(p,\lambda)$, with
$p\in\reals^2$ and $\lambda\in\reals^+$, so that $Q' = p + \lambda
Q$; $p$ is the location of the center of $Q'$, and $\lambda$ is the
\emph{scaling factor} of $Q'$ (about its center).  The homothets of
$Q$ thus have three degrees of freedom.

There is an
obvious bijection between the edges (and vertices) of $Q'$ and of $Q$,
so, with a slight abuse of notation, we will not
distinguish between them and use the same notation to refer to an
edge or vertex of $Q$ and to the corresponding edge or vertex of
$Q'$.  For a point $u\in\reals^2$, let $Q[u]$ denote the homothetic
copy of $Q$ centered at $u$ such that its boundary touches the
$d_Q$-nearest neighbor(s) of $u$ in $P$, i.e., $Q[u]$ is represented
by the pair $(u,\lambda)$ where $\lambda=\min_{p\in P} d_Q(u,p)$. In
other words, $Q[u]$ is the largest homothetic copy of $Q$ that is
centered at $u$ whose interior is $P$-empty.

\paragraph{$Q$-general position.}
To simplify the presentation, we assume our point set
$P$ to be in \emph{general position} with respect to the underlying polygon
$Q$. Specifically, this means that
\begin{itemize}
\item[(Q1)] no pair of points of $P$ lie on a line parallel to a boundary edge or a diagonal of $Q$,
\item[(Q2)] no four points of $P$ lie on the boundary of the same homothetic copy $Q'$ of $Q$, and
\item[(Q3)] if some three points in $P$ lie on the boundary of the same homothetic
copy $Q'$ of $Q$, then each of them is incident to a \emph{relatively
open edge} of $\partial Q'$ (and all the three edges are distinct, due
to (Q1)), as opposed to one or more of these points touching a vertex
of $Q'$.
\end{itemize}
The above conditions can be enforced by
an infinitesimally small rotation of $Q$ or of $P$.

\paragraph{Bisectors, corner placements, and edgelets.}

The {\em bisector} between two points $p$ and $q$, with respect to
the distance function $d_Q$ induced by $Q$, denoted by
$\bisect_{pq}$ or $\bisect_{qp}$, is the set of all
points $x\in \reals^2$ that satisfy $d_Q(x,p)=d_{Q}(x,q)$.
Equivalently, $\bisect_{pq}$ is the locus of the centers of
all  homothetic placements $Q'$ of $Q$ that touch $p$ and $q$ on
their boundaries; $Q'$ does not have to be $P$-empty, so it
may contain additional points of $P\setminus \{p,q\}$.
If $p$ and $q$ are not parallel to an edge of $Q$ (assumption (Q1)), then
$b_{pq}$ is a one-dimensional polygonal curve, whose structure will
be described in detail momentarily.

A homothetic placement $Q'$ centered along $\bisect_{pq}$ that
touches one of $p$ and $q$, say, $p$, at a vertex, and touches $q$ at
the relative interior of an edge (as must be the case in general
position) is called a \emph{corner placement}
 \emph{at} $p$; see
Figure~\ref{Fig:CornerEdge}~(a). Note that a corner placement at
which a vertex $v_i$ of (a copy $Q'$ of) $Q$ touches $p$ has the
property that the center $o'$ of $Q'$ lies on the fixed ray
emanating from $p$ in direction $\vec{v_io}$.

\begin{figure}[htbp]
\centering
\begin{tabular}{ccc}
\input{CornerPlacement.pspdftex}& \hspace*{1cm}&
\input{EdgeletStraight.pspdftex}\\
\small (a) &&\small (b)
\end{tabular}
\caption{Possible placements of $Q'$ on $\bisect_{pq}$. (a) A corner placement
at $p$. The center of $Q'$ lies on the fixed ray emanating from $p$
in direction $\vec{v_io}$. (b) A placement of $Q'$ with edge
contacts of $e_i$ and $e_j$ at $p$ and $q$, respectively. (The
centers of) such placements trace an edgelet of $\bisect_{pq}$
with label $(e_i,e_j)$.}
\label{Fig:CornerEdge}
\end{figure}
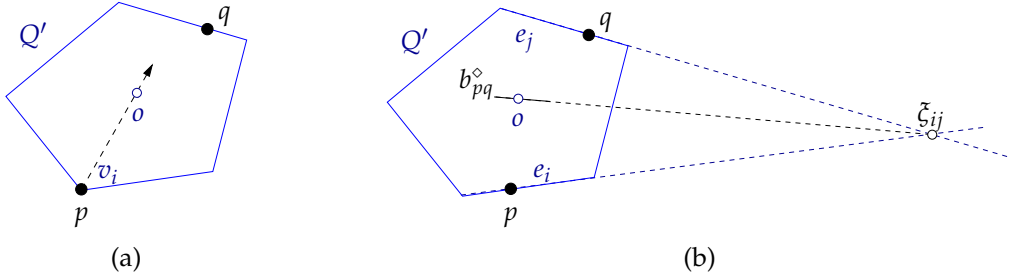

A non-corner placement $Q'$ centered on $\bisect_{pq}$ can be
classified according to the pair of edges of $Q'$, say, $e_i$ and
$e_j$, that touch $p$ and $q$, respectively. We may assume (by (Q1)) that
$e_i\ne e_j$. Slide $Q'$ so
that its center $o'$ moves along $\bisect_{pq}$ and its size
expands or shrinks to keep it touching $p$ and $q$ at the edges
$e_i$ and $e_j$, respectively. If $e_i$ and $e_j$ are parallel, then
the center $o'$ of $Q'$ traces a line segment in the direction
parallel to $e_i$ and $e_j$; otherwise $o'$ traces a segment in the
direction that connects it to the intersection point $\xi_{ij}$ of
the lines containing (the copies on $\bd Q'$ of) $e_i$ and $e_j$.
See Figure \ref{Fig:CornerEdge}~(b) for the latter scenario. We
refer to such a segment $g$ as an \emph{edgelet} of $b_{pq}$,
and label it by the pair $\psi(g)=(e_i,e_j)$ (or by
$(i,j)$ for brevity). The orientation of the edgelet depends only on
the corresponding edges $e_i,e_j$, and is independent of $p$ and
$q$. The structure of $\bisect_{pq}$ is fully determined by
the following proposition, with a fairly straightforward proof that
is omitted from here.

\begin{lemma}\label{Prop:edgelet}
An edgelet $g$ with the label $\psi(g)=(e_i,e_j)$ appears on
$\bisect_{pq}$ if and only if there is an oriented line parallel to $\vec{pq}$
that crosses $\partial Q$ at (the relative interiors of) $e_i$ and $e_j$,
in this order.
\end{lemma}

\begin{figure}[htbp]
\centering
\input{bisect-new1.pspdftex}
\caption{The edgelets of $\bisect_{pq}$. The breakpoints of
$\bisect_{pq}$ correspond to corner placements of $Q$. We have
$C_{pq}=\langle 2,1\rangle$ and $C_{qp}=\langle 3,4,5\rangle$.
The terminal edgelets of $\bisect_{pq}$ are the rays with labels $(1,5)$
and $(2,3)$.}
\label{Fig:BisectLabel}
\end{figure}
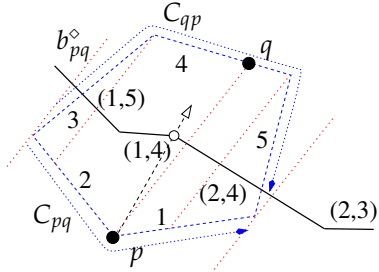

See Figure \ref{Fig:BisectLabel}  for an illustration.  The
endpoints of edgelets are called the \emph{breakpoints} of
$\bisect_{pq}$. Each breakpoint is the center of a corner
placement of $Q$; If $e_i$ and $e_j$ are adjacent, then the edgelet
labeled $(i,j)$ is a ray and the common endpoint of $e_i,e_j$ is one
of the two vertices of $Q$ extremal in the direction orthogonal to
$\vec{pq}$ (i.e., these vertices have a supporting line parallel to
$\vec{pq}$).

Assuming $Q$-general position of $P$, Lemma
\ref{lem:edgelet-labels} below implies that $\bisect_{pq}$ is
the concatenation of exactly $k-1$ edgelets. Let $C_{pq}$ and
$C_{qp}$ denote the two chains of $\partial Q$, delimited by the
vertices that are extremal in the direction orthogonal to
$\vec{pq}$, such that $p$ lies on an edge of $C_{pq}$ and $q$ on an
edge of $C_{qp}$ at all placements of $Q$ touching $p$ and $q$ and
centered along the bisector $\bisect_{pq}$. We orient both
$C_{pq}$, $C_{qp}$ so that they start (resp., terminate) at the
vertex of $Q$ that is furthest to the left (resp., to the right) of
$\vec{pq}$; see Figure \ref{Fig:BisectLabel}.

Our characterization of $\bisect_{pq}$ is completed by the
following lemma, which follows from Lemma \ref{Prop:edgelet}
and the preceding discussion.
\begin{lemma}\label{lem:edgelet-labels}
Let $(e_{11},e_{21}), \ldots, (e_{1h},e_{2h})$ be the sequence of
labels of the edgelets of $\bisect_{pq}$ in their order along
$\bisect_{pq}$ when we trace it so that $p$ lies on its left
side and $q$ on its right side. Then $e_{11},\ldots,e_{1h}$ appear
(with possible repetitions as consecutive elements) in
this order along $C_{pq}$, and $e_{21},\ldots,e_{2h}$ appear
(again, with possible repetitions) this order
along $C_{qp}$. Furthermore, the following additional properties
also hold:
\begin{itemize}
\item[(a)] All the edges of $C_{pq}$ (resp., $C_{qp}$) appear,
possibly with repetitions, in the first (resp., second) sequence.
\item[(b)] The elements of $C_{qp}$ appear in clockwise order and the
elements of $C_{pq}$ in counterclockwise order along $\bd Q$.
\item[(c)] Assuming general position, the passage from a label
$(e_{1i},e_{2i})$ to the next label $(e_{1(i+1)},e_{2(i+1)})$ is
effected by either replacing  $e_{1i}$ by the following edge on
$C_{pq}$ or by replacing $e_{2i}$ by the following edge on $C_{qp}$.
In the former (resp., latter)  case, the common endpoint of the two edgelets corresponds to the corner placement of $p$ (resp., $q$) at the common vertex of
$e_{1i}$ and $e_{1(i+1)}$ (resp., $e_{2i}$ and $e_{2(i+1)}$).
\end{itemize}
\end{lemma}

The proof of the lemma, whose details are omitted, proceeds by sweeping
a line parallel to $\vec{pq}$, and keeping track of the pairs of edges of $Q$
that are crossed by the line, mapping each position of the line to a homothetic
placement of $Q$ that touches $p$ and $q$ at the images of the two intersection
points.

For $0 \le i < j < k$, let $\theta_{ij}$ be the orientation of the line
passing through the vertices $v_i$ and $v_j$ of $Q$, and let $\Theta$ be the set of these orientations. $\Theta$ partitions the unit circle
$\sphere^1$ into a collection $\I$ of $O(k^2)$ angular intervals (for a regular $k$-gon,
the number of intervals is only $\Theta(k)$).
Lemmas~\ref{Prop:edgelet} and~\ref{lem:edgelet-labels} implies the following corollary:
\begin{cor}
\label{cor:bisect-seq}
The sequence of edgelet labels along $\bisect_{pq}$ is the same for all the
ordered pairs of points $p, q$ such that the orientation of the vector $pq$ lies
in the same interval of $\I$.
\end{cor}

The following additional property of bisectors is
crucial for understanding the topological structure of $\VD(P)$.

\begin{lemma}
\lemlab{edge-connect}
Let $p, q_1, q_2$ be three distinct points of $P$.
The bisectors $b_{pq_1}$, $b_{pq_2}$ can intersect
at most once, assuming that $p,q_1$ and $q_2$ are in $Q$-general position.
\end{lemma}

\begin{proof}
Suppose to the contrary that  $b_{pq_1}$, $b_{pq_2}$ intersect
at two points. Then there exist two homothetic copies $Q_1$ and $Q_2$
of $Q$ such that $p, q_1, q_2 \in \partial Q_1 \cap \partial Q_2$.
However, it is well known
that homothetic placements of $Q$ behave like pseudo-disks, in the
sense that the portion of the boundary of each of them outside the
other homothetic placement is connected; see, e.g., \cite{KLPS}.
Therefore, $\partial Q_1$ and $\partial Q_2$  intersect in at most
two connected portions, each of which is either a point or  a
segment parallel to some edge of $Q$. Clearly, one of these
connected components of $\partial Q_1 \cap \partial Q_2$  must
contains two out of the three points $p$, $q_1$, and $q_2$, in
contradiction to the fact that the points are in $Q$-general
position.
\end{proof}

The following lemma provides additional details concerning the
structure of the breakpoints of the bisectors in case $Q$ is a regular $k$-gon.

\begin{lemma} \lemlab{interleave}
Let  $Q$ be a regular $k$-gon, and let $p$ and $q$ be two points in
general position with respect to $Q$. The breakpoints along the
bisector $\bisect_{pq}$ correspond  alternatingly  to corner
placements at $p$ and corner placements at $q$.
\end{lemma}

\begin{figure}[htbp]
\centering
\input{bisect.pspdftex}
\caption{The bisector $\bisect_{pq}$ for a regular octagon
$Q$; it has seven edgelets and the centers of the corner placements
along $\bisect_{pq}$ alternate between $p$ (hollow circles) and $q$
(filled circles).}
\label{Fig:bisect}
\end{figure}
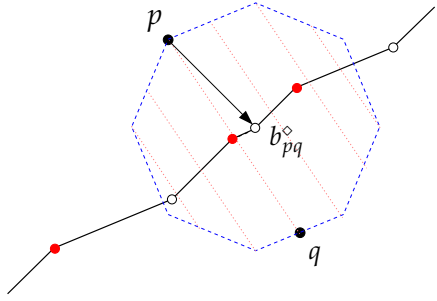

\begin{proof}
Refer to Figure~\ref{Fig:bisect}.
Suppose that two consecutive breakpoints of $\bisect_{pq}$
correspond to corner placements at $p$. From
Lemmas~\ref{Prop:edgelet} and~\ref{lem:edgelet-labels}, we
obtain that these corner placements are formed by two adjacent
vertices, say $v_0$ and $v_1$ of $Q$, and $q$ lies in the relative
interior of (the homothetic copies of) the same edge $e$ of $Q$ at
these placements. This implies that the projections of $v_0$ and
$v_1$ in direction $\vec{pq}$ lie in the interior of the projection
of the edge $e$ of $Q$ in direction $\vec{pq}$, which is impossible
if $Q$ is a regular $k$-gon. Indeed, the convex hull of $e_0=v_0v_1$
and of $e$ is an isosceles trapezoid $\tau$, which
implies that, for any other strip $\sigma$ bounded by two parallel
lines through $v_0$ and $v_1$, $e$ cannot cross both boundary lines
of $\sigma$. We note that if $\sigma$ is the strip
spanned by $\tau$, then $e$ touches both lines bounding $\sigma$ but
does not cross any of them. This completes the proof of the lemma.
\end{proof}

\paragraph{Voronoi cells, edges, and vertices.}
Each bisector $\bisect_{pq}$ partitions the plane into open
regions $H_{pq}=\{x\mid d_Q(p,x)<d_Q(q,x)\}$ and
$H_{qp}=\{x\mid d_Q(q,x)<d_Q(p,x)\}$. Hence, for each point
$p\in P$, its $Q$-Voronoi cell $\Vor(p)$ can be described as
$\bigcap_{q\in P\setminus \{p\}} H_{pq}$.

By $Q$-general position of $P$, for any $p \in P$,
$\bd\Vor(p)$ is composed of \emph{$Q$-Voronoi edges},
where each such edge is a maximal connected portion of the bisector
$\bisect_{pq}$, for some other point $q\in P$, that lies within
$\bd \Vor(p)\cap \bd \Vor(q)$.
The portion of $\bisect_{pq}$ within this common boundary
can be described as
$$
\bisect_{pq}\cap \bigcap_{r\neq p,q}H_{pr} =
\bisect_{pq}\cap \bigcap_{r\neq p,q}H_{qr}.
$$
That is, this portion is the locus of all centers $x$ of placements
 of $Q$ for which the equal distances
$d_Q(x,p)=d_Q(x,q)$ are the smallest among the distances from $x$ to
the points of $P$. Note that the homothetic copy $x+d_Q(x,p) Q$ of $Q$ touches $p$ and $q$ and is $P$-empty.

Since $P$ is in $Q$-general position,
\lemref{edge-connect} guarantees that this portion of $\bisect_{pq}$ is either
connected or empty. Therefore, any bisector $\bisect_{pq}$
contains \emph{at most one} $Q$-Voronoi edge, which we denote by $e_{pq}$.
This edge is called a \emph{corner edge} if it contains
a breakpoint (i.e., a center of a corner placement); otherwise it
is a \emph{non-corner edge}---a line segment.

The endpoints of $Q$-Voronoi edges $e_{pq}$ are called
\emph{$Q$-Voronoi vertices}. By the $Q$-general position of $P$,
each such vertex is incident on exactly three
Voronoi cells $\Vor(p)$, $\Vor(q)$, and $\Vor(r)$.
This vertex, denoted by $\nu_{pqr}$, can be described as the center of the unique
homothetic $P$-empty placement $Q'=Q[\nu_{pqr}]$ of $Q$, whose boundary contains only
the three points $p$, $q$, and $r$ of $P$.
>From the Delaunay point
of view, $\DT(P)$ contains the triangle $\triangle pqr$.

We say that $g$ is an \emph{edgelet of} $e_{pq}$ if (i) $g$ is
an edgelet of $\bisect_{pq}$, and (ii) the Voronoi edge
$e_{pq}$ either contains or, at least, overlaps $g$.
We refer to an edgelet $g$ of $e_{pq}$ as \emph{external}
if it contains one of the endpoints of $e_{pq}$, namely,
a vertex of $\VD(P)$, and as \emph{internal} otherwise.
In general position, an external edgelet of $e_{pq}$ is always
properly contained in an edgelet of $\bisect_{pq}$.

We conclude this section by making the following remarks. If assumption (Q3)
does not hold, then a Voronoi vertex may coincide with a breakpoint
of an edge adjacent to it; if (Q2) does not hold then a Voronoi
vertex may have degree larger than three; if the segment $pq$ connecting a pair
$p,q \in P$ is parallel to a diagonal of $Q$, then an edgelet of a Voronoi edge may
degenerate to a single point; and if such a segment $pq$ is parallel to an edge of
$Q$ then $\bisect_{pq}$ may be a two-dimensional region
(Figure~\ref{Fig:DegenBisect} middle).  These degenerate configurations are
discussed in detail in the next section.



\section{Kinetic Voronoi and Delaunay diagrams}
\label{sec:kinetic}

As the points of $P$ move along continuous trajectories, $\VD(P)$ also changes continuously,
namely, vertices of $\VD(P)$ and breakpoints of edgelets trace continuous trajectories, but,
unless the motion is very degenerate, the topological structure of $\VD(P)$ changes only
at discrete times, at which an edgelet in a Voronoi edge appears/disappears, a Voronoi vertex
moves from one edgelet to another, or two adjacent Voronoi cells cease to be adjacent or
vice versa (equivalently, an edge appears or disappears in $\DT(P)$), because of a
$Q$-cocircularity of four points of $P$.
In this section we discuss when do these changes occur and how does
$\VD(P)$ change at such instances. To simplify the presentation, (i) we
assume that the orientations of the edges and diagonals $v_iv_j$,
for all pairs of vertices $v_i, v_j$ of $Q$,
are distinct, and that they are different from those of $ov_i$, for any vertex
$v_i$; (ii) we make certain general-position assumptions on the trajectories of
$P$; and (iii) we augment $P$ with some points at infinity.
At the end of the section, we remark what happens if we do not make
these assumptions or do not augment $P$ in this manner.

\paragraph{Augmenting $P$.}
We add points to $P$ so that the convex
hull of the augmented set does not change as the (original) points
move, and the boundary of $\DT(P)$ is this stationary convex hull at all times.
Specifically, for each vertex $v_i$ of $Q$, we add a corresponding
point $q_i$ at infinity, so that $q_i$ lies in the direction
$\vec{v_io}$. Let $P_\infty$ denote the set of these $k$ new points.
We maintain $\VD(P\cup P_\infty)$ and $\DT(P\cup
P_\infty)$. It can be checked that $\DT(P\cup P_\infty)$
contains all edges of $\DT(P)$, some ``unbounded'' edges
(connecting points of $P$ to points of $P_\infty$), and $k$ edges at
infinity (forming the convex hull of $P\cup P_\infty$). Furthermore,
every edge of $\DT(P\cup P_\infty)$ incident on at least one
point of $P$ is adjacent to two triangles; only the edges at
infinity are ``boundary'' edges of the triangulation. During the
motion of the points of $P$, the points of $P_\infty$ remain
stationary.

Let $\triangle$ be a triangle of $\DT(P\cup P_\infty)$. There
is a ($P\cup P_\infty$)-empty homothetic copy
$Q_\triangle$ of $Q$ associated with $\triangle$, whose boundary
touches the three vertices of $\triangle$. If two vertices of
$\triangle$ belong to $P$ and one vertex of $\triangle$ is a point $q_i$
at infinity, then $Q_\triangle$ is a wedge formed by the two
corresponding consecutive edges $e_{i-1}$ and $e_i$ of $Q$, each
touching a vertex of $\triangle$ not in $P_\infty$ (e.g., $C_1$ in
Figure~\ref{Fig:extended-DT}).  If only one vertex of $\triangle$,
say $p$, belongs to $P$, then $\triangle$ is of the form $\triangle
pq_iq_{i+1}$ (for some $0\leq i< k$), and there are arbitrarily
large empty homothetic copies of $Q$ incident on $p$ at the edge
$e_i=v_iv_{i+1}$ (e.g., $C_2$ in Figure~\ref{Fig:extended-DT}). The
number of triangles of the latter kind is only $k$, one for each
edge of $Q$. Abusing the notation slightly, we will use $P$ to
denote $P\cup P_\infty$ from now on.

\begin{figure}[htbp]
\centering
\input{extended-DT.pspdftex}
\caption{The extended $\DT(P)$ under the $L_\infty$-metric
(where $Q$ is an axis-parallel square)
for the set of points in Figure~\protect\ref{fig:VDT}, with four points
$q_1, \ldots, q_4$ added at infinity in directions $(\pm 1, \pm 1)$.
Each of the four shaded triangles has two vertices at infinity, and
the unbounded half-strips between them represent triangles with one
vertex at infinity.  The empty wedge $C_1$ corresponds to
$\triangle 12q_1$ (the half-strip right above the left shaded triangle),
and the arbitrarily large empty square $C_2$ (a halfplane in the limit)
corresponds to $\triangle 9q_2q_3$ (the right shaded triangle).}
\label{Fig:extended-DT}
\end{figure}
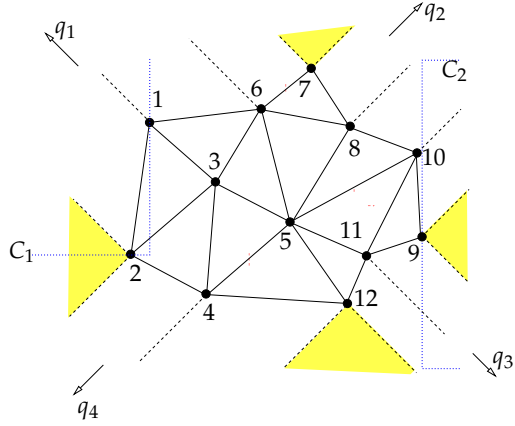

\paragraph{$Q$-general position for trajectories.}
We assume that the trajectories of the points of $P$ are in
\emph{$Q$-general position}, which we define below. Informally, if
the motion of each point of $P$ is algebraic of bounded degree, as we assume,
then the time instances at which degenerate configurations occur, namely,
configurations violating one of the assumptions (Q1)--(Q3), can be represented
as the roots of certain constant-degree polynomials in $t$. The present
``kinetic'' general-position assumption for the trajectories says that none
of these polynomials is identically zero (so each
of them has $O(1)$ roots), that each root has multiplicity one (so the sign of
the polynomial changes in the neighborhood of each root), and that the roots
of all polynomials are distinct.
We now spell out these conditions in more detail and make them geometrically concrete.

\medskip
\noindent\textbf{(T1)}
For any pair of points $p,q\in P$, $p(t) \ne q(t)$ for all $t$,
namely, $p$ and $q$ do not collide during the motion.

\medskip
\noindent\textbf{(T2)}
For any pair of points $p,q\in P$, there exist at most
$O(1)$ times when the segment $pq$ is parallel to any given edge or
diagonal $v_iv_j$ of $Q$, and at each of these times $p$ {\it properly}
crosses the line through $q$ parallel to $v_iv_j$, which moves continuously,
together with $q$ (and $q$ does the same for the parallel line through $p$).

\medskip
\noindent\textbf{(T3)}
For any ordered set of three points $p,q,r\in P$ and
for any vertex $v_p$ and a pair of edges $e_q$, $e_r$ of $Q$, there
exist at most $O(1)$ times when $q$ touches $e_q$ and $r$ touches $e_r$ at
a corner placement $Q'$ of $Q$ at $p$ in which $p$ touches $v_p$.
Furthermore, given that $e_q$ is not adjacent to $v_p$, at each such time the point $r$ either enters or leaves the interior 
of the unique $P$-empty homothetic copy of $Q$ that touches $p$ at $v_p$ and $q$ at $e_q$.


\medskip

\noindent\textbf{(T4)}
For any four points $p,q,a,b\in P$ and  any ordered quadruple $e_p$,
$e_q$, $e_a$, $e_b$ of  edges of $Q$, such that at least three of
these edges are distinct, there are only $O(1)$ times at which there
exists a placement $Q'$ of $Q$ such that $p,q,a,b$ touch the respective
relative interiors of  $e_p$, $e_q$, $e_a$, $e_b$. We say that
$p,q,a,b$ are \emph{$Q$-cocircular} at these $O(1)$ times.
At any such $Q$-cocircularity, the four points $p,q,a,b$
are partitioned into two pairs, say, $(p,q)$ and $(a,b)$, so that
right before the cocircularity there exists a homothetic copy of $Q$ that
is disjoint from $a$ and $b$ and whose boundary touches $p$ and $q$,
and right afterwards there exists a homothetic copy of $Q$ that is
disjoint from $p$ and $q$ and whose boundary touches $a$ and $b$.
\medskip

\noindent\textbf{(T5)}
Events of type (T2)--(T4) do not occur simultaneously, except when
two points $p$ and $q$ become parallel to an edge of $Q$.
In this case there could be many events of type (T3) and
(T4) that occur simultaneously, each of which involves $p, q$;
see below for more details.

\paragraph{Events.}
Since the motion of $P$ is continuous, the topological changes in $\VD(P)$ occur
only when some points of $P$ are involved in a degenerate configuration, i.e.,
they violate one of the assumptions (Q1)--(Q3). However not every degenerate
configuration causes a change in $\VD(P)$. We define an \emph{event} to be the
occurrence of a $P$-empty placement of a homothetic copy of $Q$ whose boundary
contains two, three, or four points of $P$ that are in a $Q$-degenerate configuration.
The center of such a placement lies on an edge or at a vertex of $\VD(P)$.
The subset of points involved in the degenerate configuration is referred to as
the subset \emph{involved in the event}.
The event is called a \emph{bisector, corner}, or \emph{flip} event
if assumption (Q1), (Q2), or (Q3), respectively, is violated.

An event is called \emph{\degenerate} if some pair among the (constantly
many) points involved in the event span a line parallel to an edge of
$Q$. Otherwise, we say that the event is \emph{\generic}.
The $Q$-general-position assumption on the trajectories of the points of $P$ implies,
in particular, that (i) no \generic event can occur simultaneously with any other event, and
(ii) all \degenerate events that occur at a given time, must involve the
\emph{same} pair of points $p,q$ that span a line parallel to an edge of $Q$.

The changes in $\VD(P)$ are simple and local at a \generic event, but $\VD(P)$ can undergo
a major change at a \degenerate event. We therefore first discuss the changes at
a \generic event and then discuss \degenerate events.

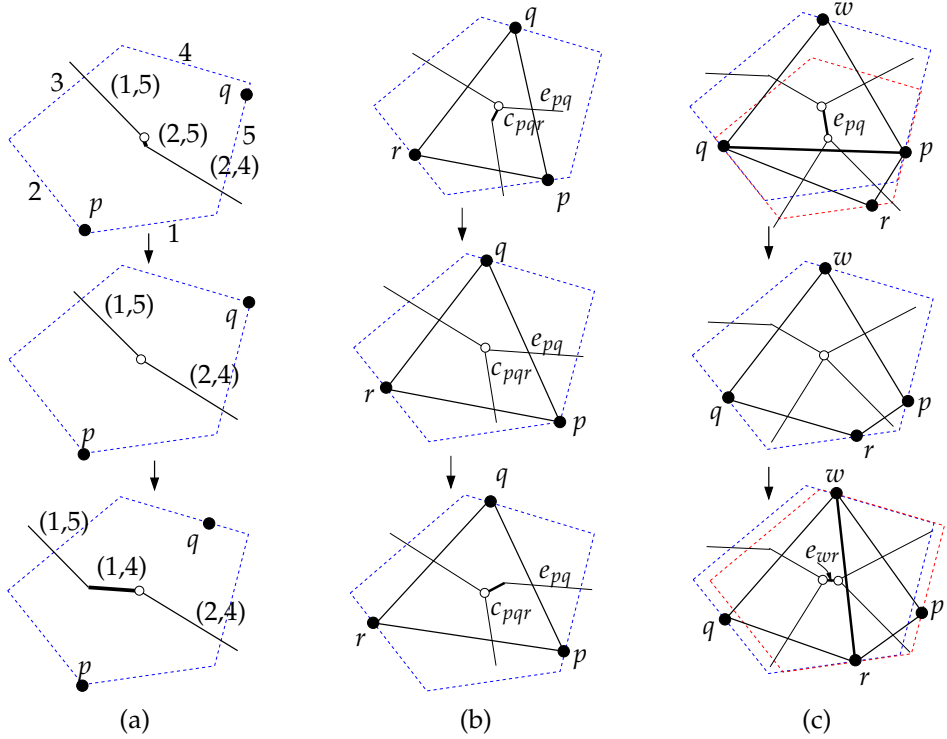
\begin{figure}[htb]
\centering
\begin{tabular}{ccccc}
\input{bisector-event.pspdftex}&\hspace*{3mm}&
\input{corner-event.pspdftex}&\hspace*{3mm}&
\input{generic-flip.pspdftex}\\
\small (a)&& \small (b)&& \small (c)
\end{tabular}
\caption{Different types of \generic events for a pentagon $Q$:
(a) a bisector event, (b) a corner event, and (c) a flip event.
Thick segments denote an appropriate portion of $\DT(P)$ and
the elements of $\VD(P)$ that change at the event.}
\label{Fig:events}
\end{figure}

\subsection{\Generic events}
\label{Subsec:Generic}

Recall that the orientation of $pq$, for every pair $p,q \in P$,  at a \generic event is different from that of any edge of $Q$, which implies that no two
points of $P$ lie on the same edge of  a homothetic copy of $Q$ at a \generic event.

\paragraph{Bisector event.}
A pair of points $p, q \in P$ are incident to the vertices
 $v_i$ and $v_j$, respectively, of a $P$-empty homothetic copy of $Q$ so that the
vertices $v_i$ and $v_j$ are \emph{not consecutive} along $\bd Q$.
In particular, $e_{pq}$ is an edge in $\VD(P)$.

Recall that in our notation, $v_i$ is adjacent to the consecutive
edges $e_{i-1},e_i$, and $v_j$ is adjacent to the consecutive edges $e_{j-1},e_{j}$; in
the present scenario, these four edges are all distinct. Without
loss of generality, we may assume that before the event,
there is an oriented line parallel to $\vec{pq}$ that intersects
$e_{i-1}$ and $e_{j-1}$ in this order, and there is no such line after the event.
Similarly, after the event there is an oriented line parallel to $\vec{pq}$ that intersects
$e_{i}$ and $e_{j}$ in this order, and there is no such line before the event.
 Hence, Lemma \ref{Prop:edgelet}
and assumption (T2) imply that $e_{pq}$ loses a bounded edgelet with label
$(i-1,j-1)$, which is replaced by a new bounded edgelet with label $(i,j)$.
Our assumption (T5) implies
$(i-1,j-1)$ cannot be an external
edgelet of $e_{pq}$. Hence, $(i-1,j-1)$ appears shortly before the event as an internal edgelet of
$e_{pq}$, shrinks to a point and is replaced by
the new internal edgelet $(i,j)$; see Figure~\ref{Fig:events}~(a).
This is the only topological change in $\VD(P)$ at this event.

Notice that whenever the direction of $\vec{pq}$ coincides with that of a diagonal $v_iv_j$ of $Q$, $p$ and $q$ are incident to the vertices $v_i$ and $v_j$, respectively, of a unique copy of $Q$.
If this copy contains further points of $P$, then the bisector $\bisect_{pq}$ still loses an edgelet
$(i-1,j-1)$, which is replaced by a new edgelet $(i,j)$. However, both of these edgelets now belong to the portion of $\bisect_{pq}$ outside $e_{pq}$, so the discrete structure of $\VD(P)$ does not change (and, therefore, no bisector event is recorded).


\paragraph{Corner event.}
A corner event occurs when there is a $P$-empty homothetic copy
$Q'=Q[\nu_{pqr}]$ of $Q$ with a corner placement of a vertex of $v\in Q$ at
$p$ and two other points $q$ and $r$ lie on two
distinct edges of $Q$, none of which is incident to $v$.
We refer to such an event as a \emph{\generic corner event of $p$}.

This event corresponds to a vertex
$\nu_{pqr}$ of $\VD(P)$, an endpoint of an edge $e_{pq}$,
coinciding with a breakpoint of $\bisect_{pq}$.
Then $\nu_{pqr}$, also an endpoint of the Voronoi edge
$e_{pr}$, coincides with a breakpoint of $\bisect_{pr}$
as well.  By assumption (T3), one of the $Q$-Voronoi edges $e_{pq}$ and $e_{pr}$ gains
a new edgelet and the other loses an edgelet at this event;
see Figure~\ref{Fig:events}~(b).

\paragraph{Flip event.}
A flip event occurs when there is a  $P$-empty homothetic copy $Q'$ of
$Q$ that touches four points $p, p', q, q'$ at four distinct edges of $Q'$,
in this circular order along $\partial Q'$.
By assumption (T4), up to a cyclic relabeling of the points, the
Voronoi edge $e_{pq}$ \emph{flips} to a new Voronoi edge
$e_{p'q'}$ at this event; see Figure~\ref{Fig:events}~(c).
Note that $e_{pq}$ (resp., $e_{p'q'}$) is
a non-corner edge immediately before (resp., after) the flip event, as
both the vanishing edge $e_{pq}$ and the newly emerging edge
$e_{p'q'}$ are ``too short'' to have breakpoints near the event
(this is a consequence of the kinetic $Q$-general position
assumptions).

\medskip

This completes the description of the changes in $\VD(P)$ at a \generic event.
We remark that a Voronoi edge newly appears or disappears only at a flip event, so, by
definition, $\DT(P)$ changes only at a flip event.  Suppose the Voronoi edge
$e_{pq}$ flips to the edge $e_{p'q'}$ at a flip event.  Then $p,q,p',q'$ are
vertices of two adjacent triangles $\triangle pqp'$ and $\triangle pqq'$ immediately
before the event, the edge $pq$ of $\DT(P)$ flips to $p'q'$ at the event, and the
Delaunay triangles $\triangle pqp', \triangle pqq'$ flip to
$\triangle p'q'p, \triangle p'q'q$ at the event; again, see Figure~\ref{Fig:events}~(c).

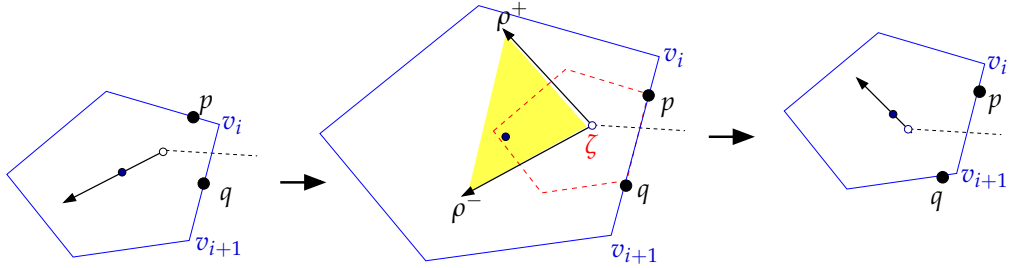
\begin{figure}[htbp]
\centering
\input{degenerate-bisector-new.pspdftex}
\caption{A \degenerate bisector event. The ray $\rho^-$ of $\bisect_{pq}$
 is replaced instantly by the ray $\rho^+$,
and the entire shaded wedge is part of $\bisect_{pq}$ at the
event itself.} \label{Fig:DegenBisect}
\end{figure}

\subsection{\Degenerate events}\label{Subsec:Degenerate}

Recall that a \degenerate event occurs, at time $t_0$, if
two points $p,q\in P$ lie on an edge $e_i=v_iv_{i+1}$ of a $P$-empty
homothetic copy of $Q$. Hence, a {\it \degenerate bisector event}
(involving $p$ and $q$) occurs at $t_0$.
We may assume that neither $p$ nor $q$ is
in $P_\infty$ since in such a case the
orientation of $pq$ remains fixed throughout the motion (namely, it is
$\vec{ov_i}$) and, as we have
assumed, different from the orientations of the edges of $Q$.

\paragraph{Changes in bisectors.}
Assume that $pq$ becomes parallel to the edge $e_i$, and, without loss of generality,
assume that $\vec{pq}$ and $\vec{v_iv_{i+1}}$ have the same orientation, as in
Figure~\ref{Fig:DegenBisect} (center). When this occurs,
the set of placements $Q'$ of $Q$ at which both $p$
and $q$ touch $e_i$ is a wedge $W_0$ whose boundary rays $\rho^-$ and
$\rho^+$ have respective directions $\vec{v_io}$ and
$\vec{v_{i+1}o}$, and whose apex $\zeta$ corresponds to the
placement at which $p$ and $q$ touch $v_i$ and $v_{i+1}$
respectively; see Figure~\ref{Fig:DegenBisect}~(center).

Let $t_0^-$ (resp., $t_0^+$) denote an instance of time immediately
before (resp., after) $t_0$, so that no event occurs in the interval
$[t_0^-,t_0)$ (resp., $(t_0,t_0^+]$). Then the terminal ray of $\bisect_{pq}$
that becomes the wedge $W_0$ at time $t_0$ is either in direction $\vec{v_io}$
or $\vec{v_{i+1}o}$ at time $t_0^-$. Without loss of generality, throughout the
present discussion of the \degenerate event, we assume that this ray is in the direction
$\vec{v_io}$, i.e., it consists of all  placements with $e_i$ touching
$q$ and $e_{i-1}$, the other edge adjacent to $v_i$, touching
$p$. This ray is parallel to $\rho^-$ and approaches $\rho^-$  as
time approaches $t_0$; see Figure~\ref{Fig:DegenBisect}~(left).
By assumption (T2), the bisector $\bisect_{pq}$
at time $t_0^+$ contains a terminal ray parallel to $\rho^+$,
which consists of all
placements with $e_i$ touching $p$, and with the other edge
$e_{i+1}$ adjacent to $v_{i+1}$ touching $q$. At time $t_0$ this ray
coincides with $\rho^+$, which is clearly different from $\rho^-$.
See Figure~\ref{Fig:DegenBisect}~(right). That is, the terminal ray
of $\bisect_{pq}$ \emph{instantly} switches from $\rho^-$ to
$\rho^+$ at time $t_0$.




\paragraph{Changes in $\VD(P)$.}
All topological changes in $\VD(P)$ at the
time $t_0$ of a \degenerate event occur on the boundaries of
the $Q$-Voronoi cells $\Vor(p)$ and $\Vor(q)$. Since $p$
and $q$ are not in $P_\infty$, both of these cells are
\emph{bounded}.

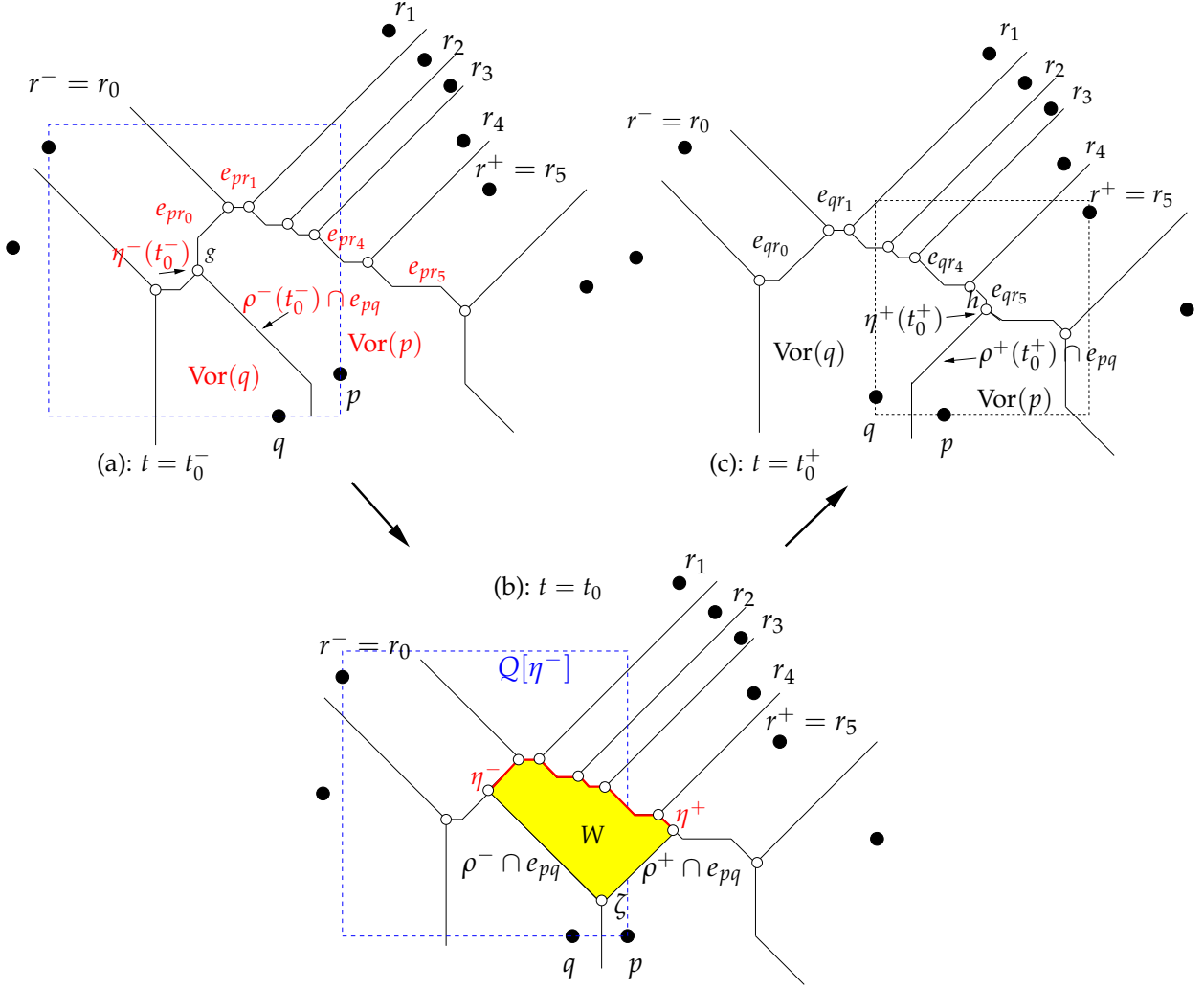
\begin{figure}[htb]
\centering
\input{degenerateVorBefore.pspdftex} \hspace{0.3cm}
\input{degenerateVorAfter.pspdftex}\\
\input{arrows.pspdftex}\\
\input{degenerateVort0.pspdftex}
\caption{Changes in $\VD(P)$ at the time $t_0$ of a \degenerate
bisector event. (a): $\VD(P(t_0^-))$; (b):$\VD(P(t_0))$:
the terminal ray of $\bisect_{pq}$ instantly switches
from $\rho^-$ to $\rho^+$ and the cell $\Vor(p)$ loses the
entire region $W$ to $\Vor(q)$; (c): $\VD(P(t_0^+))$.}
\label{Fig:Change-degenerate-VD}
\end{figure}

Refer to the state of $\VD(P)$ at times $t_0^-$, $t_0$, and
$t_0^+$, as illustrated in Figure \ref{Fig:Change-degenerate-VD}.
For $t\in [t_0^-,t_0)$, let $\rho^-(t)$ be the edgelet of
the bisector $b_{pq}$ that is parallel to $\rho^-$ at time $t$, and
let $\eta^-(t)$ be the
Voronoi vertex which is incident to $e_{pq}\cap \rho^-(t)$ and
to some other pair of edges $e_{pr^-}$ and $e_{qr^-}$.
Similarly, for $t\in (t_0,t_0^+]$, let $\rho^+(t)$ be the edgelet
of the bisector $b_{pq}$ that is parallel to $\rho^+$ at time $t$,
and let $\eta^+(t)$ be the Voronoi vertex which is incident to
$e_{pq}\cap \rho^+(t)$ and to some other pair of edges
$e_{pr^+}$ and $e_{qr^+}$.

Let $\eta^-$ (resp., $\eta^+$) be the limit of
$\eta^-(t)$ (resp., of $\eta^+(t)$) as
$t\uparrow t_0$ (resp., $t\downarrow t_0$).
Alternatively, $\eta^-$ is the center of a $P$-empty corner
placement of $Q$ that touches, at time $t_0$, $p$ at $v_i$, and also
touches $q$ (at $e_i$) and $r^-$. As is easily checked, the edgelets $e_{pr^-}$ and
$e_{qr^-}$ incident on $\eta^-$ are collinear at time $t_0$, but just before that time
they were separated by another short edgelet that has shrunk to a point;
see Figure~\ref{Fig:Change-degenerate-VD} (a).
Similarly, $\eta^+$ is the center of a $P$-empty corner
placement of $Q$ that touches, at time $t_0$, $q$ at $v_{i+1}$, and also touches $p$ (at $e_i$)
and $r^+$. Here too the edgelets of $e_{pr^+}$ and $e_{qr^+}$ incident on $\eta^+$ are collinear
at $t_0$, and get separated from each other by a newly emerging short edgelet;
see Figure~\ref{Fig:Change-degenerate-VD} (c).

Let $\gamma$ be the polygonal chain connecting $\eta^-$ and
$\eta^+$ at time $t_0$, consisting of all centers of placements
touching $p$ and $q$ (at $e_i$) and some other point $r_i$ of $P$.
At time $t_0$ the degenerate Voronoi edge $e_{pq}$ includes a two-dimensional
star-shaped polygonal region, denoted by $W$, bounded by $\zeta\eta^-$, $\zeta\eta^+$,
and $\gamma$. (It is a portion of the wedge $W_0$ discussed earlier.)

At time $t_0$ the terminal edgelet of $e_{pq}$ instantly
switches from $\zeta\eta^-$ to $\zeta\eta^+$, the Voronoi cell
$\Vor(p)$ loses the entire region $W$ to
$\Vor(q)$. As a result, $\VD(P)$ can experience
$\Omega(n)$ topological changes at time $t_0$ (in addition to
the obvious change of the edgelet structure of $e_{pq}$, as just discussed).

Specifically, let $e_{pr^-}=e_{pr_0},e_{pr_1},\ldots,e_{pr_s}=e_{pr^+}$
be the  edges of $\Vor(p)$ at times $t\in [t_0^-,t_0)$ that
also overlap $\gamma$ at $t_0$ (listed in the order their
appearance along $\partial \Vor(p)$). Right after time $t_0$, the point $\eta^+=\rho^+\cap
e_{pr^+}$ becomes a new Voronoi vertex
instead of $\eta^- = \rho^- \cap e_{pr^-}$.

Assuming $s\geq 1$ (or alternatively, $r^-\neq r^+$), every old
edge $e_{pr_j}$ of $\Vor(P)$, that is completely
contained in $\gamma$ (such edges exist only if $s \ge 2$),
is instantly \emph{relabeled} as the new edge $e_{qr_j}$ of $\Vor(q)$.
The edge $e_{pr^-}$, which was incident at times $t\in
[t_0^-,t_0)$ to $\eta^-(t)$, loses its portion within $\gamma$
to the adjacent edge $e_{qr^-}$. Symmetrically, the old edge
$e_{pr^+}$ of $\Vor(p)$, which is  hit by $\rho^+(t)$,
for $t \in (t_0,t_0^+]$ is split at $\eta^+$ into $e_{qr^+}$
(its portion within $\gamma$) and $e_{pr^+}$ (its
portion outside $\gamma$).

\paragraph{Changes in $\DT(P)$.}
If $s\geq 1$ then each of the old Delaunay edges $pr_j$, for
$j=0,\ldots,s-1$, flips in $\DT(P)$ to the new edge
$qr_{j+1}$; see Figure \ref{Fig:degenerateManyFlips}. These are the
only changes that $\DT(P)$ experiences at time $t_0$.

%

\paragraph{\Degenerate sequences}
While we can regard the overall change in $\VD(P)$ and $\DT(P)$ at a \degenerate event as
a single compound event, it is more convenient for our analysis, for the KDS
implementation presented in Section~\ref{sec:KDS}, and perhaps also for a better
intuitive perception of this change, to consider it as a sequence of individual
separate \degenerate corner and flip events, as we describe next.

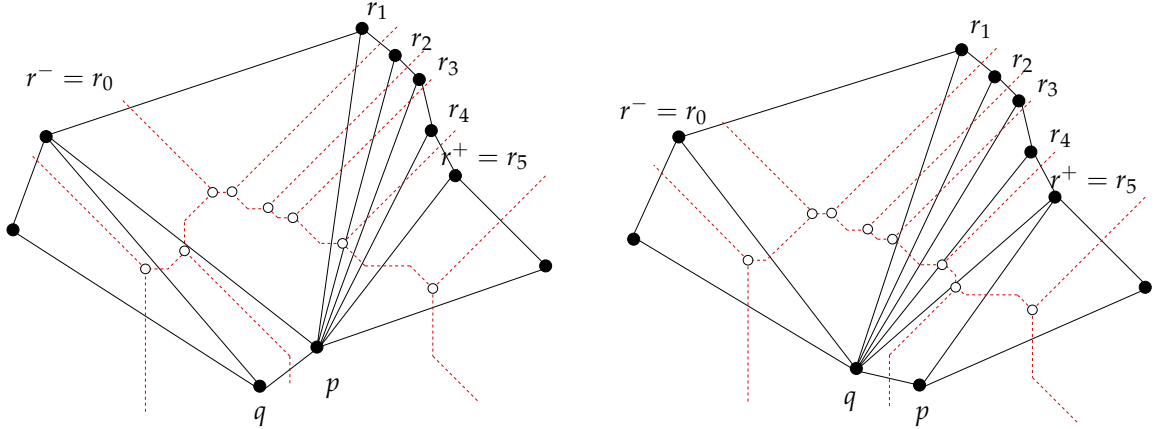
\begin{figure}[htb]
\centering
\input{degenerateDelBefore-new.pspdftex}\hspace{1cm}\input{degenerateDelAfter-new.pspdftex}\\
\caption{Topological changes in $\DT(P)$
at the time $t_0$ of a \degenerate bisector event; five \degenerate flip
events occur at $t_0$, where $pr_0,\ldots,pr_4$ (left) flip to
$qr_1,\ldots,qr_5$ (right).}
\label{Fig:degenerateManyFlips}
\end{figure}

To do so, we ``stretch" the time $t_0$ and
\emph{continuously} rotate the terminal ray $\rho$ of
$\bisect_{pq}$ from $\rho^-$ to $\rho^+$.
Hence, the intersection point
$\eta=\rho\cap \gamma$ traces $\gamma$ from
$\eta^-$ to $\eta^+$; see Figure \ref{Fig:degenerateSweep}.
At any given moment during this virtual rotation, $\rho$ hits some old Voronoi edge
$e_{pr_j}$, for $0\leq j\leq s$, which is split by the current
$\eta$ into $e_{pr_j}$ (the portion not yet swept by $\rho$)
and $e_{qr_j}$ (the swept portion). That is, we can interpret
$\eta$ as an instantaneous $Q$-Voronoi vertex $\nu_{pqr_j}$ which is incident to three Voronoi
edges, namely, $e_{pq}$, $e_{pr_j}$ and $e_{qr_j}$ (since $e_{pq}$ is two-dimensional,
this does not fix the vertex and indeed leaves it with one degree of freedom, of moving along
the appropriate portion of $\gamma$).
The topological structure of $\VD(P)$ changes (during the
above rotation) only at the following two kinds of events, which
closely resemble their \generic counterparts in Section \ref{Subsec:Generic}.

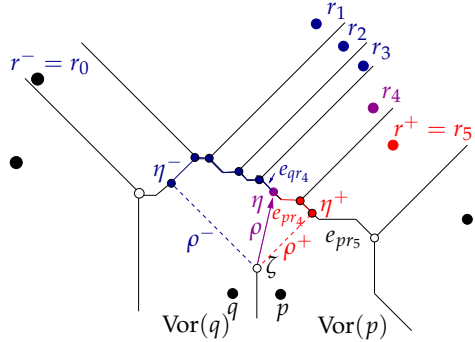
\begin{figure}[htb]
\centering
\input{degenerateSweep.pspdftex}
\caption{``Stretching" the time $t_0$. As the terminal ray $\rho$ of
$\bisect_{pq}$ continuously rotates from $\rho^-$ to $\rho^+$,
the vertex $\eta=\rho\cap \gamma$ traces $\gamma$ from
$\eta^-$ to $\eta^+$. We only show the intersection of the rays
$\rho^-$, $\rho$, and $\rho^+$ with $e_{pq}$. In the depicted
snapshot $\rho$ hits $\gamma$ within the old edge
$e_{pr_4}$, splitting it into $e_{pr_4}$ and
$e_{qr_4}$. Hence, $\eta$ is the Voronoi vertex adjacent to
$\Vor(p)$, $\Vor(q)$, and $\Vor(r_4)$.}
\label{Fig:degenerateSweep}
\end{figure}



\paragraph{\Degenerate corner event.}
This occurs when the Voronoi vertex $\eta=\nu_{pqr_j}$ coincides with a breakpoint
of $\gamma\subseteq \bisect_{pr_j} = \bisect_{qr_j}$, occurring along
the corresponding Voronoi edge.
There are three types of \degenerate corner events, depending on whether the corner placement occurs at $p$, $q$, or $r_j$ for some $0 \le j \le s$.
Each of the first two types occurs just once at a \degenerate
bisector event, whereas the third type may occur multiple times.
\medskip

\noindent {\bf (i)} $\eta=\nu_{pqr_j}$ is  the center of a corner
placement at $p$ along $e_{pr_j}$.
This occurs at $\eta=\eta^-$, i.e., at the starting point of the rotation,
for $Q[\eta^-]$ is indeed a corner placement at $p$.
We refer to this event as the \emph{initial} corner event of the \degenerate sequence.


As already discussed, $Q[\eta^-]$ touches $p$ at the vertex $v_i$, $q$ on the
edge $v_iv_{i+1}$, and the third point $r^-$ at some other edge
$e^-$ (see Figure \ref{Fig:Change-degenerate-VD} (a)---the figure depicts what happens just before
$t_0$; the limiting situation is depicted in Figure \ref{Fig:Change-degenerate-VD} (b)).
Immediately after the \degenerate corner event of $p$ (in the
sense of stretching the time during the virtual rotation of $\rho$), the
$Q$-Voronoi edge $e_{pr^-}$ loses an edgelet (denoted as $g$ in
that figure), namely the edgelet corresponding to $p$ touching
$e_{i-1}$ and $r^-$ touching $e^-$. It is interesting to note that,
unlike at a \generic corner event, $e_{pq}$ does not gain an edgelet.
(The only edgelet that it could have gained is the one that encodes
the double contact of $p$ and $q$ at $e_i$, which is not a real edgelet.)
Note also that the external edgelet of $e_{qr^-}$ becomes aligned with the
next edgelet of $e_{pr^-}$, and as $\eta$ starts rotating, begins to ``annex'' it.
\medskip

\noindent {\bf (ii)} $\eta=\nu_{pqr_j}$ is the center of a corner placement at $q$ along $e_{qr_j}$.
This event occurs only at the end of the rotation,
when $\rho=\rho^+$ and $\eta$ coincides with the Voronoi vertex
$\eta^+=\nu_{pqr^+}$ newly created at this \degenerate event. We refer to this event as the \emph{final} corner event of the \degenerate sequence.

The situation is symmetric to that in (i). Specifically,
$Q[\eta^+]$ touches $q$ at the vertex $v_{i+1}$, $p$ at
the edge $v_iv_{i+1}$, and the third point $r^+$ at some other edge
$e^+$. Immediately after the \degenerate corner event at $q$
(as the ``real'' time $t$ increases past $t_0$), the
$Q$-Voronoi edge $e_{qr^+}$ gains an edgelet (marked by $h$ in
Figure \ref{Fig:Change-degenerate-VD} (c)), but $e_{pq}$
does not lose an edgelet. The external edges of $e_{pr^+}$ and of
$e_{qr^+}$, which were aligned as $\eta$ approaches $\eta^+$,
begin to shift apart from each other, with $h$ in between them.

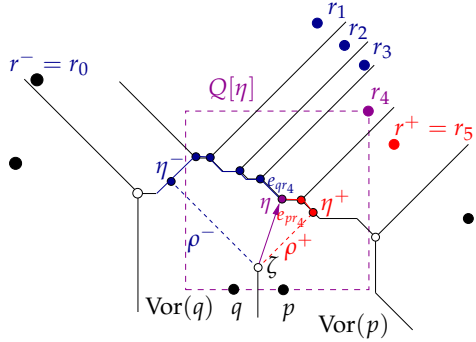
\begin{figure}[htbp]
\centering
\input{degenerateCorner.pspdftex}
\caption{A \degenerate corner event at $r_4$. A corner of $Q[\eta]$
touches $r_4$, so $\eta$ coincides with a breakpoint on
each of the overlapping bisectors $\bisect_{pr_4}$ and
$\bisect_{qr_4}$.}
\label{Fig:degenerate-corner2}
\end{figure}

\noindent {\bf (iii)} $\eta=\nu_{pqr_j}$ coincides with the center of a
corner placement at $r_j$ along the currently traced edge
$e_{pr_j}$ (see Figure~\ref{Fig:degenerate-corner2}).
Let $Q[\eta]=Q[\nu_{pqr_j}]$ be the resulting corner placement of $Q$
at $r_j$, which touches also $p$ and $q$, both on the edge
$e_i=v_iv_{i+1}$ of $Q[\eta]$.
We refer to this event as an \emph{intermediate} corner event of the \degenerate sequence.
Immediately after a \degenerate corner event of $r_j$ (again, in the sense of
the stretched time), the
$Q$-Voronoi edge $e_{pr_j}$ loses an edgelet, and the adjacent
$Q$-Voronoi edge $e_{qr_j}$ gains an edgelet.

\paragraph{\Degenerate flip event.}
The moving vertex $\eta=\nu_{pqr_j}$ coincides with a
$Q$-Voronoi vertex $\nu_{pr_jr_{j+1}}$, which is the common endpoint of the Voronoi edges $e_{pr_j}$ and
$e_{pr_{j+1}}$ along $\gamma$ (thinking of the scenario just before $t_0$; see Figure
\ref{Fig:degenerate-flip}). In the stretched time during
the rotation of $\rho$, the $Q$-Voronoi edge $e_{pr_j}$
shrinks (or, more precisely, ``overtaken'' by $e_{qr_j}$) and disappears from $\VD(P)$. The new edge
$e_{qr_{j+1}}$ is born, as $\eta$ starts moving along the old
edge $e_{pr_{j+1}}$ of $\Vor(p)$, annexing a portion of it for $\Vor(q)$.
The growing portion of that edge between $\nu_{qr_jr_{j+1}}$ and $\eta$ becomes
$e_{qr_{j+1}}$, and $e_{pr_{j+1}}$ is the shrinking
remainder of that edge.  Accordingly, the edge $pr_j$ of
$\DT(P)$ flips to $qr_{j+1}$.

Analogous to a \generic flip event, $e_{pr_j}$ is a non-corner edge when the above flip event occurs. It shrinks to a point at the event,
and the four points $p,q,r_j,r_{j+1}$, which are  $Q$-cocircular, are
vertices of the two adjacent triangles $\triangle pr_jq$ and
$\triangle pr_jr_{j+1}$ of $\DT(P)$ (sharing the edge $pr_j$).
The event is \degenerate because $p$ and $q$ are on the same edge of $Q[\eta]$
(with label $e_i$), and the remaining two points $r_j$ and $r_{j+1}$ are
incident to some pair of other distinct edges of $Q[\eta]$.

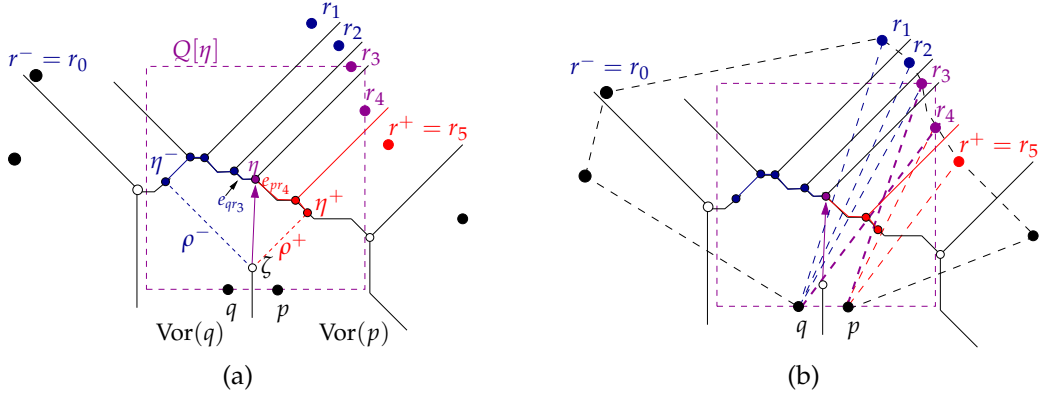
\begin{figure}[htbp]
\centering
\begin{tabular}{ccc}
\input{DegenerateFlip.pspdftex}&
\hspace*{0.5cm}&
\input{degenerateFlipDelaunay-new.pspdftex}\\
\small (a)&&\small (b)
\end{tabular}
\caption{A \degenerate flip event involving $p,q,r_3$ and $r_4$.
(a): The Voronoi perspective. The vertex $\eta$ coincides with the
common endpoint $\nu_{pr_3r_4}$ of $e_{pr_3}$ and
$e_{pr_4}$. The corresponding copy $Q[\eta]$ touches $p,q,r_3$
and $r_4$ in this clockwise order. (b): The Delaunay perspective.
The previously Delaunay edge $pr_3$ flips into $qr_4$.}
\label{Fig:degenerate-flip}
\end{figure}

We remark that, unless $j=0$, the edge $qr_j$
does not belong to $\DT(P)$ at time $t^-_0$. It  becomes an
edge of $\DT(P)$ only after executing the previous \degenerate flip,
which adds it to $\DT(P)$. In other words, the rotational
order in which $\rho$ generates these events is such that each flip
event facilitates the next one, by introducing the appropriate edges
$qr_j$ into the diagram.


\medskip

We refer to the entire sequence of events that are triggered by a
\degenerate bisector event (including the \degenerate bisector event
itself) as a \emph{\degenerate sequence}. The above order, in which
the events of a \degenerate sequence are encountered during the
continuous rotation of the terminal edgelet $\rho$, is called the
\emph{rotational order} of this sequence.

A crucial property of this interpretation of a \degenerate event is that,
in the stretched time, the various \degenerate corner and flip events
occur in the order at which $\eta$ encounters the various breakpoints
and vertices along $\gamma$. At each instance of the rotation, only
one \degenerate event, namely the next one in the sequence, is locally
consistent with the current structure of the diagrams, in the sense
that it can be effected by the appropriate local rule---a flip of
Delaunay and Voronoi edges, or the passage of an edgelet from one
Voronoi edge to another.

\medskip
\noindent\textbf{Remark.}
(i) If we allow $Q$ to have several mutually parallel diagonals, then
several interior edgelets (up to $\Theta(k)$ in the worst case)
of $\bisect_{pq}$ are simultaneously
replaced by new ones when $pq$ becomes parallel to all of them. We treat each of them as a separate bisector event. Similarly, at a \degenerate event, besides the changes discussed above, some interior edgelets of $\bisect_{pq}$ may also change
(if the relevant edge of $Q$ is parallel to some diagonals).

\smallskip
\noindent
(ii) If the trajectories of $P$ are not in general position, as defined above,
we can perform symbolic perturbation, using the framework proposed by
Yap~\cite{Yap}, to simulate general position.

\smallskip
\noindent
(iii) If we do not augment $P$ with $P_\infty$, the set of points at infinity,
some of the edges of $\DT(P)$ are adjacent to only one
triangle, i.e., the edges of the support hull of $P$
(defined in the introduction, cf.~Figure~\ref{fig:VDT}~(b,c)), and this set can change over time
with the motion of $P$. Consequently, besides edge flips, edges need
to be inserted into or deleted from $\DT(P)$. From the Voronoi
diagram perspective, some of the Voronoi cells of $P$ will be
unbounded, and their status from bounded to unbounded will change
over time, as their defining sites become or cease to be vertices of the support
hull of $P$. Consequently, at some \degenerate biscetor events
involving a pair $p, q$, the two-dimensional region of
$\bisect_{pq}$ (the polygonal region $W$ in
Figure~\ref{Fig:Change-degenerate-VD}) may be
unbounded. Therefore extra care is needed while performing the
rotational sweep to process the \degenerate sequence of these events.
Adding $P_\infty$ simplifies and unifies all these special handlings.



\section{Bounding the number of events}
\label{sec:count}

In this section we show that the overall number of events, at which the
topological structure of
$\VD(P)$ changes, during the motion of the points is only
nearly quadratic in $n$. The bounds on the number of events hold
even if the trajectories of points in $P$ are not in $Q$-general
position, but for simplicity, the proofs are presented under the general-position
assumption. We can use a symbolic-perturbation based argument to extend the proof
when trajectories are not in general position; see for
example \cite[Chapter~5]{SA95}.

\paragraph{The number of bisector events.}
The number of \generic (resp., \degenerate) bisector events is
$O(k^2n^2)$ (resp., $O(kn^2)$) since the orientation of the segment
$pq$ for a given pair $p,q\in P$, can become parallel to a fixed
diagonal $v_iv_j$ (resp., a fixed edge $e_i=v_iv_{i+1}$) of $Q$ only
$O(1)$ times, as implied by our kinetic general position assumption.

Each time the orientation of $pq$ becomes parallel to a fixed
diagonal or edge of $Q$ there are $O(1)$ topological changes in the
structure of $e_{pq}$.  Therefore, all the $Q$-Voronoi edges
$e_{pq}$ undergo a total of $O(k^2n^2)$ changes in their
edgelet structure.
Note that we only count the bisector events themselves
in the $O(kn^2)$ bound on the number of \degenerate bisector events,
and not the \degenerate sequences of corner and flip events associated with them; these
other events will be bounded separately in what follows.

\paragraph{The number of corner events.}
The number of corner events (which can be either \generic or \degenerate) is bounded
by the following lemma.

\begin{lemma}
\lemlab{corners}
Let $Q$ be a convex $k$-gon, and let $P$ be a set of $n$ moving points in $\reals^2$
along algebraic trajectories of bounded degree.
The overall number of corner events (\generic and \degenerate) in $\DT(P)$
is $O(k^2n\lambda_r(n))$, where $r$ is a constant that depends on the degree of
the motion of the points of $P$.
\end{lemma}

\begin{proof}
Fix a point $p$ and a vertex $v_i$ of $Q$, and consider all the corner
events in which (for an appropriate homothetic copy of $Q$)
the vertex (with label) $v_i$ touches $p$. As noted above, at any such event the
center $c$ of $Q$ lies on the ray $u_i[p]$ emanating from $p$ in direction $v_io$
(where $o$ denotes the center of $Q$). Note that, since $p$ is moving, $u_i[p]$
is a moving ray, but its orientation remains fixed. For every other point
$q\in P\setminus\{p\}$, let $\distfn_i[p,q](t)$ denote the distance,
at time $t$, from $p$ along $u_i[p]$ to the center of a homothetic copy of
$Q$ that touches $p$ (at $v_i$) and $q$.

Each function $\distfn_i[p,q](t)$ is well defined only when
$q$ lies strictly in the wedge between two rays emanating from $p$,
one in the direction $\vec{v_i v_{i-1}}$ and the other in the
direction $\vec{v_i v_{i+1}}$. Since the segment $pq$ becomes parallel to
$e_{i-1}$ or to $e_{i}$ only $O(1)$ times, the domain in which
$\distfn_i[p,q](t)$ is defined consists of $O(1)$ open intervals.

The value $\min_q \distfn_i[p,q](t)$ represents the
intersection of $\bd \Vor(p)$ with $u_i[p]$ at time $t$ (recall that
$\Vor(p)$ is the Voronoi cell of $p$ in $\VD(P)$). The point $q$ that
attains the minimum defines the Voronoi edge $e_{pq}$ of $\Vor(p)$ that
$u_i[p]$ intersects.  At appropriate discrete times, the minimum may be
attained by more than one point $q$, and then $u_i[p]$ hits a vertex of $\Vor(p)$.

In other words, we have a collection of $n-1$ partially defined
functions $\distfn_i[p,q]$, for $q\in P\setminus \{p\}$, and the
breakpoints of their lower envelope (where two different functions
attain the envelope simultaneously) represent the corner events that
involve the contact of $v_i$ with $p$.

By our assumption on the motion of $P$, each function
$\distfn_i[p,q]$, within each interval of its domain, is
piecewise algebraic, with $O(k)$ pieces. Each piece encodes a
continuously varying contact of $q$ with a specific edge of $Q$, and
has constant description complexity. Hence (see, e.g.,
\cite[Corollary 1.6]{SA95}) the complexity of the envelope is at
most $O(k\lambda_r(n))$, for an appropriate constant $r$ that
depends on the degree of the motion of the points of $P$. Repeating
the analysis for each point $p$ and each vertex $v_i$ of $Q$, we
obtain the bound asserted in the lemma.

To complete the proof, we also need to take into account the
initial and the final corner events of each \degenerate sequence.
The number of these events, though, is subsumed by the bound
asserted in the lemma, and the proof is completed.
\end{proof}

\paragraph{The number of flip events.}
We distinguish between \generic flips and their \degenerate
counterparts, handling the \degenerate flip events first.

\begin{lemma}
\lemlab{degenerate}
Let $Q$ be a convex $k$-gon, and let $P$ be a set of $n$ moving points
in $\reals^2$ along algebraic trajectories of bounded degree.
The number of \degenerate flip
events in $\DT(P)$ is $O(k^3n\lambda_r(n))$, where $r$ is the
same constant as in the statement of \lemref{corners}.
\end{lemma}

\begin{proof}
We choose three arbitrary edges $e_1,e_2,e_3$ of $Q$, and analyze the
number of \degenerate flips at which two points of $P$ touch $e_1$ and one
point of $P$ lies on each of $e_2$ and $e_3$.  The overall bound on the
number of \degenerate flips is obtained by repeating our analysis for
the $O(k^3)$ such triples of edges.

Let $\tilde{Q}$ be the convex hull of $e_1, e_2$ and $e_3$, with at most
six vertices. Every \degenerate flip event with respect to $Q$ that involves
$e_1$, $e_2$, and $e_3$ in the manner stated in the preceding paragraph
is also a \degenerate flip event with respect to $\tilde{Q}$. To bound the
number of \degenerate flip events for $\tilde{Q}$, we first show how to
charge each of them to a \degenerate corner event (again, with respect to
$\tilde{Q}$) so that each corner event is charged only $O(1)$ times, and
then make use of \lemref{corners}.

Consider the homothetic copy of $\tilde{Q}$ associated with a
\degenerate flip involving four points $p,q,r,s \in P$; we abuse the
notation slightly, as we did earlier, and refer to this copy also as $\tilde{Q}$.
Suppose, more concretely, that $p$ and $q$ touch $e_1$, $r$ touches
$e_2$, and $s$ touches $e_3$. We continuously slide $\tilde{Q}$ so
that it maintains its contacts with
$p,q$ at $e_1$, and with $r$ at $e_2$, and moves\footnote{%
  Note that our continuous motion argument refers to the \emph{stationary}
  point set $P=P(t_0)$, where $t_0$ is the time of the \degenerate flip event.
  Thus we move $\tilde{Q}$ (for the purpose of the analysis) but the
  points of $P$ remain fixed.}
away from $s$.  Specifically, the segments $pr,qr$ partition $\tilde{Q}$
into three parts: the portion that touches $s$, the triangle $\triangle pqr$,
and the third complementary portion (bounded by $pr$). During the continuous motion
of $\tilde{Q}$, the part that was previously incident to $s$ shrinks,
the triangle $\triangle pqr$ does not change (as a portion of $\tilde{Q}$),
and the third part expands; see Figure~\ref{Fig:degenerate-flip-charge}~(a).
The motion of $\tilde{Q}$ stops when one of the following situations is reached:

\begin{figure}[htb]
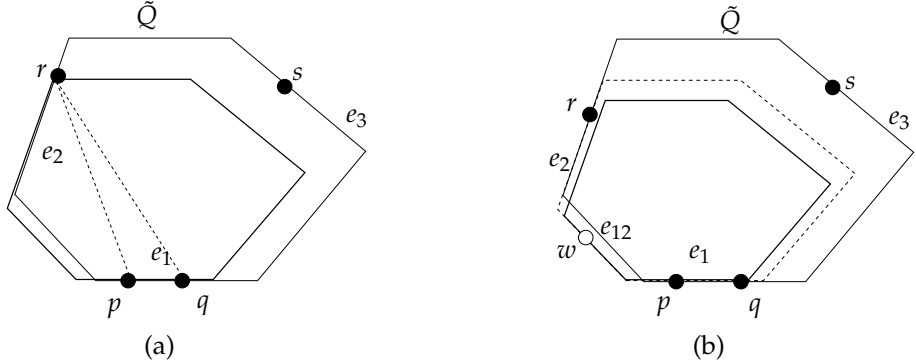

\centering
\input degenerate-flip-count.pspdftex
\caption{Charging a \degenerate flip event. We slide $\tilde{Q}$
while maintaining the contacts with $p$, $q$, $r$:
(a) one of the
points runs into a vertex, (b) $\partial\tilde{Q}$ touches a new point.
Thick polygons denote the corner event to which the flip
event is charged.}
\label{Fig:degenerate-flip-charge}
\end{figure}

\smallskip
\noindent
(i) \emph{One of the points $p,q,r$ runs into a vertex of
$\tilde{Q}$};
see Figure~\ref{Fig:degenerate-flip-charge}~(a) (where $r$ is that
point):  This is a \degenerate corner event of $p,q$ or $r$, which
can be charged for our \degenerate flip event.
This is because both
events occur simultaneously, and any corner event is charged only
$O(1)$ times in this manner (reversing the sliding process,
the fourth point $s$ is the first new point that the expanding
part of $\tilde{Q}$ hits).

\smallskip
\noindent
(ii) \emph{The boundary of the expanding third portion of $\tilde{Q}$
touches a new point $w$ of $P$}; see Figure~\ref{Fig:degenerate-flip-charge}~(b).
In this case $e_2$ cannot be adjacent
to $e_1$, because then there would be no expanding portion of $\tilde{Q}$,
and no new point could be reached. Let $e_{12}$ be the (unique) edge of
$\tilde{Q}$ joining $e_1$ and $e_2$. Since $e_3$ belongs to the
shrinking portion of $\tilde{Q}$, the boundary of the expanding
portion of $\tilde{Q}$ consists of $e_{12}$ and of portions of
$e_1$, $e_2$, and thus $w$ must lie on $e_{12}$;
see Figure \ref{Fig:degenerate-flip-charge}~(b).
We now slide $\tilde{Q}$ along its contacts with $p, q$, and $w$
and away from $r$. Since $e_{12}$ is adjacent to $e_1$, the motion
now stops, as in case (i), when one of $p,q,w$ reaches a vertex of
$\tilde{Q}$. We charge the flip event to this corner event.
There are $O(1)$ possible ways to backtrack from the resulting
corner event to the original \degenerate flip (first by expanding until
$r$ is reached, and then by expanding again until $s$ is reached), which
occurs simultaneously with it. Hence, each corner event is charged
in this manner only $O(1)$ times.

By \lemref{corners}, the number of \degenerate corner events with
respect to $\tilde{Q}$  is $O(n\lambda_r(n))$, so the number of
\degenerate flip events with respect to $\tilde{Q}$ is also
$O(n\lambda_r(n))$. Multiplying this bound by the number $O(k^3)$
of triples of edges $e_1$, $e_2$, $e_3$ yields the bound asserted
in the lemma.
\end{proof}

\begin{lemma} \lemlab{flips}
Let $Q$ be a convex $k$-gon, and let $P$ be a set of $n$ moving points in $\reals^2$ along algebraic
trajectories of bounded degree.
The number of \generic flip events in $\DT(P)$ is
$O(k^4n\lambda_r(n))$, where $r$ is the same constant as in
the statement of \lemref{corners}.
\end{lemma}

\begin{proof}
Each \generic flip event involves a
placement of an empty homothetic copy $Q'$ of $Q$ that touches
simultaneously four points $p_1,p_2,p_3,p_4$ of $P$, in this
counterclockwise order along $\partial Q'$, so that the
Voronoi edge $e_{p_1p_3}$, which was a non-corner edge
before the event, shrinks to a point and is replaced by the
newly emerging non-corner edge $e_{p_2p_4}$ right
after the event. Let $e_i$ denote the edge of $Q'$ that touches $p_i$,
for $i=1,2,3,4$. Since this is a \generic flip, the $e_i$'s are distinct.

We fix the quadruple of edges $e_1,e_2,e_3,e_4$, analyze the number
of flip events involving a quadruple contact with these edges,
and sum the bound over all $O(k^4)$ choices of four edges of $Q$.
For a fixed quadruple of edges $e_1,e_2,e_3,e_4$, we replace $Q$ by
the convex hull $\tilde{Q}$ of these edges, which is at most octagonal,
and note that any flip event of $Q$ involving these four edges is also
a flip event for $\tilde{Q}$. We therefore restrict our attention to
$\tilde{Q}$, which is a convex $k_0$-gon, for some $k_0\le 8$, and
argue that the number of \generic flip events for $\tilde{Q}$, with
contacts at the above four edges $e_1,\ldots, e_4$, is $O(n\lambda_r(n))$.

We note that if $pq$ is a Delaunay edge containing an edgelet which
is the locus of centers of placements of copies $\tilde{Q}'$ of
$\tilde{Q}$ where $p$ and $q$ touch two {\em adjacent} edges of
$\tilde{Q}'$, then $pq$ must be a corner edge. Indeed, shrinking
$\tilde{Q}'$ towards the vertex common to the two edges, so that it
continues to touch $p$ and $q$, will keep it empty, and eventually
reach a placement where either $p$ or $q$ touches a corner of
$\tilde{Q}'$.

Consider the situation just before the critical event takes place,
as depicted in Figure~\ref{Fig:ChewProof} (a).
The $\tilde{Q}$-Voronoi edge $e_{p_1p_3}$ (to simplify the notation, we
write this edge as $e_{13}$, and similarly for the other edges
and vertices in this analysis) is delimited by two $\tilde{Q}$-Voronoi vertices
$\nu_{123}$ and $\nu_{143}$, so that $\tilde{Q}[\nu_{123}]$
touches $p_1,p_2,p_3$ at the respective edges $e_1,e_2,e_3$,
and $\tilde{Q}[\nu_{143}]$ touches $p_1,p_4,p_3$ at the respective edges
$e_1,e_4,e_3$.

\begin{figure}[hbtp]
\centering
\begin{tabular}{ccc}
\input{ChewProof.pspdftex}&\hspace*{3cm}&\input{ChewProof1.pspdftex}\\
\small (a)&&\small (b)
\end{tabular}
\caption{(a) The edge $e_{13}$ in the diagram $\VD(P)$
before it disappears. The endpoint $\nu_{123}$ (resp.,
$\nu_{143}$) of $e_{13}$ corresponds to the homothetic
copy of $\tilde{Q}$ whose edges $e_1,e_2,e_3$ (resp., $e_1,e_4,e_3$)
are incident to the respective points $p_1,p_2,p_3$ (resp.,
$p_1,p_4,p_3$). (b) The tree of non-corner edges that contains
$e_{13}$. In the ``worst-case'' scenario depicted here, all
five solid edges belong to the tree, and the dashed edges are all
corner edges.}\label{Fig:ChewProof}
\end{figure}
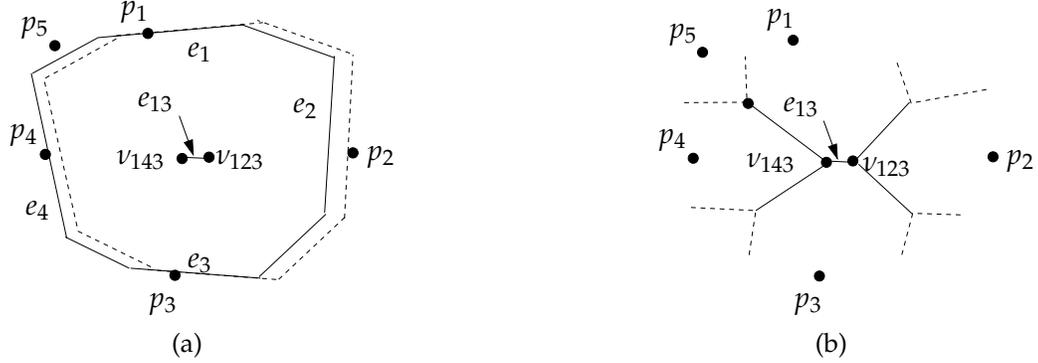

Let $T$ be the maximum connected component of the union of the non-corner
$\tilde{Q}$-Voronoi edges that contains $e_{13}$.  We claim that
$T$ is a tree with at most five edges; specifically, it can include only
$e_{13}$ itself,  $e_{12}$, $e_{23}$ (the edges adjacent
to $\nu_{123}$), and $e_{14}$, $e_{34}$ (the edges
adjacent to $\nu_{143}$); see Figure~\ref{Fig:ChewProof}(b).
To see this, consider,
for specificity, the neighboring edge $e_{12}$ incident to
the vertex $\nu_{123}$, and assume that $e_{12}$ is a
non-corner edge (otherwise it does not belong to $T$, by
construction), so $e_{12}$ is part of $T$. As we move
the center of $\tilde{Q}$ along that edge away from $\nu_{123}$,
$\tilde{Q}$ loses the contact with $p_3$, shrinks on the side of
$p_1p_2$ that contains $p_3$ (and $p_4$, already away from $\tilde{Q}$),
and expands on the other side. Being a non-corner edge,
the other endpoint of $e_{12}$ is a placement at which
the (artificial) edge $e_{12}$ of $\tilde{Q}$ between $e_1$ and
$e_2$ touches another point $p_{5}$ (while $p_1$ and $p_2$ continue
to touch $e_1$ and $e_2$, respectively). Since $e_{12}$ is adjacent
to both edges $e_1$, $e_2$, the new Voronoi edges $e_{15}$ and
$e_{25}$ must both be corner edges. The same argument applies
for the other three Voronoi edges adjacent to $e_{13}$, and
we conclude that $T$ is a tree consisting of at most five edges.

We note that as long as no discrete change occurs at any of the
surrounding corner edges of $T$, it can undergo only $O(1)$ discrete
changes because all its edges are defined by a total of $O(1)$
points of $P$ (and the size of $\tilde{Q}$ is independent of $k$).
Furthermore when a corner edge undergoes a discrete change, this can
affect only $O(1)$ adjacent non-corner trees of the above kind.

Hence, the number of changes in non-corner edges, and in particular
the number of flip events they are involved in, is proportional to
the number of corner events, because only corner events affect the
actual number of breakpoints on an edge. More specifically, if
$e$ is an edge adjacent to $T$, its interaction with $T$ can
change only when its last edgelet (the one nearest to $T$) changes,
and that can happen only at a corner event.  By \lemref{corners} and
the reasoning preceding it (applied to $\tilde{Q}$), the number of
corner events is $O(n\lambda_r(n))$ (the same $r$ as in the previous
lemmas). Multiplying by the $O(k^4)$ choices of quadruples of edges
of $Q$, we conclude that the total number of \generic flip events is
$O(k^4n\lambda_r(n))$.
\end{proof}

Combining the above three lemmas, we obtain the following summary result.
\begin{theorem} \label{Thm:PolygonalVoronoi}
Let $P$ be a set of $n$ moving points in $\reals^2$ along algebraic
trajectories of bounded degree, and let $Q$ be a convex $k$-gon.
Then the number of topological changes in $\VD(P)$ (and in
$\DT(P)$) with respect to $Q$, is $O(k^4n\lambda_r(n))$,
where $r$ is a constant that depends on the degree of the motion
of the points of $P$.
\end{theorem}

\section{A KDS for kinetic maintenance of $\mathbf{VD(P)}$ and $\mathbf{DT(P)}$}
\label{sec:KDS}

In this section we turn the detailed analysis of the kinetic behavior of
$\VD(P)$ and $\DT(P)$, as studied in Section~\ref{sec:kinetic},
into an efficient algorithm for their kinetic maintenance, using the KDS
framework of Basch~\etal~\cite{bgh-dsmd-99}.
In particular, as the points of $P$ move continuously,
$\DT(P)$ and $\VD(P)$  are maintained
using a collection of local \emph{certificates} (Boolean predicates) that assert the
correctness of the topological structure of the current Voronoi diagram (and
thus also of the Delaunay triangulation). When one of these certificates fails,
one of the events discussed in Section~\ref{sec:kinetic} takes place and the KDS repair
mechanism updates $\VD(P), \DT(P)$, and the certificates.
A global \emph{event queue} is used to determine when the next certificate fails.
For simplicity, the KDS is described under the assumptions that the
orientations of all edges and diagonals of $Q$ are distinct, the trajectories
of points in $P$ are in $Q$-general position (i.e., satisfy
(T1)--(T5)), and $P$ is augmented with the points of $P_\infty$.
In cases where these assumptions do not hold, we apply an infinitisimal symbolic perturbation of the vertices of $Q$, thereby putting it in general position, and ensuring the distinctness of the orientations of all the edgese and diagonals of $Q$. With suitable (minor) adjustments, this will allow our KDS to operate in such degenerate situations too.

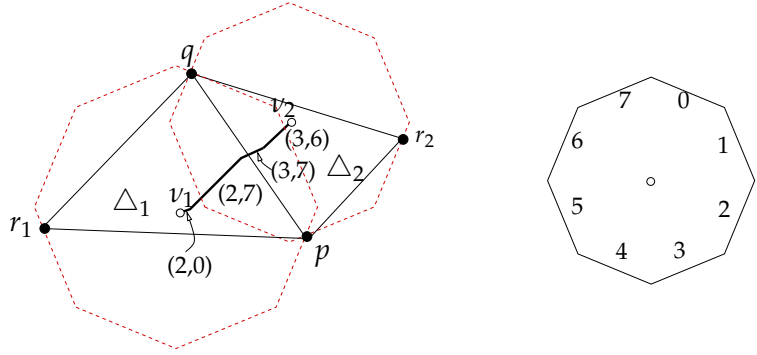
\begin{figure}[htbp]
\centering
\input{edgelet.pspdftex}
\caption{A Delaunay edge $pq$ and its two adjacent triangles $\triangle_1,\triangle_2$.
We have $\delta(v_1,p)=2$, $\delta(v_1,q)=0$, $\delta(v_2,p)=3$, and $\delta(v_2,q)=6$,
using the edge numbering given on the right,
where $v_i$ is the Voronoi vertex corresponding to $\triangle_i$, for $i=1,2$.
The edgelets appearing in $e_{pq}$ are labeled $(2,0), (2,7), (3,7)$, and $(3,6)$.}
\label{Fig:bisector-certificate}
\end{figure}

We need the following definitions and notation to describe the KDS.
For a Voronoi vertex $\nu$ of some Voronoi region $\Vor(p)$, where
$p$ is a point of $P$ not lying at infinity (so $p$ is a vertex of
the Delaunay triangle $\triangle_\nu$ corresponding to $\nu$), we
define $\delta(\nu,p)$ to be the edge of $Q[\nu]$ that touches $p$
(at a \generic time instance, $p$ does not lie at a vertex of
$Q[\nu]$); see Figure \ref{Fig:bisector-certificate}.

Let $e_{pq}$ be a Voronoi edge with endpoints $\nu_1$ and $\nu_2$.
Let $g_1$ (resp., $g_2$) be the edgelet of $\bisect_{pq}$ that
contains the endpoint $v_1$ (resp., $v_2$). Writing
$\psi(g_1)=(i_1,j_1)$ and $\psi(g_2)=(i_2,i_2)$, we have
$\delta(\nu_1,p)=i_1$, $\delta(\nu_2,p)=i_2$, $\delta(\nu_1,q)=j_1$,
and $\delta(\nu_2,q)=j_2$. That is, $\delta(\nu_1,p)$ and
$\delta(\nu_1,q)$ together encode the edgelet of $\bisect_{pq}$
containing $\nu_1$, which is an external edgelet of $e_{pq}$, and an
analogous property holds for $\nu_2$, so the mapping
$\delta(\cdot,\cdot)$ encodes the external edgelets for all edges of
$\VD(P)$.  If $e_{pq}$ is a non-corner edge of $\VD(P)$, then
$g_1=g_2$, and therefore $\delta(\nu_1,p)=\delta(\nu_2,p)$ and
$\delta(\nu_1,q)=\delta(\nu_2,q)$.

 We focus on the more detailed representation of $\VD (P)$;  $\DT(P)$ is easier to represent, and for
the ongoing analysis we regard it merely as the dual of $\VD (P)$.
We assume for now that $\VD (P)$ is represented {\em explicitly};
that is, each Voronoi edge is represented as a sequence of its
edgelets.
 We represent $\VD (P)$ (and $\DT(P)$ too, if so desired) using the so-called \emph{doubly
connected edge list} (DCEL) data structure~\cite{BCKO}, a commonly
used structure for representing planar subdivisions. It stores a
record for each face, edgelet, and vertex (a Voronoi vertex or a
breakpoint) of the subdivision, along with some adjacency
information. The basic entities in this representation are {\em
half-edges}, where each half-edge is an oriented copy of some
Voronoi edgelet. Each half-edge stores pointers to its endpoints
(each of which is a breakpoint or a Voronoi vertex), the Voronoi
cells adjacent to it, its neighboring edgelets, and its label
$(i,j)$. Each Voronoi cell $\Vor(p)$ is associated with its input
point $p$, which stores its trajectory. With each breakpoint or
vertex $\nu$ we can either store its trajectory, or compute it on
the fly from the trajectories of the two or three points defining
$\nu$, and the edges of $Q[\nu]$ that they touch (we can determine
the identity of these edges of $Q[\nu]$ from the labels of the
edgelets adjacent to $\nu$). Note that as long as no event happens,
the trajectories of each vertex and breakpoint are algebraic of
constant complexity. See~\cite{BCKO} for more details on the DCEL
data structure. This explicit representation of $\VD (P)$ requires
$O(nk)$ space.

\subsection{Certificates}\label{Subsec:Certificates}
We maintain three types of certificates, each corresponding to a
different type of event. Each certificate is stored at the event
queue, with its next future failure time as its priority. In
addition, we store with each certificate information that will be
needed to update the diagram and the event queue when the
certificate fails.

\medskip
\noindent \textbf{\textit{Bisector certificates.}}\hspace{2mm} For
each internal edgelet $g$ of each Voronoi edge $e_{pq}$, with
$\psi(g) = (i,j)$, we maintain a certificate asserting that the
orientation of $pq$ lies between the two extreme orientations of
lines that cross both edges $e_i$ and $e_j$ of $Q$ (these are the
orientations of $v_iv_j$ and $v_{i+1}v_{j+1}$). We recall that $i$
and $j$ are not consecutive for internal edgelets. This certificate
fails when $\vec{pq}$ attains one of these extreme orientations; See
Figure \ref{Fig:bisector}. The failure causes $g$ to shrink and
disappear and be replaced by a new internal edgelet $g'$ with
$\psi(g') = (i-1,j-1)$ or $\psi(g') = (i+1,j+1)$, as appropriate. To
efficiently perform the required updates when this certificate
fails, the certificate also stores a direct pointer to the edgelet
itself (in the DCEL). We call this certificate a {\em generic
bisector certificate}.

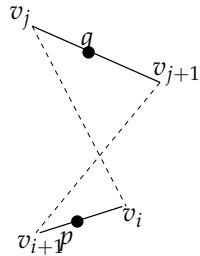
\begin{figure}[htbp]
\centering
\input{bisector-certificate.pspdftex}
\caption{The generic bisector certificate corresponding to an
internal edgelet $g$ such that $\psi(g) = (i,j)$. It fails when
$\vec{pq}$ becomes parallel to $\vec{v_iv_j}$ or
$\vec{v_{i+1}v_{j+1}}$} \label{Fig:bisector}
\end{figure}

External edgelets of an edge $e_{pq}$ come in two types: edgelets
that are portions of non-terminal (bounded) edgelets of the bisector
$b_{pq}$, and edgelets that are portions of the terminal rays of
$b_{pq}$; see Figure \ref{Fig:BisectLabel}. Our general position
assumption (T5) implies that we do not have to associate bisector
events with edgelets of the former type, because before such an
edgelet shrinks to a point, as an edgelet of $b_{pq}$, it will
disappear as an edgelet of $e_{pq}$, at a corresponding corner
event, if $e_{pq}$ is a corner edge, or at a flip event, if $e_{pq}$
is a non-corner edge.

Consider then external edgelets of the latter type.  Each such
edgelet $g$ has a label of the form $\psi(g) = (i,i+1)$, and its
bisector certificate fails when the corresponding points $p$ and $q$
are such that $\vec{pq}$ becomes parallel to either $e_i$ or
$e_{i+1}$. This is a singular bisector event. We refer to such a
certificate as a {\em singular bisector certificate}.

\medskip
\noindent \textbf{\textit{(Voronoi) Vertex
certificates.}}\hspace{2mm} For each vertex $\nu$ of $\VD(P)$ not
lying at infinity and for each point $p\in P$ such that $\Vor(p)$ is
adjacent to $\nu$ and $p$ is not at infinity, we maintain a
certificate  asserting that  $\delta(\nu,p) = v_iv_{i+1}$ for a
suitable index $i$, i.e.,
 $p$ lies on the edge $v_iv_{i+1}$ of the homothetic copy $Q[\nu]$ of $Q$.
Because the motion of $p$ is continuous, the certificate fails only
when $p$ coincides with $v_i$ or $v_{i+1}$ in $Q[\nu]$; that is, the
failure occurs at a corner event. This certificate stores pointers
to $\nu$, $p$. The index $i$ and $\delta(\nu,p)$  can  be identified
using the labels of the edgelets adjacent to $\nu$ in the DCEL.

Let $q$ and $r$ be the other two points of $P$ defining $\nu$; i.e.\
$\nu=\nu_{pqr}$. The corner event is either generic, when the vertex of $Q$
with which $p$ coincides is not adjacent to $\delta(\nu,q)$ or to
$\delta(\nu,r)$, or singular, when this vertex is adjacent to one of
these edges, say to $\delta(\nu,q)$. In the latter case, this is the
first singular corner event in the singular sequence associated with
the corresponding singular bisector event involving $p$ and $q$ at
which $\vec{pq}$ is parallel to $\delta(\nu,q)$. When the singular
corner event happens the placement $Q[\nu]$ is a placement at which
$p$ lies at the common vertex of $\delta(\nu,p)$ and
$\delta(\nu,q)$, and $q$ lies on the adjacent edge $\delta(\nu,q)$.
In  the terminology of Section \ref{Subsec:Degenerate},
$\nu=\eta^-$; see there for more details. We refer to a vertex
certificate that fails at a generic corner event as a {\em generic
vertex certificate} and to a vertex certificate that fails at a
singular corner event as a {\em singular vertex certificate}.

For any sequence of singular events, only the initial corner event is detected through the failure of a singular vertex certificate; the remaining corner events 
are detected and treated separately, as described  below.

\medskip
\noindent \textbf{\textit{(Voronoi) Edge certificates.}}\hspace{2mm}
For each non-corner Voronoi edge $e_{pq}$, we maintain the following
certificate. Let $\nu_{pqr}$ and $\nu_{pqw}$ be the two endpoints of
$e_{pq}$, i.e., $\triangle pqr$ and $\triangle pqw$ are the two
Delaunay triangles adjacent to the edge $pq$. We maintain the
certificate asserting that $e_{pq}$ exists, namely, that $\nu_{pqr}
\not= \nu_{pqw}$; since the trajectories of these vertices are
available, and each of them is of constant complexity (as long as no
other event affecting these vertices occurs), we can compute the
failure time of this certificate in $O(1)$ time. The edge
certificate fails when $\nu_{pqr} = \nu_{pqw}$ (while maintaining
$\delta(\nu_{pqr},p)=\delta(\nu_{pqw},p) \not=
\delta(\nu_{pqr},q)=\delta(\nu_{pqw},q)\not=\delta(\nu_{pqr},r)\not=\delta(\nu_{pqw},w)$),
that is when $p$,$q$,$r$, and $w$ become $Q$-circular while touching
distinct edges of $Q$. This is a generic flip event.

Singular flip events, occurring as part of a compound singular
bisector event, as described in Section \ref{Subsec:Degenerate}, are
detected and  treated separately, as described  below.

\subsection{Handling certificate failures}\label{Subsec:OnFailure}
We now describe the KDS repair mechanism that updates $\VD(P)$, $\DT(P)$,
the certificates, and the event queue, when a certificate fails.
The queue contains the failure times of generic bisector and vertex certificates, edge certificates, and singular bisector and vertex certificates. 
Note that singular certificates come in pairs, each consisting of a
singular bisector certificate and a singular vertex
certificate, both failing {\em at the same time}. When we extract
from the queue one certificate of such pair, we immediately extract
its sibling too, and process them together, as will be described
later. We note that, except for this kind of simultaneous failure of
two distinct certificates, assumption (T5) implies that there are no
other duplications of failure times in the queue.

The processing of the failure of a certificate, when it is extracted from the queue, is handled according to its type,
and proceeds as follows.

\medskip
\noindent
\textbf{\textit{Generic bisector certificates.}}\hspace{2mm}
Consider such a certificate involving the shrinking of some edgelet $g$ of some Voronoi edge
$e_{pq}$. We remove $g$ from $\VD(P)$ and replace it by a new (internal) edgelet $g'$. We compute
$\psi(g')$ from $\psi(g)$, by adding or subtracting $1$ from its two indices, as appropriate, generate the new bisector certificate for $g'$,
compute its failure time, and insert it into the queue.

\medskip
\noindent
\textbf{\textit{Generic vertex certificates.}}\hspace{2mm}
Let $\nu := \nu_{pqr}$ be a Voronoi vertex such that one of its certificates, say
$\delta(\nu, p) = e_{i}$, fails, and suppose that this is a generic vertex certificate.
 This happens when $p$ coincides with one of the endpoints of
$e_i$, say, $v_i$, in the homothetic copy $Q[\nu]$ of $Q$ (and $\delta(\nu,q)$ and $\delta(\nu,r)$ are not adjacent to $v_i$).
 Then $\nu$ lies at a
breakpoint of $e_{pq}$ and at a breakpoint of $e_{pr}$, and we
detect a generic corner event, at which, without loss of generality,
$e_{pq}$ loses an edgelet $g$ and $e_{pr}$ gains an edgelet $g'$. We
remove $g$ from the DCEL, including its adjacent breakpoint (which
is ``overtaken'' by $\nu$), create the new edgelet $g'$ on $e_{pr}$
and the corresponding new breakpoint delimiting it, and update the
DCEL accordingly. We next update the certificates as follows.

(i) We find the new edge $e'$ of $Q$ (which is either $e_{i-1}$ or $e_{i+1}$)
that supports $p$
after the event,
 replace the old certificate
with a new certificate asserting that $\delta(\nu,p)=e'$, compute the failure time of the new certificate, and
update the event queue.

(ii) Since $p$ lies on a new edge of $Q[\nu]$ after the event, the
algebraic representation of the trajectory of $\nu$ changes, and thus we also update
the failure times of the vertex
certificates that specify $\delta(\nu_{pqr},q)$ and $\delta(\nu_{pqr},r)$.

(iii) The new edgelet  $g'$ is an external edgelet of $e_{pr}$. If
its label consists of two consecutive edges of $Q$ then $g'$
generates a singular bisector certificate, which we process into the
queue (otherwise, as noted earlier, no action is needed here).

(iv) The edgelet of $e_{pr}$ adjacent to $g'$ may now become
internal, and then we generate the corresponding generic bisector
certificate. Otherwise $e_{pr}$ was a non-corner edge before the
event and becomes now a corner edge, so we remove the edge
certificate associated with $e_{pr}$ from the queue.

(v) The edgelet $g^+$ of $e_{pq}$ adjacent to $g$ becomes external. We first remove the generic bisector certificate of $g^+$.
Next, if $e_{pq}$ becomes a non-corner edge, then we generate
a new edge certificate associated with $e_{pq}$. 

(vi) If the third Voronoi edge $e_{qr}$ adjacent to $\nu$ is a non-corner edge then its certificate
failure time has to be updated, because of the new trajectory of $\nu$.

\medskip
\noindent \textbf{\textit{Edge certificates.}}\hspace{2mm} If a
non-corner edge $e_{pq}$, with endpoints $\nu_{pqr}$ and
$\nu_{pqw}$, shrinks to a point, i.e., $p, q, r, w$ become
$Q$-cocircular, then the edge certificate associated with $e_{pq}$
fails. We replace in $\VD(P)$ the non-corner edge $e_{pq}$ with the
non-corner edge $e_{rw}$ and the vertices $\nu_{pqr}$ and
$\nu_{pqw}$ with $\nu_{rwp}$ and $\nu_{rwq}$ (these are the
endpoints of the new edge). The edge $pq$ in $\DT(P)$ is flipped to
$rw$, and $\triangle pqr, \triangle pqw$ are replaced with
$\triangle rwp$ and $\triangle rwq$.

Next, we remove the six vertex certificates associated with the
vertices $\nu_{pqr}$ and $\nu_{pqw}$. We add the edge certificate
corresponding to the edge $e_{rw}$ and the six vertex certificates
corresponding to the vertices $\nu_{rwp}$ and $\nu_{rwq}$. Finally,
since the endpoints of the edges $e_{pr}, e_{qr}, e_{pw}$, and
$e_{qw}$ have changed, if any of them is a non-corner edge then we
also update the certificate corresponding to that edge, recompute
its failure time, and update the queue.

\medskip
\noindent \textbf{\textit{Singular bisector and vertex
certificates.}}\hspace{2mm} Finally we turn to the processing of a pair of
a singular bisector certificate and a singular initial corner
certificate that fail simultaneously. Let $p$ and $q$ be the points
generating the bisector event, such that at the corresponding corner
event $p$ touches $Q$ at a vertex $v_i$ and $q$ lies on, say, the
edge $v_iv_{i+1}$. Following the notations of Section
\ref{sec:kinetic}, we denote by $\eta^- = \nu_{pqr^-}$ the Voronoi
vertex at which the corner event occurs, where $r^-$ is the third
point defining this event, and by $\rho^-$ the terminal ray of
$\bisect_{pq}$ that contains $\eta^-$. 

Let $g$ be the edgelet of $e_{pr^-}$ that disappears at the event.
We remove $g$ from the DCEL structure, and remove its two endpoints,
one of which is $\eta^-$. We also remove the external edgelet of $e_{pq}$ contained in $\rho^-$ and the
certificates associated with these features. Recall (see Section
\ref{Subsec:Degenerate})  that the next edgelet $g^+$ of $e_{pr^-}$
and the external edgelet $g^-$ of $e_{qr^-}$ adjacent to $\eta^-$
were parallel before the event and become aligned after the event.

We now execute a process that effectively simulates the rotational
sweep of $\rho$ from $\rho^-$ to $\rho^+$, as described in Section
\ref{Subsec:Degenerate}, continuously traversing the corresponding portion $\gamma$ of
$\bd \Vor(p)$ between $\eta^-$ and $\eta^+$ with the endpoint $\eta=\rho\cap \gamma$ of $e_{pq}$, and transferring this
portion to $\Vor(q)$. Specifically, we iterate over the edges and
edgelets of $\gamma$ in order, starting with the edgelet $g^+$
following $g$ on $e_{pr^-}$, and do the following. Let $f$ be the
edgelet currently traversed by $\eta$, of some old Voronoi edge $e_{pr}$. Notice that $e_{pr}$ is now split by $\eta=v_{pqr}$ into
$e_{qr}$, the portion already traversed by $\rho$, and the remainder $e_{pr}$, still labeled as such.
We test whether
$\eta^+ \in f$, that is, whether the terminal corner event, where
$q$ lies at $v_{i+1}$, occurs before $f$ ends.

Suppose first that
$\eta^+ \not\in f$. In this case the entire $f$ becomes an edgelet
of the new edge $e_{qr}$ when $\eta$ coincides with the other endpoint $\nu$ of $f$.
If $\nu$ is a breakpoint, this corresponds to a singular corner event. Otherwise, $\nu$ is an old Voronoi vertex $\nu_{prr'}$, so its coincidence with $\eta$ indicates a singular flip event of $e_{pr}$ to $e_{qr'}$.
In both cases, we replace the old bisector certificate
associated with $f$, if there was one, by a corresponding
certificate in which $q$ replaces $p$. 
In case of a singular flip event, we replace it by $\nu_{qrr'}$
and update its vertex certificates.
In addition, if the old edge $e_{pr}$ was non-corner, we also replace its edge certificate by the corresponding edge
certificate in which $q$ replaces $p$ again. 

A somewhat special treatment applies to the first edgelet $g^+$ in the sequence. If $\eta^+\not\in g^+$ then we remove $g^+$, as it now merges with the last edget $g^-$ of $e_{pr^-}$ into a common larger edgelet, and we replace the endpoint of $g^-$ (which was $\eta^-$) by the endpoint of $g^+$. The handling of this endpoint, if it is a vertex, is done exactly as just described. Finally, if $e_{qr^-}$ was a non-corner edge, we remove its certificate from the queue as it is now obsolete. If in addition
$g^+$ ends at a Voronoi vertex, $e_{qr^-}$ continues to be non-corner, and we generate a new edge certificate for it.

Finally, consider the case where $\eta^+$ lies on the current
edgelet $f$. In this case we create (i) a new breakpoint $w$, for
which $q$ touches $Q[w]$ at $v_{i+1}$ (indicating the final corner event), (ii) the new Voronoi vertex
$\eta^+ = \nu_{pqr}$, and (iii) three new edgelets $f_1$, $f_2$, and
$f_3$ that replace $f$, where $f_1$ is the portion of $f$ preceding
$w$, $f_2$ is the new edgelet between $w$ and $\eta^+$ (which is
denoted by $h$ in Figure \ref{Fig:Change-degenerate-VD}(c)), and
$f_3$ is the remainder of $f$; $f_1$ and $f_2$ are now part of
$e_{qr}$ while $f_3$ is still an (external) edgelet of $e_{pr}$. 
We
also add a new external edgelet $e_{pq}$, which is contained in $\rho^+$ 
and is adjacent to $\eta^+$. 
 We generate new certificates associated with these new features of the diagram and delete the old ones, similar to the preceding discussion.

This completes the description of the KDS repair mechanism.

The KDS as we described it maintains $O(nk)$ certificates, so its
event queue  takes $O(nk)$ space (in addition to $\VD(P)$ which also
takes  $O(nk)$ space). When a generic certificate fails we make
$O(1)$ changes to the diagram in $O(1)$ time and $O(1)$ changes to
the event queue in $O(\log n)$ time. The simulation of the
rotational sweep  when we process a singular bisector event may take
 a long time. But
 it makes $O(1)$ changes to the diagram in $O(1)$ time and
$O(1)$ changes to the event queue in $O(\log n)$ time, for each
singular corner or flip event along $\gamma$, and  the analysis of
Section \ref{sec:count}  bounds the overall number of these singular
sub-events. The formulation of the following theorem (and of Theorem
\ref{Thm:MaintainPolygDT1}) assumes this point of view.

Putting everything together, obtain the following main result of
this paper.

\begin{theorem}\label{Thm:MaintainPolygDT}
Let $P$ be a set of $n$ moving points in $\reals^2$ along algebraic
trajectories of bounded degree, and let $Q$ be a convex $k$-gon.
Then both $\DT(P)$ and $\VD(P)$, including the sequence of edgelets of each Voronoi edge, can be maintained using a KDS that
requires $O(nk)$ storage, encounters $O(k^4n\lambda_r(n))$ events,
and processes each event in $O(\log n)$ time, in a total of
$O(k^4n\lambda_r(n)\log n)$ time. Here $r$ is a constant that
depends on the degree of the motion of the points of $P$.
\end{theorem}

\subsection{Implicit representation of $\VD(P)$}

The drawback of the maintenance of the explicit representation of
$\VD(P)$, including the edgelets comprising each Voronoi edge, is
that it uses $O(nk)$ storage. We can reduce the storage to $O(n)$ by
completely ignoring the edgelets, and storing each Voronoi edge as a
single entity in the DCEL connecting the two Voronoi vertices that
delimit it (breakpoints are not maintained either).

This in fact simplifies the KDS, because there is no need anymore to
maintain bisector certificates or detect bisector events. That is,
we only maintain generic and initial singular corner certificates,
and generic edge certificates. We also store $\delta(\nu,p)$,
$\delta(\nu,q)$, and $\delta(\nu,r)$ explicitly with each vertex
$\nu=\nu_{pqr}$. As discussed before, $\delta(\nu,p)$,
$\delta(\nu,q)$, and $\delta(\nu,r)$ encode the labels of the
external edgelets adjacent to $\nu$.\footnote{We cannot compute
these values now from the labels of the external edgelets as we do
not maintain the edgelets.} Suppose that a generic vertex
certificate asserting that $\delta(\nu,p)=e_i=v_iv_{i+1}$ fails when
$p$ reaches one of the endpoints $v_i$ or $v_{i+1}$ of $e_i$. We
compute using $\delta(\nu,p)$, $\delta(\nu,q)$, $\delta(\nu,r)$,
which among the edges $e_{pr}$ and $e_{qr}$ gains an edgelet and
which loses an edgelet, and the labels of the new external edgelets
of $e_{pr}$ and $e_{qr}$ that touch $\nu$.
 We  compute the new vertex certificates of $\nu$ and insert them
   into the queue replacing the old vertex certificates associated with $\nu$. We check,
by comparing the external edgelets of $e_{pr}$ and $e_{qr}$, if any
of these edges changes its  type from a corner edge to a non-corner
edge or vice versa. If indeed the type of such edge changes, we
remove or add the corresponding edge certificate, as appropriate. In
addition we update the edge certificate of $e_{pq}$ if it is a
non-corner edge and update $\delta(\nu,p)$. The processing of a flip
event is similar to its processing when the diagram is explicit.

When an initial singular corner event is detected, we run the same
process as in the previous subsection, simulating the rotational
sweep of Section \ref{Subsec:Degenerate}. We perform this
 process by traversing  the corresponding portion $\gamma$ of
$\bd \Vor(p)$ between $\eta^-$ and $\eta^+$, essentially as before,
but here we compute the sequence of edgelets of each edge $e_{pr}$
along $\gamma$ on the fly, starting and ending at its external
edgelets (which are encoded by the $\delta$ values of the vertices
adjacent to $e_{pr}$). The tracing of the edgelets is needed only
for finding the edgelet containing $\eta^+$, so that we can generate
the relevant vertex certificates involving $\delta(\eta^+,p)$,
$\delta(\eta^+,q)$, and $\delta(\eta^+,r^+)$ (where $r^+$ is the
third point defining $\eta^+$).

Putting everything together, we obtain the following theorem.

\begin{theorem}\label{Thm:MaintainPolygDT1}
Let $P$ be a set of $n$ moving points in $\reals^2$ along algebraic
trajectories of bounded degree, and let $Q$ be a convex $k$-gon.
Then both $\DT(P)$ and $\VD(P)$ (without the edgelet subdivision of
 its edges) can be maintained using a KDS
that requires $O(n)$ storage, encounters $O(k^4n\lambda_r(n))$
events, and processes each event in $O(\log n)$ time, in a total of
$O(k^4n\lambda_r(n)\log n)$ time. Here $r$ is a constant that
depends on the degree of the motion of the points of $P$.
\end{theorem}

%
%

\section{Conclusion}
\label{sec:concl}

In this paper we have presented a detailed and comprehensive analysis
of kinetic Voronoi diagrams and Delaunay triangulations under a convex distance function
induced by a convex polygon. We have shown that the number of topological
changes, when the input points move along algebraic
trajectories of constant degree, is nearly quadratic in the input size,
and that the diagrams can be maintained kinetically by a simple
KDS, which is efficient, compact, and responsive, in the terminology of
\cite{bgh-dsmd-99}. As stated in the introduction, our result, along with the
result in the companion paper~\cite{Stable}, lead to a compact, responsive, and
efficient KDS for maintaining the stable edges of the Euclidean Delaunay
triangulation of $P$, as defined earlier, that processes a near-quadratic
number of events if the motion of $P$ is algebraic of bounded degree.
See~\cite{Stable} for further details.

An interesting open problem is to improve the dependence on $k$ of the bounds
on the number of events, as obtained in Section~\ref{sec:count}.
We tend to conjecture a near-quadratic dependence on $k$ (down from the
current factor $k^4$). A more challenging question is to develop an efficient
KDS for Voronoi diagrams under polyhedral distance functions in $\reals^3$.
Of course, the major open problem on this topic is to extend the recent result by
Rubin~\cite{RubinUnit} to the case when the points of $P$ move with different speeds.

\end{document}

%% file: CornerPlacement.pspdftex
\begin{picture}(0,0)%
\includegraphics{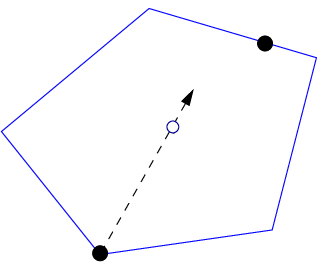}%
\end{picture}%
\setlength{\unitlength}{2072sp}%
\begingroup\makeatletter\ifx\SetFigFont\undefined%
\gdef\SetFigFont#1#2#3#4#5{%
  \reset@font\fontsize{#1}{#2pt}%
  \fontfamily{#3}\fontseries{#4}\fontshape{#5}%
  \selectfont}%
\fi\endgroup%
\begin{picture}(2904,2786)(507,-5163)
\put(3058,-2644){\makebox(0,0)[lb]{\smash{{\SetFigFont{10}{12.0}{\rmdefault}{\mddefault}{\updefault}{\color[rgb]{0,0,0}$q$}%
}}}}
\put(673,-2932){\makebox(0,0)[lb]{\smash{{\SetFigFont{10}{12.0}{\rmdefault}{\mddefault}{\updefault}{\color[rgb]{0,0,.56}$Q'$}%
}}}}
\put(1329,-5054){\makebox(0,0)[lb]{\smash{{\SetFigFont{10}{12.0}{\rmdefault}{\mddefault}{\updefault}{\color[rgb]{0,0,0}$p$}%
}}}}
\put(1614,-4550){\makebox(0,0)[lb]{\smash{{\SetFigFont{10}{12.0}{\rmdefault}{\mddefault}{\updefault}{\color[rgb]{0,0,.56}$v_i$}%
}}}}
\put(2029,-3863){\makebox(0,0)[lb]{\smash{{\SetFigFont{10}{12.0}{\rmdefault}{\mddefault}{\updefault}{\color[rgb]{0,0,.56}$o$}%
}}}}
\end{picture}%

%% file: EdgeletStraight.pspdftex
\begin{picture}(0,0)%
\includegraphics{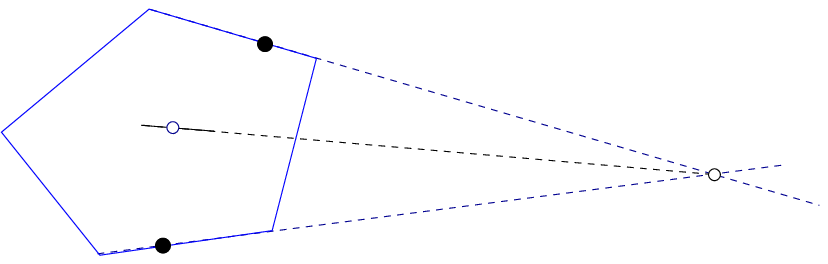}%
\end{picture}%
\setlength{\unitlength}{2072sp}%
\begingroup\makeatletter\ifx\SetFigFont\undefined%
\gdef\SetFigFont#1#2#3#4#5{%
  \reset@font\fontsize{#1}{#2pt}%
  \fontfamily{#3}\fontseries{#4}\fontshape{#5}%
  \selectfont}%
\fi\endgroup%
\begin{picture}(7501,2714)(507,-5091)
\put(1899,-4982){\makebox(0,0)[lb]{\smash{{\SetFigFont{10}{12.0}{\rmdefault}{\mddefault}{\updefault}{\color[rgb]{0,0,0}$p$}%
}}}}
\put(673,-2932){\makebox(0,0)[lb]{\smash{{\SetFigFont{10}{12.0}{\rmdefault}{\mddefault}{\updefault}{\color[rgb]{0,0,.56}$Q'$}%
}}}}
\put(2010,-3855){\makebox(0,0)[lb]{\smash{{\SetFigFont{10}{12.0}{\rmdefault}{\mddefault}{\updefault}{\color[rgb]{0,0,.56}$o$}%
}}}}
\put(2269,-4457){\makebox(0,0)[lb]{\smash{{\SetFigFont{10}{12.0}{\rmdefault}{\mddefault}{\updefault}{\color[rgb]{0,0,.56}$e_i$}%
}}}}
\put(2053,-2840){\makebox(0,0)[lb]{\smash{{\SetFigFont{10}{12.0}{\rmdefault}{\mddefault}{\updefault}{\color[rgb]{0,0,.56}$e_j$}%
}}}}
\put(1377,-3394){\makebox(0,0)[lb]{\smash{{\SetFigFont{10}{12.0}{\rmdefault}{\mddefault}{\updefault}{\color[rgb]{0,0,0}$b_{pq}^\poly$}%
}}}}
\put(3058,-2644){\makebox(0,0)[lb]{\smash{{\SetFigFont{10}{12.0}{\rmdefault}{\mddefault}{\updefault}{\color[rgb]{0,0,0}$q$}%
}}}}
\put(7023,-3790){\makebox(0,0)[b]{\smash{{\SetFigFont{10}{12.0}{\rmdefault}{\mddefault}{\updefault}{\color[rgb]{0,0,0}$\xi_{ij}$}%
}}}}
\end{picture}%

%% file: bisect-new1.pspdftex
\begin{picture}(0,0)%
\includegraphics{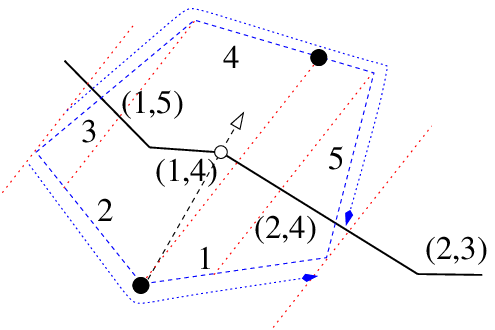}%
\end{picture}%
\setlength{\unitlength}{2072sp}%
\begingroup\makeatletter\ifx\SetFigFont\undefined%
\gdef\SetFigFont#1#2#3#4#5{%
  \reset@font\fontsize{#1}{#2pt}%
  \fontfamily{#3}\fontseries{#4}\fontshape{#5}%
  \selectfont}%
\fi\endgroup%
\begin{picture}(4542,3269)(1256,-5351)
\put(1866,-2688){\makebox(0,0)[lb]{\smash{{\SetFigFont{10}{12.0}{\rmdefault}{\mddefault}{\updefault}{\color[rgb]{0,0,0}$b_{pq}^\poly$}%
}}}}
\put(2732,-5213){\makebox(0,0)[lb]{\smash{{\SetFigFont{10}{12.0}{\rmdefault}{\mddefault}{\updefault}{\color[rgb]{0,0,0}$p$}%
}}}}
\put(4308,-2736){\makebox(0,0)[lb]{\smash{{\SetFigFont{10}{12.0}{\rmdefault}{\mddefault}{\updefault}{\color[rgb]{0,0,0}$q$}%
}}}}
\put(3142,-2349){\makebox(0,0)[lb]{\smash{{\SetFigFont{10}{12.0}{\rmdefault}{\mddefault}{\updefault}{\color[rgb]{0,0,0}$C_{qp}$}%
}}}}
\put(2069,-4773){\makebox(0,0)[rb]{\smash{{\SetFigFont{10}{12.0}{\rmdefault}{\mddefault}{\updefault}{\color[rgb]{0,0,0}$C_{pq}$}%
}}}}
\end{picture}%

%% file: bisect.pspdftex
\begin{picture}(0,0)%
\includegraphics{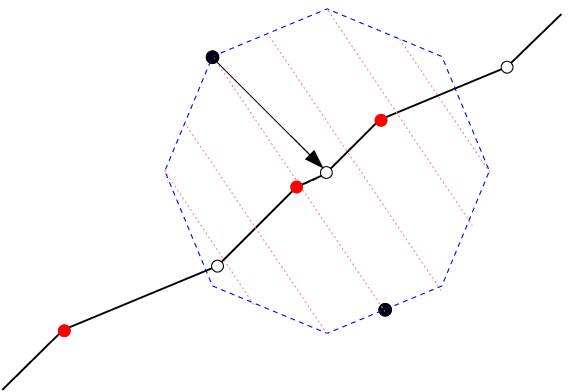}%
\end{picture}%
\setlength{\unitlength}{1973sp}%
\begingroup\makeatletter\ifx\SetFigFont\undefined%
\gdef\SetFigFont#1#2#3#4#5{%
  \reset@font\fontsize{#1}{#2pt}%
  \fontfamily{#3}\fontseries{#4}\fontshape{#5}%
  \selectfont}%
\fi\endgroup%
\begin{picture}(5409,3704)(1222,-3462)
\put(5026,-3001){\makebox(0,0)[lb]{\smash{{\SetFigFont{10}{12.0}{\rmdefault}{\mddefault}{\updefault}{\color[rgb]{0,0,0}$q$}%
}}}}
\put(4530,-1575){\makebox(0,0)[lb]{\smash{{\SetFigFont{10}{12.0}{\rmdefault}{\mddefault}{\updefault}{\color[rgb]{0,0,0}$\bisect_{pq}^\poly$}%
}}}}
\put(2984,-49){\makebox(0,0)[lb]{\smash{{\SetFigFont{11}{13.2}{\rmdefault}{\mddefault}{\updefault}{\color[rgb]{0,0,0}$p$}%
}}}}
\end{picture}%

%% file: extended-DT.pspdftex
\begin{picture}(0,0)%
\includegraphics{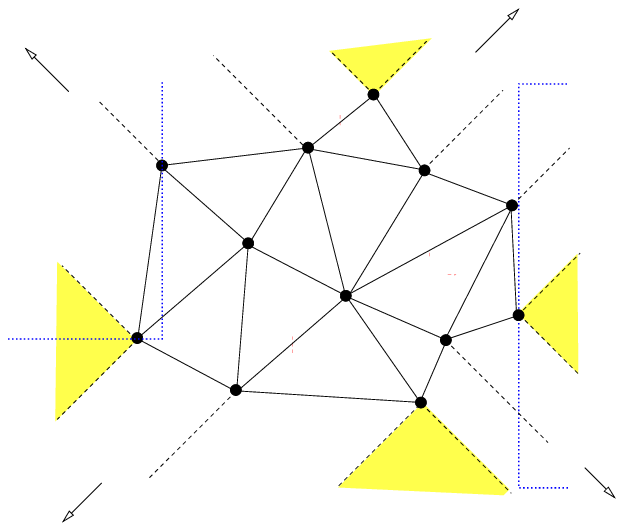}%
\end{picture}%
\setlength{\unitlength}{1657sp}%
\begingroup\makeatletter\ifx\SetFigFont\undefined%
\gdef\SetFigFont#1#2#3#4#5{%
  \reset@font\fontsize{#1}{#2pt}%
  \fontfamily{#3}\fontseries{#4}\fontshape{#5}%
  \selectfont}%
\fi\endgroup%
\begin{picture}(7116,6360)(1512,-7174)
\put(3473,-2477){\makebox(0,0)[lb]{\smash{{\SetFigFont{9}{10.8}{\rmdefault}{\mddefault}{\updefault}{\color[rgb]{0,0,0}1}%
}}}}
\put(3139,-5076){\makebox(0,0)[lb]{\smash{{\SetFigFont{9}{10.8}{\rmdefault}{\mddefault}{\updefault}{\color[rgb]{0,0,0}2}%
}}}}
\put(4406,-3410){\makebox(0,0)[b]{\smash{{\SetFigFont{9}{10.8}{\rmdefault}{\mddefault}{\updefault}{\color[rgb]{0,0,0}3}%
}}}}
\put(4322,-5659){\makebox(0,0)[b]{\smash{{\SetFigFont{9}{10.8}{\rmdefault}{\mddefault}{\updefault}{\color[rgb]{0,0,0}4}%
}}}}
\put(5472,-4610){\makebox(0,0)[b]{\smash{{\SetFigFont{9}{10.8}{\rmdefault}{\mddefault}{\updefault}{\color[rgb]{0,0,0}5}%
}}}}
\put(5039,-2310){\makebox(0,0)[b]{\smash{{\SetFigFont{9}{10.8}{\rmdefault}{\mddefault}{\updefault}{\color[rgb]{0,0,0}6}%
}}}}
\put(5755,-2327){\makebox(0,0)[b]{\smash{{\SetFigFont{9}{10.8}{\rmdefault}{\mddefault}{\updefault}{\color[rgb]{0,0,0}7}%
}}}}
\put(6505,-3210){\makebox(0,0)[b]{\smash{{\SetFigFont{9}{10.8}{\rmdefault}{\mddefault}{\updefault}{\color[rgb]{0,0,0}8}%
}}}}
\put(6455,-4443){\makebox(0,0)[b]{\smash{{\SetFigFont{9}{10.8}{\rmdefault}{\mddefault}{\updefault}{\color[rgb]{0,0,0}11}%
}}}}
\put(6672,-5459){\makebox(0,0)[b]{\smash{{\SetFigFont{9}{10.8}{\rmdefault}{\mddefault}{\updefault}{\color[rgb]{0,0,0}12}%
}}}}
\put(8001,-2024){\makebox(0,0)[b]{\smash{{\SetFigFont{9}{10.8}{\rmdefault}{\mddefault}{\updefault}{\color[rgb]{0,0,0}$C_2$}%
}}}}
\put(2017,-1366){\makebox(0,0)[lb]{\smash{{\SetFigFont{9}{10.8}{\rmdefault}{\mddefault}{\updefault}{\color[rgb]{0,0,0}$q_1$}%
}}}}
\put(8566,-6315){\makebox(0,0)[lb]{\smash{{\SetFigFont{9}{10.8}{\rmdefault}{\mddefault}{\updefault}{\color[rgb]{0,0,0}$q_3$}%
}}}}
\put(7684,-3349){\makebox(0,0)[b]{\smash{{\SetFigFont{9}{10.8}{\rmdefault}{\mddefault}{\updefault}{\color[rgb]{0,0,0}10}%
}}}}
\put(7390,-4783){\makebox(0,0)[b]{\smash{{\SetFigFont{9}{10.8}{\rmdefault}{\mddefault}{\updefault}{\color[rgb]{0,0,0}9}%
}}}}
\put(1527,-4733){\makebox(0,0)[b]{\smash{{\SetFigFont{9}{10.8}{\rmdefault}{\mddefault}{\updefault}{\color[rgb]{0,0,0}$C_1$}%
}}}}
\put(7547,-1105){\makebox(0,0)[lb]{\smash{{\SetFigFont{9}{10.8}{\rmdefault}{\mddefault}{\updefault}{\color[rgb]{0,0,0}$q_2$}%
}}}}
\put(2356,-7056){\makebox(0,0)[lb]{\smash{{\SetFigFont{9}{10.8}{\rmdefault}{\mddefault}{\updefault}{\color[rgb]{0,0,0}$q_4$}%
}}}}
\end{picture}%

%% file: bisector-event.pspdftex
\begin{picture}(0,0)%
\includegraphics{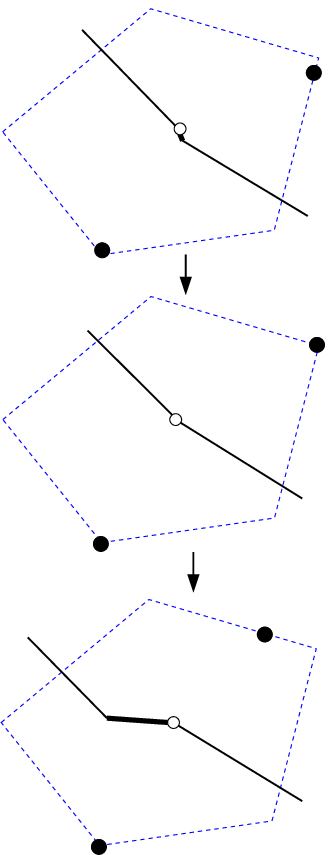}%
\end{picture}%
\setlength{\unitlength}{2072sp}%
\begingroup\makeatletter\ifx\SetFigFont\undefined%
\gdef\SetFigFont#1#2#3#4#5{%
  \reset@font\fontsize{#1}{#2pt}%
  \fontfamily{#3}\fontseries{#4}\fontshape{#5}%
  \selectfont}%
\fi\endgroup%
\begin{picture}(3069,7777)(867,-7343)
\put(2106,-169){\makebox(0,0)[lb]{\smash{{\SetFigFont{10}{12.0}{\rmdefault}{\mddefault}{\updefault}{\color[rgb]{0,0,0}(1,5)}%
}}}}
\put(2676,-754){\makebox(0,0)[lb]{\smash{{\SetFigFont{10}{12.0}{\rmdefault}{\mddefault}{\updefault}{\color[rgb]{0,0,0}(2,5)}%
}}}}
\put(3291,-1144){\makebox(0,0)[lb]{\smash{{\SetFigFont{10}{12.0}{\rmdefault}{\mddefault}{\updefault}{\color[rgb]{0,0,0}(2,4)}%
}}}}
\put(3051,-3649){\makebox(0,0)[lb]{\smash{{\SetFigFont{10}{12.0}{\rmdefault}{\mddefault}{\updefault}{\color[rgb]{0,0,0}(2,4)}%
}}}}
\put(2031,-2809){\makebox(0,0)[lb]{\smash{{\SetFigFont{10}{12.0}{\rmdefault}{\mddefault}{\updefault}{\color[rgb]{0,0,0}(1,5)}%
}}}}
\put(1251,-5404){\makebox(0,0)[lb]{\smash{{\SetFigFont{10}{12.0}{\rmdefault}{\mddefault}{\updefault}{\color[rgb]{0,0,0}(1,5)}%
}}}}
\put(3126,-6484){\makebox(0,0)[lb]{\smash{{\SetFigFont{10}{12.0}{\rmdefault}{\mddefault}{\updefault}{\color[rgb]{0,0,0}(2,4)}%
}}}}
\put(1941,-5944){\makebox(0,0)[lb]{\smash{{\SetFigFont{10}{12.0}{\rmdefault}{\mddefault}{\updefault}{\color[rgb]{0,0,0}(1,4)}%
}}}}
\put(3691,-781){\makebox(0,0)[lb]{\smash{{\SetFigFont{10}{12.0}{\rmdefault}{\mddefault}{\updefault}{\color[rgb]{0,0,0}5}%
}}}}
\put(2926,209){\makebox(0,0)[lb]{\smash{{\SetFigFont{10}{12.0}{\rmdefault}{\mddefault}{\updefault}{\color[rgb]{0,0,0}4}%
}}}}
\put(1396,-151){\makebox(0,0)[lb]{\smash{{\SetFigFont{10}{12.0}{\rmdefault}{\mddefault}{\updefault}{\color[rgb]{0,0,0}3}%
}}}}
\put(1126,-1456){\makebox(0,0)[lb]{\smash{{\SetFigFont{10}{12.0}{\rmdefault}{\mddefault}{\updefault}{\color[rgb]{0,0,0}2}%
}}}}
\put(2791,-1951){\makebox(0,0)[lb]{\smash{{\SetFigFont{10}{12.0}{\rmdefault}{\mddefault}{\updefault}{\color[rgb]{0,0,0}1}%
}}}}
\put(1839,-1591){\makebox(0,0)[lb]{\smash{{\SetFigFont{10}{12.0}{\rmdefault}{\mddefault}{\updefault}{\color[rgb]{0,0,0}$p$}%
}}}}
\put(3564,-226){\makebox(0,0)[rb]{\smash{{\SetFigFont{10}{12.0}{\rmdefault}{\mddefault}{\updefault}{\color[rgb]{0,0,0}$q$}%
}}}}
\put(1749,-4336){\makebox(0,0)[lb]{\smash{{\SetFigFont{10}{12.0}{\rmdefault}{\mddefault}{\updefault}{\color[rgb]{0,0,0}$p$}%
}}}}
\put(1704,-7126){\makebox(0,0)[lb]{\smash{{\SetFigFont{10}{12.0}{\rmdefault}{\mddefault}{\updefault}{\color[rgb]{0,0,0}$p$}%
}}}}
\put(3639,-2911){\makebox(0,0)[rb]{\smash{{\SetFigFont{10}{12.0}{\rmdefault}{\mddefault}{\updefault}{\color[rgb]{0,0,0}$q$}%
}}}}
\put(3174,-5611){\makebox(0,0)[rb]{\smash{{\SetFigFont{10}{12.0}{\rmdefault}{\mddefault}{\updefault}{\color[rgb]{0,0,0}$q$}%
}}}}
\end{picture}%

%% file: corner-event.pspdftex
\begin{picture}(0,0)%
\includegraphics{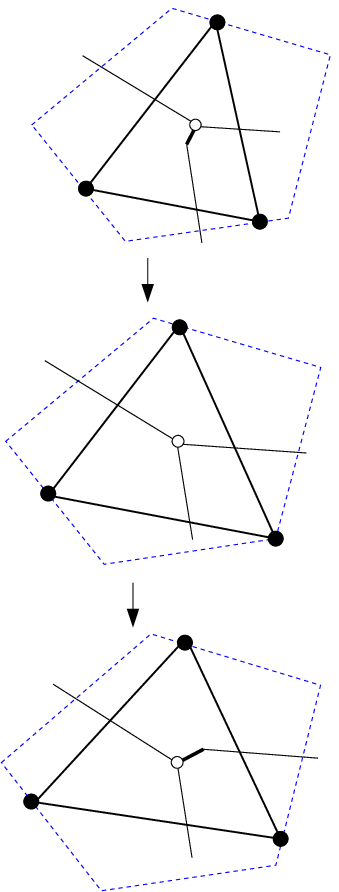}%
\end{picture}%
\setlength{\unitlength}{2072sp}%
\begingroup\makeatletter\ifx\SetFigFont\undefined%
\gdef\SetFigFont#1#2#3#4#5{%
  \reset@font\fontsize{#1}{#2pt}%
  \fontfamily{#3}\fontseries{#4}\fontshape{#5}%
  \selectfont}%
\fi\endgroup%
\begin{picture}(3032,8259)(3151,-8378)
\put(4841,-7462){\makebox(0,0)[lb]{\smash{{\SetFigFont{10}{12.0}{\rmdefault}{\mddefault}{\updefault}{\color[rgb]{0,0,0}$c_{pqr}$}%
}}}}
\put(5421,-7012){\makebox(0,0)[lb]{\smash{{\SetFigFont{10}{12.0}{\rmdefault}{\mddefault}{\updefault}{\color[rgb]{0,0,0}$e_{pq}$}%
}}}}
\put(5326,-4242){\makebox(0,0)[lb]{\smash{{\SetFigFont{10}{12.0}{\rmdefault}{\mddefault}{\updefault}{\color[rgb]{0,0,0}$e_{pq}$}%
}}}}
\put(4846,-4582){\makebox(0,0)[lb]{\smash{{\SetFigFont{10}{12.0}{\rmdefault}{\mddefault}{\updefault}{\color[rgb]{0,0,0}$c_{pqr}$}%
}}}}
\put(5436,-1307){\makebox(0,0)[lb]{\smash{{\SetFigFont{10}{12.0}{\rmdefault}{\mddefault}{\updefault}{\color[rgb]{0,0,0}$e_{pq}$}%
}}}}
\put(4996,-1607){\makebox(0,0)[lb]{\smash{{\SetFigFont{10}{12.0}{\rmdefault}{\mddefault}{\updefault}{\color[rgb]{0,0,0}$c_{pqr}$}%
}}}}
\put(5816,-8023){\makebox(0,0)[lb]{\smash{{\SetFigFont{10}{12.0}{\rmdefault}{\mddefault}{\updefault}{\color[rgb]{0,0,0}$p$}%
}}}}
\put(3356,-7833){\makebox(0,0)[rb]{\smash{{\SetFigFont{10}{12.0}{\rmdefault}{\mddefault}{\updefault}{\color[rgb]{0,0,0}$r$}%
}}}}
\put(5811,-5233){\makebox(0,0)[lb]{\smash{{\SetFigFont{10}{12.0}{\rmdefault}{\mddefault}{\updefault}{\color[rgb]{0,0,0}$p$}%
}}}}
\put(4891,-3153){\makebox(0,0)[lb]{\smash{{\SetFigFont{10}{12.0}{\rmdefault}{\mddefault}{\updefault}{\color[rgb]{0,0,0}$q$}%
}}}}
\put(3461,-4903){\makebox(0,0)[rb]{\smash{{\SetFigFont{10}{12.0}{\rmdefault}{\mddefault}{\updefault}{\color[rgb]{0,0,0}$r$}%
}}}}
\put(3781,-1998){\makebox(0,0)[rb]{\smash{{\SetFigFont{10}{12.0}{\rmdefault}{\mddefault}{\updefault}{\color[rgb]{0,0,0}$r$}%
}}}}
\put(5251,-338){\makebox(0,0)[lb]{\smash{{\SetFigFont{10}{12.0}{\rmdefault}{\mddefault}{\updefault}{\color[rgb]{0,0,0}$q$}%
}}}}
\put(5621,-2548){\makebox(0,0)[lb]{\smash{{\SetFigFont{10}{12.0}{\rmdefault}{\mddefault}{\updefault}{\color[rgb]{0,0,0}$p$}%
}}}}
\put(4906,-5853){\makebox(0,0)[lb]{\smash{{\SetFigFont{10}{12.0}{\rmdefault}{\mddefault}{\updefault}{\color[rgb]{0,0,0}$q$}%
}}}}
\end{picture}%

%% file: generic-flip.pspdftex
\begin{picture}(0,0)%
\includegraphics{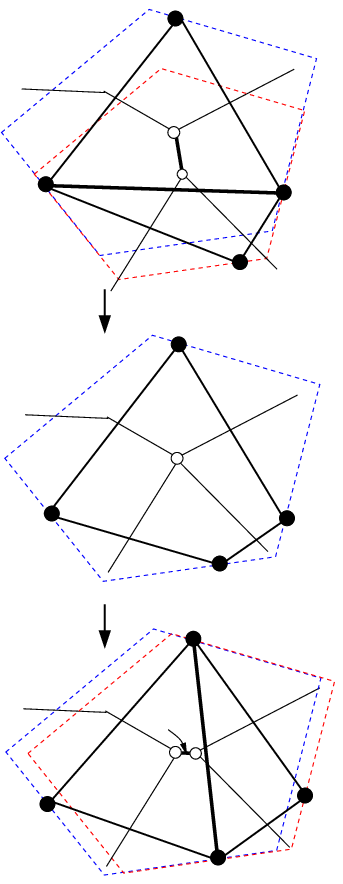}%
\end{picture}%
\setlength{\unitlength}{2072sp}%
\begingroup\makeatletter\ifx\SetFigFont\undefined%
\gdef\SetFigFont#1#2#3#4#5{%
  \reset@font\fontsize{#1}{#2pt}%
  \fontfamily{#3}\fontseries{#4}\fontshape{#5}%
  \selectfont}%
\fi\endgroup%
\begin{picture}(3069,8309)(4895,-8380)
\put(6626,-1607){\makebox(0,0)[lb]{\smash{{\SetFigFont{10}{12.0}{\rmdefault}{\mddefault}{\updefault}{\color[rgb]{0,0,0}$e_{pq}$}%
}}}}
\put(6291,-6782){\makebox(0,0)[lb]{\smash{{\SetFigFont{10}{12.0}{\rmdefault}{\mddefault}{\updefault}{\color[rgb]{0,0,0}$e_{wr}$}%
}}}}
\put(7181,-2838){\makebox(0,0)[lb]{\smash{{\SetFigFont{10}{12.0}{\rmdefault}{\mddefault}{\updefault}{\color[rgb]{0,0,0}$r$}%
}}}}
\put(7641,-1948){\makebox(0,0)[lb]{\smash{{\SetFigFont{10}{12.0}{\rmdefault}{\mddefault}{\updefault}{\color[rgb]{0,0,0}$p$}%
}}}}
\put(6631,-278){\makebox(0,0)[lb]{\smash{{\SetFigFont{10}{12.0}{\rmdefault}{\mddefault}{\updefault}{\color[rgb]{0,0,0}$w$}%
}}}}
\put(5131,-1938){\makebox(0,0)[rb]{\smash{{\SetFigFont{10}{12.0}{\rmdefault}{\mddefault}{\updefault}{\color[rgb]{0,0,0}$q$}%
}}}}
\put(5301,-5118){\makebox(0,0)[rb]{\smash{{\SetFigFont{10}{12.0}{\rmdefault}{\mddefault}{\updefault}{\color[rgb]{0,0,0}$q$}%
}}}}
\put(7621,-5008){\makebox(0,0)[lb]{\smash{{\SetFigFont{10}{12.0}{\rmdefault}{\mddefault}{\updefault}{\color[rgb]{0,0,0}$p$}%
}}}}
\put(6991,-5578){\makebox(0,0)[lb]{\smash{{\SetFigFont{10}{12.0}{\rmdefault}{\mddefault}{\updefault}{\color[rgb]{0,0,0}$r$}%
}}}}
\put(6621,-3268){\makebox(0,0)[lb]{\smash{{\SetFigFont{10}{12.0}{\rmdefault}{\mddefault}{\updefault}{\color[rgb]{0,0,0}$w$}%
}}}}
\put(5216,-7693){\makebox(0,0)[rb]{\smash{{\SetFigFont{10}{12.0}{\rmdefault}{\mddefault}{\updefault}{\color[rgb]{0,0,0}$q$}%
}}}}
\put(6936,-8293){\makebox(0,0)[lb]{\smash{{\SetFigFont{10}{12.0}{\rmdefault}{\mddefault}{\updefault}{\color[rgb]{0,0,0}$r$}%
}}}}
\put(7776,-7433){\makebox(0,0)[lb]{\smash{{\SetFigFont{10}{12.0}{\rmdefault}{\mddefault}{\updefault}{\color[rgb]{0,0,0}$p$}%
}}}}
\put(6526,-5866){\makebox(0,0)[lb]{\smash{{\SetFigFont{10}{12.0}{\rmdefault}{\mddefault}{\updefault}{\color[rgb]{0,0,0}$w$}%
}}}}
\end{picture}%

%% file: degenerate-bisector-new.pspdftex
\begin{picture}(0,0)%
\includegraphics{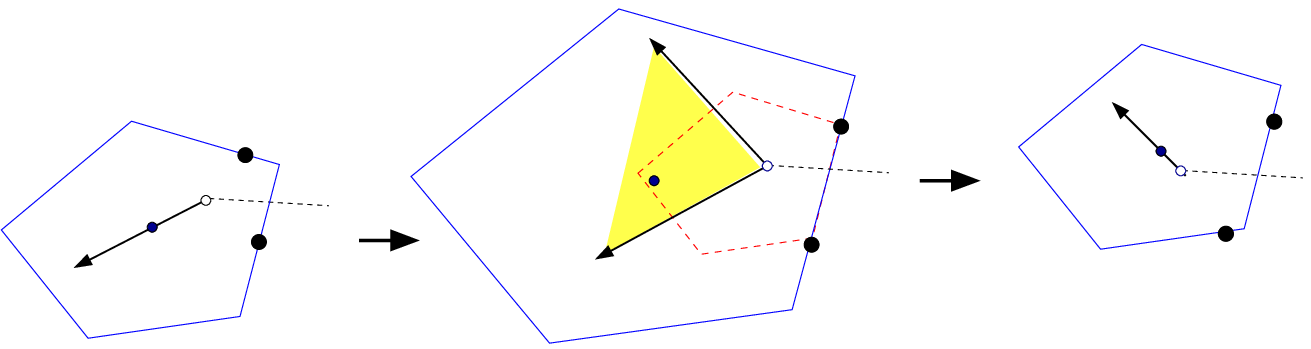}%
\end{picture}%
\setlength{\unitlength}{2072sp}%
\begingroup\makeatletter\ifx\SetFigFont\undefined%
\gdef\SetFigFont#1#2#3#4#5{%
  \reset@font\fontsize{#1}{#2pt}%
  \fontfamily{#3}\fontseries{#4}\fontshape{#5}%
  \selectfont}%
\fi\endgroup%
\begin{picture}(11921,3122)(1176,-3626)
\put(7246,-736){\makebox(0,0)[b]{\smash{{\SetFigFont{10}{12.0}{\rmdefault}{\mddefault}{\updefault}{\color[rgb]{0,0,0}$\rho^+$}%
}}}}
\put(6706,-3031){\makebox(0,0)[b]{\smash{{\SetFigFont{10}{12.0}{\rmdefault}{\mddefault}{\updefault}{\color[rgb]{0,0,0}$\rho^-$}%
}}}}
\put(8731,-2806){\makebox(0,0)[lb]{\smash{{\SetFigFont{10}{12.0}{\rmdefault}{\mddefault}{\updefault}{\color[rgb]{0,0,0}$q$}%
}}}}
\put(9001,-1771){\makebox(0,0)[lb]{\smash{{\SetFigFont{10}{12.0}{\rmdefault}{\mddefault}{\updefault}{\color[rgb]{0,0,0}$p$}%
}}}}
\put(8416,-3526){\makebox(0,0)[lb]{\smash{{\SetFigFont{10}{12.0}{\rmdefault}{\mddefault}{\updefault}{\color[rgb]{0,0,.82}$v_{i+1}$}%
}}}}
\put(8105,-2276){\makebox(0,0)[lb]{\smash{{\SetFigFont{10}{12.0}{\rmdefault}{\mddefault}{\updefault}{\color[rgb]{1,0,0}$\zeta$}%
}}}}
\put(3736,-2848){\makebox(0,0)[lb]{\smash{{\SetFigFont{10}{12.0}{\rmdefault}{\mddefault}{\updefault}{\color[rgb]{0,0,0}$q$}%
}}}}
\put(3478,-1738){\makebox(0,0)[lb]{\smash{{\SetFigFont{10}{12.0}{\rmdefault}{\mddefault}{\updefault}{\color[rgb]{0,0,0}$p$}%
}}}}
\put(3442,-3446){\makebox(0,0)[lb]{\smash{{\SetFigFont{10}{12.0}{\rmdefault}{\mddefault}{\updefault}{\color[rgb]{0,0,.82}$v_{i+1}$}%
}}}}
\put(3761,-1997){\makebox(0,0)[lb]{\smash{{\SetFigFont{10}{12.0}{\rmdefault}{\mddefault}{\updefault}{\color[rgb]{0,0,.82}$v_{i}$}%
}}}}
\put(12916,-1726){\makebox(0,0)[lb]{\smash{{\SetFigFont{10}{12.0}{\rmdefault}{\mddefault}{\updefault}{\color[rgb]{0,0,0}$p$}%
}}}}
\put(12916,-1276){\makebox(0,0)[lb]{\smash{{\SetFigFont{10}{12.0}{\rmdefault}{\mddefault}{\updefault}{\color[rgb]{0,0,.82}$v_{i}$}%
}}}}
\put(12601,-2626){\makebox(0,0)[lb]{\smash{{\SetFigFont{10}{12.0}{\rmdefault}{\mddefault}{\updefault}{\color[rgb]{0,0,.82}$v_{i+1}$}%
}}}}
\put(9046,-1186){\makebox(0,0)[lb]{\smash{{\SetFigFont{10}{12.0}{\rmdefault}{\mddefault}{\updefault}{\color[rgb]{0,0,.82}$v_{i}$}%
}}}}
\put(12241,-2896){\makebox(0,0)[lb]{\smash{{\SetFigFont{10}{12.0}{\rmdefault}{\mddefault}{\updefault}{\color[rgb]{0,0,0}$q$}%
}}}}
\end{picture}%

%% file: degenerateVorBefore.pspdftex
\begin{picture}(0,0)%
\includegraphics{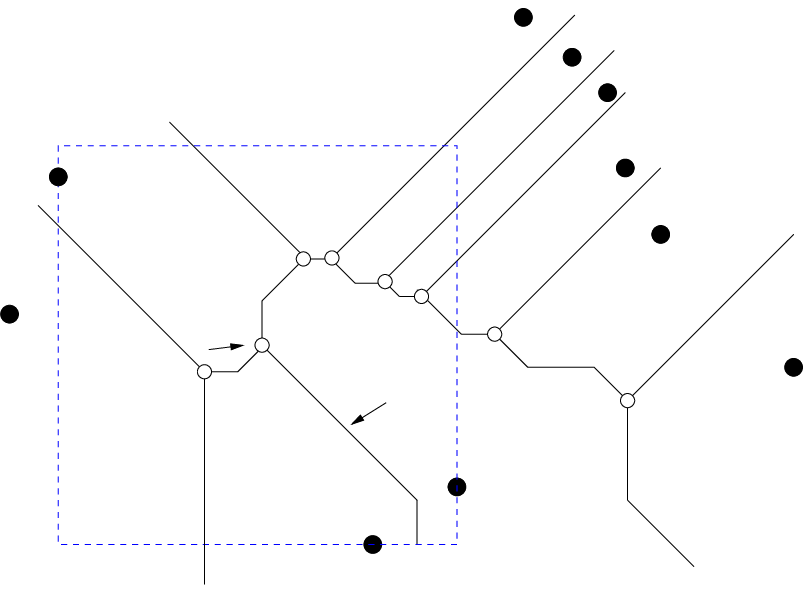}%
\end{picture}%
\setlength{\unitlength}{1865sp}%
\begingroup\makeatletter\ifx\SetFigFont\undefined%
\gdef\SetFigFont#1#2#3#4#5{%
  \reset@font\fontsize{#1}{#2pt}%
  \fontfamily{#3}\fontseries{#4}\fontshape{#5}%
  \selectfont}%
\fi\endgroup%
\begin{picture}(8161,6631)(1478,-6400)
\put(5274,-6039){\makebox(0,0)[b]{\smash{{\SetFigFont{11}{13.2}{\rmdefault}{\mddefault}{\updefault}{\color[rgb]{0,0,0}$q$}%
}}}}
\put(7034, 12){\makebox(0,0)[b]{\smash{{\SetFigFont{11}{13.2}{\rmdefault}{\mddefault}{\updefault}{\color[rgb]{0,0,0}$r_1$}%
}}}}
\put(7699,-455){\makebox(0,0)[b]{\smash{{\SetFigFont{11}{13.2}{\rmdefault}{\mddefault}{\updefault}{\color[rgb]{0,0,0}$r_2$}%
}}}}
\put(8097,-823){\makebox(0,0)[b]{\smash{{\SetFigFont{11}{13.2}{\rmdefault}{\mddefault}{\updefault}{\color[rgb]{0,0,0}$r_3$}%
}}}}
\put(8256,-1518){\makebox(0,0)[b]{\smash{{\SetFigFont{11}{13.2}{\rmdefault}{\mddefault}{\updefault}{\color[rgb]{0,0,0}$r_4$}%
}}}}
\put(8634,-2234){\makebox(0,0)[b]{\smash{{\SetFigFont{11}{13.2}{\rmdefault}{\mddefault}{\updefault}{\color[rgb]{0,0,0}$r^+=r_5$}%
}}}}
\put(4771,-4021){\makebox(0,0)[lb]{\smash{{\SetFigFont{10}{12.0}{\rmdefault}{\mddefault}{\updefault}{\color[rgb]{1,0,0}$\rho^-(t_0^-)\cap e_{pq}$}%
}}}}
\put(6751,-4651){\makebox(0,0)[b]{\smash{{\SetFigFont{10}{12.0}{\rmdefault}{\mddefault}{\updefault}{\color[rgb]{1,0,0}$\Vor(p)$}%
}}}}
\put(4456,-2266){\makebox(0,0)[lb]{\smash{{\SetFigFont{10}{12.0}{\rmdefault}{\mddefault}{\updefault}{\color[rgb]{1,0,0}$e_{pr_1}$}%
}}}}
\put(5941,-3121){\makebox(0,0)[lb]{\smash{{\SetFigFont{10}{12.0}{\rmdefault}{\mddefault}{\updefault}{\color[rgb]{1,0,0}$e_{pr_4}$}%
}}}}
\put(7066,-3571){\makebox(0,0)[lb]{\smash{{\SetFigFont{10}{12.0}{\rmdefault}{\mddefault}{\updefault}{\color[rgb]{1,0,0}$e_{pr_5}$}%
}}}}
\put(4051,-3391){\makebox(0,0)[rb]{\smash{{\SetFigFont{10}{12.0}{\rmdefault}{\mddefault}{\updefault}{\color[rgb]{1,0,0}$\eta^-(t_0^-)$}%
}}}}
\put(3556,-2761){\makebox(0,0)[lb]{\smash{{\SetFigFont{10}{12.0}{\rmdefault}{\mddefault}{\updefault}{\color[rgb]{1,0,0}$e_{pr_0}$}%
}}}}
\put(2431,-1051){\makebox(0,0)[b]{\smash{{\SetFigFont{11}{13.2}{\rmdefault}{\mddefault}{\updefault}{\color[rgb]{0,0,0}$r^-=r_0$}%
}}}}
\put(6301,-5371){\makebox(0,0)[b]{\smash{{\SetFigFont{11}{13.2}{\rmdefault}{\mddefault}{\updefault}{\color[rgb]{0,0,0}$p$}%
}}}}
\put(4321,-3346){\makebox(0,0)[b]{\smash{{\SetFigFont{8}{9.6}{\rmdefault}{\mddefault}{\updefault}{\color[rgb]{0,0,0}$g$}%
}}}}
\put(3531,-6318){\makebox(0,0)[b]{\smash{{\SetFigFont{10}{12.0}{\rmdefault}{\mddefault}{\updefault}{\color[rgb]{0,0,0}(a): $t=t_0^-$}%
}}}}
\put(4006,-5101){\makebox(0,0)[lb]{\smash{{\SetFigFont{10}{12.0}{\rmdefault}{\mddefault}{\updefault}{\color[rgb]{1,0,0}$\Vor(q)$}%
}}}}
\end{picture}%

%% file: degenerateVorAfter.pspdftex
\begin{picture}(0,0)%
\includegraphics{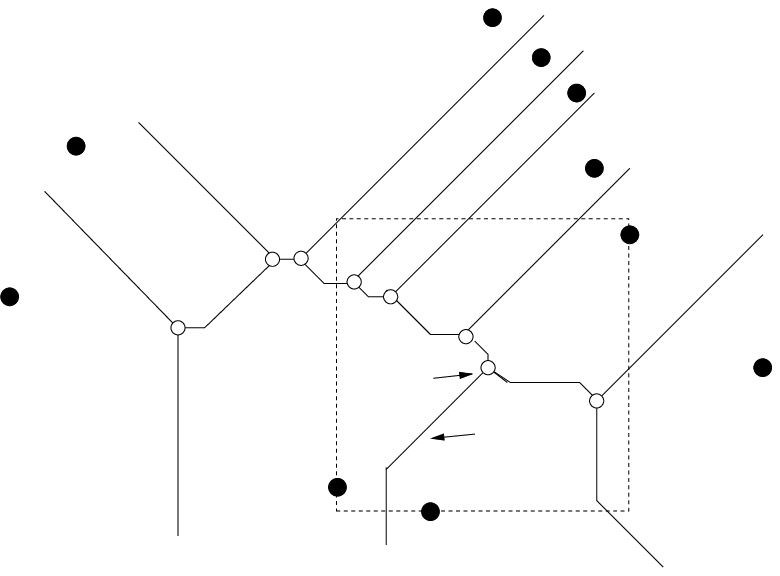}%
\end{picture}%
\setlength{\unitlength}{1865sp}%
\begingroup\makeatletter\ifx\SetFigFont\undefined%
\gdef\SetFigFont#1#2#3#4#5{%
  \reset@font\fontsize{#1}{#2pt}%
  \fontfamily{#3}\fontseries{#4}\fontshape{#5}%
  \selectfont}%
\fi\endgroup%
\begin{picture}(7846,6327)(1973,-6038)
\put(6383,-5674){\makebox(0,0)[b]{\smash{{\SetFigFont{10}{12.0}{\rmdefault}{\mddefault}{\updefault}{\color[rgb]{0,0,0}$p$}%
}}}}
\put(7268, 94){\makebox(0,0)[b]{\smash{{\SetFigFont{10}{12.0}{\rmdefault}{\mddefault}{\updefault}{\color[rgb]{0,0,0}$r_1$}%
}}}}
\put(7895,-413){\makebox(0,0)[b]{\smash{{\SetFigFont{10}{12.0}{\rmdefault}{\mddefault}{\updefault}{\color[rgb]{0,0,0}$r_2$}%
}}}}
\put(8274,-781){\makebox(0,0)[b]{\smash{{\SetFigFont{10}{12.0}{\rmdefault}{\mddefault}{\updefault}{\color[rgb]{0,0,0}$r_3$}%
}}}}
\put(8462,-1537){\makebox(0,0)[b]{\smash{{\SetFigFont{10}{12.0}{\rmdefault}{\mddefault}{\updefault}{\color[rgb]{0,0,0}$r_4$}%
}}}}
\put(5311,-5461){\makebox(0,0)[b]{\smash{{\SetFigFont{10}{12.0}{\rmdefault}{\mddefault}{\updefault}{\color[rgb]{0,0,0}$q$}%
}}}}
\put(6121,-3076){\makebox(0,0)[lb]{\smash{{\SetFigFont{10}{12.0}{\rmdefault}{\mddefault}{\updefault}{\color[rgb]{0,0,0}$e_{qr_4}$}%
}}}}
\put(7336,-5056){\makebox(0,0)[b]{\smash{{\SetFigFont{10}{12.0}{\rmdefault}{\mddefault}{\updefault}{\color[rgb]{0,0,0}$\Vor(p)$}%
}}}}
\put(4591,-2221){\makebox(0,0)[lb]{\smash{{\SetFigFont{10}{12.0}{\rmdefault}{\mddefault}{\updefault}{\color[rgb]{0,0,0}$e_{qr_1}$}%
}}}}
\put(8956,-2221){\makebox(0,0)[b]{\smash{{\SetFigFont{10}{12.0}{\rmdefault}{\mddefault}{\updefault}{\color[rgb]{0,0,0}$r^+=r_5$}%
}}}}
\put(6841,-4426){\makebox(0,0)[lb]{\smash{{\SetFigFont{10}{12.0}{\rmdefault}{\mddefault}{\updefault}{\color[rgb]{0,0,0}$\rho^+(t_0^+)\cap e_{pq}$}%
}}}}
\put(2521,-1186){\makebox(0,0)[b]{\smash{{\SetFigFont{10}{12.0}{\rmdefault}{\mddefault}{\updefault}{\color[rgb]{0,0,0}$r^-=r_0$}%
}}}}
\put(7006,-3526){\makebox(0,0)[lb]{\smash{{\SetFigFont{10}{12.0}{\rmdefault}{\mddefault}{\updefault}{\color[rgb]{0,0,0}$e_{qr_5}$}%
}}}}
\put(6346,-3946){\makebox(0,0)[rb]{\smash{{\SetFigFont{10}{12.0}{\rmdefault}{\mddefault}{\updefault}{\color[rgb]{0,0,0}$\eta^+(t_0^+)$}%
}}}}
\put(3976,-4411){\makebox(0,0)[lb]{\smash{{\SetFigFont{10}{12.0}{\rmdefault}{\mddefault}{\updefault}{\color[rgb]{0,0,0}$\Vor(q)$}%
}}}}
\put(6751,-3691){\makebox(0,0)[b]{\smash{{\SetFigFont{10}{12.0}{\rmdefault}{\mddefault}{\updefault}{\color[rgb]{0,0,0}$h$}%
}}}}
\put(3691,-2851){\makebox(0,0)[lb]{\smash{{\SetFigFont{10}{12.0}{\rmdefault}{\mddefault}{\updefault}{\color[rgb]{0,0,0}$e_{qr_0}$}%
}}}}
\put(3871,-5956){\makebox(0,0)[b]{\smash{{\SetFigFont{10}{12.0}{\rmdefault}{\mddefault}{\updefault}{\color[rgb]{0,0,0}(c): $t=t_0^+$}%
}}}}
\end{picture}%

%% file: arrows.pspdftex
\begin{picture}(0,0)%
\includegraphics{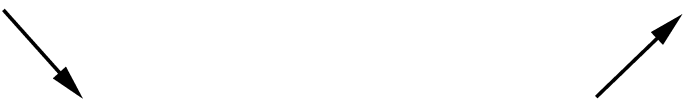}%
\end{picture}%
\setlength{\unitlength}{2072sp}%
\begingroup\makeatletter\ifx\SetFigFont\undefined%
\gdef\SetFigFont#1#2#3#4#5{%
  \reset@font\fontsize{#1}{#2pt}%
  \fontfamily{#3}\fontseries{#4}\fontshape{#5}%
  \selectfont}%
\fi\endgroup%
\begin{picture}(6272,876)(3778,-7341)
\end{picture}%

%% file: degenerateVort0.pspdftex
\begin{picture}(0,0)%
\includegraphics{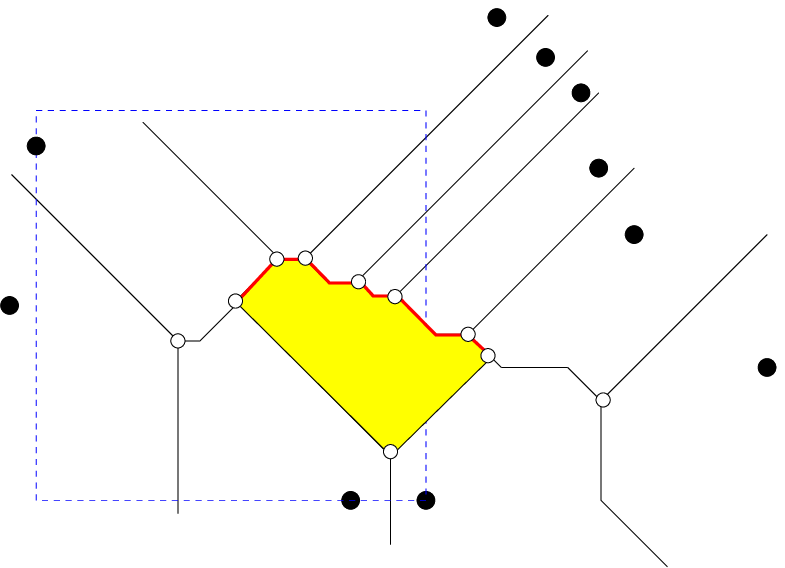}%
\end{picture}%
\setlength{\unitlength}{1865sp}%
\begingroup\makeatletter\ifx\SetFigFont\undefined%
\gdef\SetFigFont#1#2#3#4#5{%
  \reset@font\fontsize{#1}{#2pt}%
  \fontfamily{#3}\fontseries{#4}\fontshape{#5}%
  \selectfont}%
\fi\endgroup%
\begin{picture}(7891,6019)(1838,-5743)
\put(7124, 57){\makebox(0,0)[b]{\smash{{\SetFigFont{11}{13.2}{\rmdefault}{\mddefault}{\updefault}{\color[rgb]{0,0,0}$r_1$}%
}}}}
\put(7789,-410){\makebox(0,0)[b]{\smash{{\SetFigFont{11}{13.2}{\rmdefault}{\mddefault}{\updefault}{\color[rgb]{0,0,0}$r_2$}%
}}}}
\put(8187,-778){\makebox(0,0)[b]{\smash{{\SetFigFont{11}{13.2}{\rmdefault}{\mddefault}{\updefault}{\color[rgb]{0,0,0}$r_3$}%
}}}}
\put(8346,-1473){\makebox(0,0)[b]{\smash{{\SetFigFont{11}{13.2}{\rmdefault}{\mddefault}{\updefault}{\color[rgb]{0,0,0}$r_4$}%
}}}}
\put(8724,-2189){\makebox(0,0)[b]{\smash{{\SetFigFont{11}{13.2}{\rmdefault}{\mddefault}{\updefault}{\color[rgb]{0,0,0}$r^+=r_5$}%
}}}}
\put(5356,-5551){\makebox(0,0)[b]{\smash{{\SetFigFont{11}{13.2}{\rmdefault}{\mddefault}{\updefault}{\color[rgb]{0,0,0}$q$}%
}}}}
\put(5504,-3768){\makebox(0,0)[lb]{\smash{{\SetFigFont{10}{12.0}{\rmdefault}{\mddefault}{\updefault}{\color[rgb]{0,0,0}$W$}%
}}}}
\put(7276,-3466){\makebox(0,0)[rb]{\smash{{\SetFigFont{10}{12.0}{\rmdefault}{\mddefault}{\updefault}{\color[rgb]{1,0,0}$\eta^+$}%
}}}}
\put(6271,-5551){\makebox(0,0)[b]{\smash{{\SetFigFont{11}{13.2}{\rmdefault}{\mddefault}{\updefault}{\color[rgb]{0,0,0}$p$}%
}}}}
\put(4426,-2941){\makebox(0,0)[rb]{\smash{{\SetFigFont{10}{12.0}{\rmdefault}{\mddefault}{\updefault}{\color[rgb]{1,0,0}$\eta^-$}%
}}}}
\put(4561,-4156){\makebox(0,0)[b]{\smash{{\SetFigFont{11}{13.2}{\rmdefault}{\mddefault}{\updefault}{\color[rgb]{0,0,0}$\rho^-\cap e_{pq}$}%
}}}}
\put(7066,-4216){\makebox(0,0)[b]{\smash{{\SetFigFont{11}{13.2}{\rmdefault}{\mddefault}{\updefault}{\color[rgb]{0,0,0}$\rho^+\cap e_{pq}$}%
}}}}
\put(2521,-1141){\makebox(0,0)[b]{\smash{{\SetFigFont{11}{13.2}{\rmdefault}{\mddefault}{\updefault}{\color[rgb]{0,0,0}$r^-=r_0$}%
}}}}
\put(6076,-4741){\makebox(0,0)[b]{\smash{{\SetFigFont{11}{13.2}{\rmdefault}{\mddefault}{\updefault}{\color[rgb]{0,0,0}$\zeta$}%
}}}}
\put(4876,-1441){\makebox(0,0)[b]{\smash{{\SetFigFont{11}{13.2}{\rmdefault}{\mddefault}{\updefault}{\color[rgb]{0,0,1}$Q[\eta^-]$}%
}}}}
\put(5061,-318){\makebox(0,0)[b]{\smash{{\SetFigFont{10}{12.0}{\rmdefault}{\mddefault}{\updefault}{\color[rgb]{0,0,0}(b): $t=t_0$}%
}}}}
\end{picture}%

%% file: degenerateDelBefore-new.pspdftex
\begin{picture}(0,0)%
\includegraphics{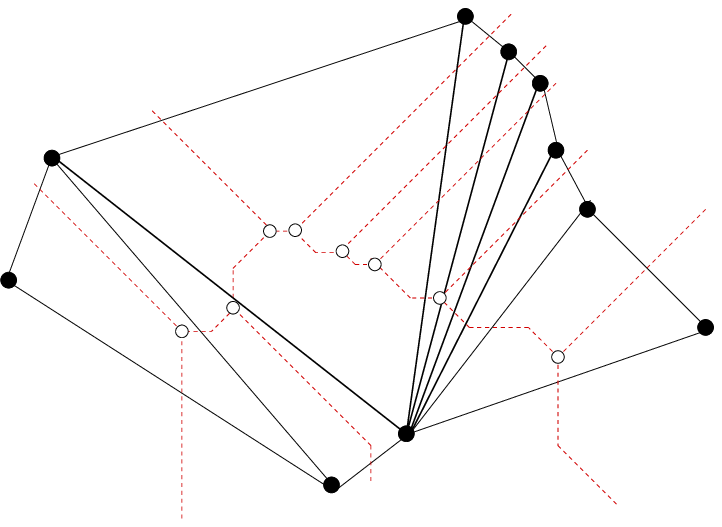}%
\end{picture}%
\setlength{\unitlength}{1657sp}%
\begingroup\makeatletter\ifx\SetFigFont\undefined%
\gdef\SetFigFont#1#2#3#4#5{%
  \reset@font\fontsize{#1}{#2pt}%
  \fontfamily{#3}\fontseries{#4}\fontshape{#5}%
  \selectfont}%
\fi\endgroup%
\begin{picture}(8161,6474)(1478,-6171)
\put(5274,-6039){\makebox(0,0)[b]{\smash{{\SetFigFont{10}{12.0}{\rmdefault}{\mddefault}{\updefault}{\color[rgb]{0,0,0}$q$}%
}}}}
\put(7034, 12){\makebox(0,0)[b]{\smash{{\SetFigFont{10}{12.0}{\rmdefault}{\mddefault}{\updefault}{\color[rgb]{0,0,0}$r_1$}%
}}}}
\put(7699,-455){\makebox(0,0)[b]{\smash{{\SetFigFont{10}{12.0}{\rmdefault}{\mddefault}{\updefault}{\color[rgb]{0,0,0}$r_2$}%
}}}}
\put(8097,-823){\makebox(0,0)[b]{\smash{{\SetFigFont{10}{12.0}{\rmdefault}{\mddefault}{\updefault}{\color[rgb]{0,0,0}$r_3$}%
}}}}
\put(8256,-1518){\makebox(0,0)[b]{\smash{{\SetFigFont{10}{12.0}{\rmdefault}{\mddefault}{\updefault}{\color[rgb]{0,0,0}$r_4$}%
}}}}
\put(8634,-2234){\makebox(0,0)[b]{\smash{{\SetFigFont{10}{12.0}{\rmdefault}{\mddefault}{\updefault}{\color[rgb]{0,0,0}$r^+=r_5$}%
}}}}
\put(2431,-1051){\makebox(0,0)[b]{\smash{{\SetFigFont{10}{12.0}{\rmdefault}{\mddefault}{\updefault}{\color[rgb]{0,0,0}$r^-=r_0$}%
}}}}
\put(6354,-5618){\makebox(0,0)[b]{\smash{{\SetFigFont{10}{12.0}{\rmdefault}{\mddefault}{\updefault}{\color[rgb]{0,0,0}$p$}%
}}}}
\end{picture}%

%% file: degenerateDelAfter-new.pspdftex
\begin{picture}(0,0)%
\includegraphics{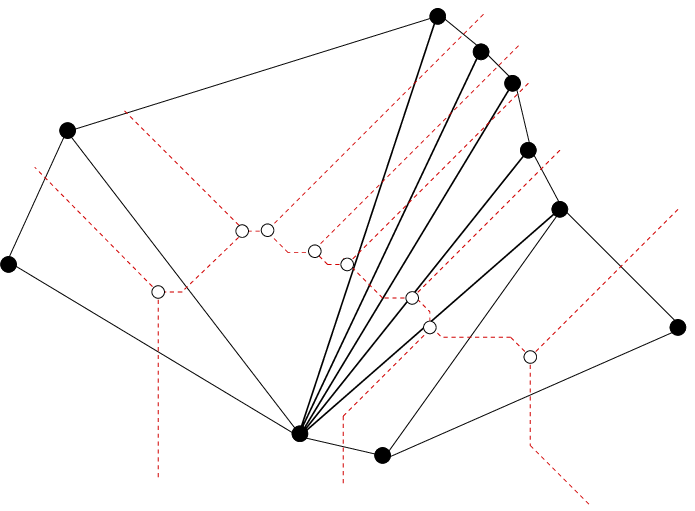}%
\end{picture}%
\setlength{\unitlength}{1657sp}%
\begingroup\makeatletter\ifx\SetFigFont\undefined%
\gdef\SetFigFont#1#2#3#4#5{%
  \reset@font\fontsize{#1}{#2pt}%
  \fontfamily{#3}\fontseries{#4}\fontshape{#5}%
  \selectfont}%
\fi\endgroup%
\begin{picture}(7846,6191)(1973,-5806)
\put(6383,-5674){\makebox(0,0)[b]{\smash{{\SetFigFont{10}{12.0}{\rmdefault}{\mddefault}{\updefault}{\color[rgb]{0,0,0}$p$}%
}}}}
\put(7268, 94){\makebox(0,0)[b]{\smash{{\SetFigFont{10}{12.0}{\rmdefault}{\mddefault}{\updefault}{\color[rgb]{0,0,0}$r_1$}%
}}}}
\put(7895,-413){\makebox(0,0)[b]{\smash{{\SetFigFont{10}{12.0}{\rmdefault}{\mddefault}{\updefault}{\color[rgb]{0,0,0}$r_2$}%
}}}}
\put(8274,-781){\makebox(0,0)[b]{\smash{{\SetFigFont{10}{12.0}{\rmdefault}{\mddefault}{\updefault}{\color[rgb]{0,0,0}$r_3$}%
}}}}
\put(8462,-1537){\makebox(0,0)[b]{\smash{{\SetFigFont{10}{12.0}{\rmdefault}{\mddefault}{\updefault}{\color[rgb]{0,0,0}$r_4$}%
}}}}
\put(5311,-5461){\makebox(0,0)[b]{\smash{{\SetFigFont{10}{12.0}{\rmdefault}{\mddefault}{\updefault}{\color[rgb]{0,0,0}$q$}%
}}}}
\put(8956,-2221){\makebox(0,0)[b]{\smash{{\SetFigFont{10}{12.0}{\rmdefault}{\mddefault}{\updefault}{\color[rgb]{0,0,0}$r^+=r_5$}%
}}}}
\put(2521,-1186){\makebox(0,0)[b]{\smash{{\SetFigFont{10}{12.0}{\rmdefault}{\mddefault}{\updefault}{\color[rgb]{0,0,0}$r^-=r_0$}%
}}}}
\end{picture}%

%% file: degenerateSweep.pspdftex
\begin{picture}(0,0)%
\includegraphics{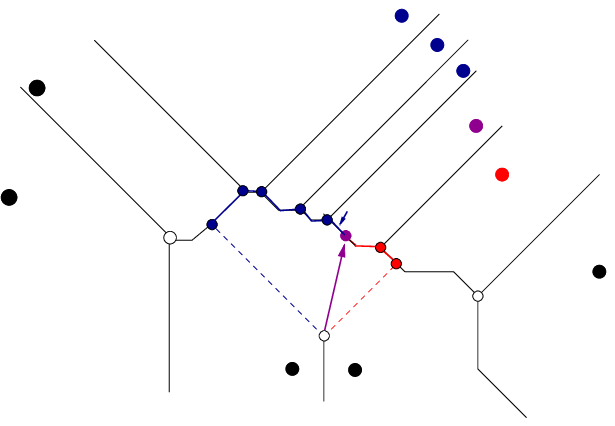}%
\end{picture}%
\setlength{\unitlength}{1657sp}%
\begingroup\makeatletter\ifx\SetFigFont\undefined%
\gdef\SetFigFont#1#2#3#4#5{%
  \reset@font\fontsize{#1}{#2pt}%
  \fontfamily{#3}\fontseries{#4}\fontshape{#5}%
  \selectfont}%
\fi\endgroup%
\begin{picture}(6933,5157)(476,-4836)
\put(4313,-3823){\makebox(0,0)[lb]{\smash{{\SetFigFont{9}{10.8}{\rmdefault}{\mddefault}{\updefault}{\color[rgb]{0,0,0}$\zeta$}%
}}}}
\put(4263,-3238){\makebox(0,0)[rb]{\smash{{\SetFigFont{9}{10.8}{\rmdefault}{\mddefault}{\updefault}{\color[rgb]{.56,0,.56}$\rho$}%
}}}}
\put(4298,-2827){\makebox(0,0)[rb]{\smash{{\SetFigFont{9}{10.8}{\rmdefault}{\mddefault}{\updefault}{\color[rgb]{.56,0,.56}$\eta$}%
}}}}
\put(6051,-1191){\makebox(0,0)[lb]{\smash{{\SetFigFont{9}{10.8}{\rmdefault}{\mddefault}{\updefault}{\color[rgb]{.56,0,.56}$r_4$}%
}}}}
\put(6226,-1766){\makebox(0,0)[lb]{\smash{{\SetFigFont{9}{10.8}{\rmdefault}{\mddefault}{\updefault}{\color[rgb]{1,0,0}$r^+=r_5$}%
}}}}
\put(4586,-3539){\makebox(0,0)[lb]{\smash{{\SetFigFont{9}{10.8}{\rmdefault}{\mddefault}{\updefault}{\color[rgb]{1,0,0}$\rho^+$}%
}}}}
\put(5622,-4718){\makebox(0,0)[b]{\smash{{\SetFigFont{9}{10.8}{\rmdefault}{\mddefault}{\updefault}{\color[rgb]{0,0,0}$\Vor(p)$}%
}}}}
\put(3574,-3421){\makebox(0,0)[rb]{\smash{{\SetFigFont{9}{10.8}{\rmdefault}{\mddefault}{\updefault}{\color[rgb]{0,0,.56}$\rho^-$}%
}}}}
\put(5226, 54){\makebox(0,0)[lb]{\smash{{\SetFigFont{9}{10.8}{\rmdefault}{\mddefault}{\updefault}{\color[rgb]{0,0,.56}$r_1$}%
}}}}
\put(5541,-261){\makebox(0,0)[lb]{\smash{{\SetFigFont{9}{10.8}{\rmdefault}{\mddefault}{\updefault}{\color[rgb]{0,0,.56}$r_2$}%
}}}}
\put(5901,-546){\makebox(0,0)[lb]{\smash{{\SetFigFont{9}{10.8}{\rmdefault}{\mddefault}{\updefault}{\color[rgb]{0,0,.56}$r_3$}%
}}}}
\put(5121,-2886){\makebox(0,0)[lb]{\smash{{\SetFigFont{9}{10.8}{\rmdefault}{\mddefault}{\updefault}{\color[rgb]{1,0,0}$\eta^+$}%
}}}}
\put(4513,-2353){\makebox(0,0)[lb]{\smash{{\SetFigFont{8}{9.6}{\rmdefault}{\mddefault}{\updefault}{\color[rgb]{0,0,.56}$e_{qr_4}$}%
}}}}
\put(4410,-3008){\makebox(0,0)[lb]{\smash{{\SetFigFont{8}{9.6}{\rmdefault}{\mddefault}{\updefault}{\color[rgb]{1,0,0}$e_{pr_4}$}%
}}}}
\put(5208,-3355){\makebox(0,0)[lb]{\smash{{\SetFigFont{9}{10.8}{\rmdefault}{\mddefault}{\updefault}{\color[rgb]{0,0,0}$e_{pr_5}$}%
}}}}
\put(2761,-4696){\makebox(0,0)[lb]{\smash{{\SetFigFont{9}{10.8}{\rmdefault}{\mddefault}{\updefault}{\color[rgb]{0,0,0}$\Vor(q)$}%
}}}}
\put(3825,-4431){\makebox(0,0)[b]{\smash{{\SetFigFont{9}{10.8}{\rmdefault}{\mddefault}{\updefault}{\color[rgb]{0,0,0}$q$}%
}}}}
\put(4561,-4442){\makebox(0,0)[b]{\smash{{\SetFigFont{9}{10.8}{\rmdefault}{\mddefault}{\updefault}{\color[rgb]{0,0,0}$p$}%
}}}}
\put(491,-776){\makebox(0,0)[lb]{\smash{{\SetFigFont{9}{10.8}{\rmdefault}{\mddefault}{\updefault}{\color[rgb]{0,0,.56}$r^-=r_0$}%
}}}}
\put(3079,-2367){\makebox(0,0)[rb]{\smash{{\SetFigFont{9}{10.8}{\rmdefault}{\mddefault}{\updefault}{\color[rgb]{0,0,.56}$\eta^-$}%
}}}}
\end{picture}%

%% file: degenerateCorner.pspdftex
\begin{picture}(0,0)%
\includegraphics{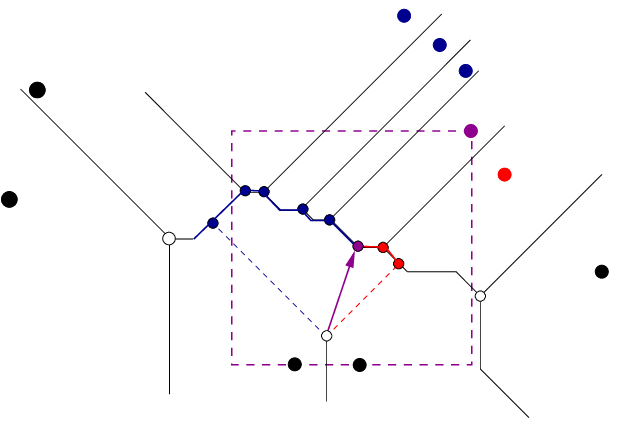}%
\end{picture}%
\setlength{\unitlength}{1657sp}%
\begingroup\makeatletter\ifx\SetFigFont\undefined%
\gdef\SetFigFont#1#2#3#4#5{%
  \reset@font\fontsize{#1}{#2pt}%
  \fontfamily{#3}\fontseries{#4}\fontshape{#5}%
  \selectfont}%
\fi\endgroup%
\begin{picture}(6958,5157)(451,-4836)
\put(466,-799){\makebox(0,0)[lb]{\smash{{\SetFigFont{9}{10.8}{\rmdefault}{\mddefault}{\updefault}{\color[rgb]{0,0,.56}$r^-=r_0$}%
}}}}
\put(4313,-3823){\makebox(0,0)[lb]{\smash{{\SetFigFont{9}{10.8}{\rmdefault}{\mddefault}{\updefault}{\color[rgb]{0,0,0}$\zeta$}%
}}}}
\put(4156,-1166){\makebox(0,0)[rb]{\smash{{\SetFigFont{9}{10.8}{\rmdefault}{\mddefault}{\updefault}{\color[rgb]{.56,0,.56}$Q[\eta]$}%
}}}}
\put(4397,-2855){\makebox(0,0)[rb]{\smash{{\SetFigFont{8}{9.6}{\rmdefault}{\mddefault}{\updefault}{\color[rgb]{.56,0,.56}$\eta$}%
}}}}
\put(5869,-1226){\makebox(0,0)[lb]{\smash{{\SetFigFont{9}{10.8}{\rmdefault}{\mddefault}{\updefault}{\color[rgb]{.56,0,.56}$r_4$}%
}}}}
\put(6226,-1766){\makebox(0,0)[lb]{\smash{{\SetFigFont{9}{10.8}{\rmdefault}{\mddefault}{\updefault}{\color[rgb]{1,0,0}$r^+=r_5$}%
}}}}
\put(3574,-3421){\makebox(0,0)[rb]{\smash{{\SetFigFont{9}{10.8}{\rmdefault}{\mddefault}{\updefault}{\color[rgb]{0,0,.56}$\rho^-$}%
}}}}
\put(4586,-3539){\makebox(0,0)[lb]{\smash{{\SetFigFont{9}{10.8}{\rmdefault}{\mddefault}{\updefault}{\color[rgb]{1,0,0}$\rho^+$}%
}}}}
\put(5901,-546){\makebox(0,0)[lb]{\smash{{\SetFigFont{9}{10.8}{\rmdefault}{\mddefault}{\updefault}{\color[rgb]{0,0,.56}$r_3$}%
}}}}
\put(5541,-261){\makebox(0,0)[lb]{\smash{{\SetFigFont{9}{10.8}{\rmdefault}{\mddefault}{\updefault}{\color[rgb]{0,0,.56}$r_2$}%
}}}}
\put(5226, 54){\makebox(0,0)[lb]{\smash{{\SetFigFont{9}{10.8}{\rmdefault}{\mddefault}{\updefault}{\color[rgb]{0,0,.56}$r_1$}%
}}}}
\put(5121,-2886){\makebox(0,0)[lb]{\smash{{\SetFigFont{9}{10.8}{\rmdefault}{\mddefault}{\updefault}{\color[rgb]{1,0,0}$\eta^+$}%
}}}}
\put(5622,-4718){\makebox(0,0)[b]{\smash{{\SetFigFont{9}{10.8}{\rmdefault}{\mddefault}{\updefault}{\color[rgb]{0,0,0}$\Vor(p)$}%
}}}}
\put(4640,-4465){\makebox(0,0)[b]{\smash{{\SetFigFont{9}{10.8}{\rmdefault}{\mddefault}{\updefault}{\color[rgb]{0,0,0}$p$}%
}}}}
\put(3872,-4450){\makebox(0,0)[b]{\smash{{\SetFigFont{9}{10.8}{\rmdefault}{\mddefault}{\updefault}{\color[rgb]{0,0,0}$q$}%
}}}}
\put(4329,-2528){\makebox(0,0)[lb]{\smash{{\SetFigFont{7}{8.4}{\rmdefault}{\mddefault}{\updefault}{\color[rgb]{0,0,.56}$e_{qr_4}$}%
}}}}
\put(4468,-3031){\makebox(0,0)[lb]{\smash{{\SetFigFont{7}{8.4}{\rmdefault}{\mddefault}{\updefault}{\color[rgb]{1,0,0}$e_{pr_4}$}%
}}}}
\put(2503,-4440){\makebox(0,0)[lb]{\smash{{\SetFigFont{9}{10.8}{\rmdefault}{\mddefault}{\updefault}{\color[rgb]{0,0,0}$\Vor(q)$}%
}}}}
\put(3109,-2332){\makebox(0,0)[rb]{\smash{{\SetFigFont{9}{10.8}{\rmdefault}{\mddefault}{\updefault}{\color[rgb]{0,0,.56}$\eta^-$}%
}}}}
\end{picture}%

%% file: DegenerateFlip.pspdftex
\begin{picture}(0,0)%
\includegraphics{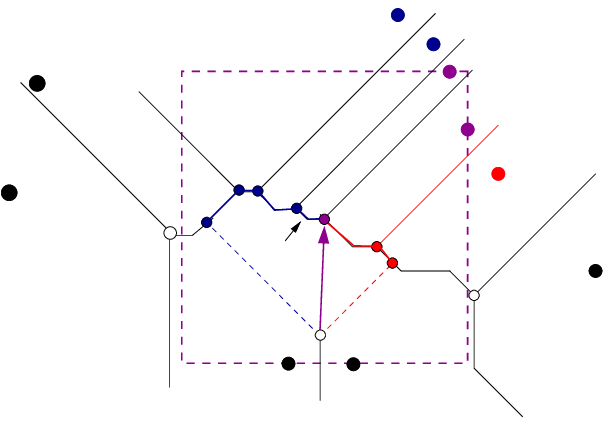}%
\end{picture}%
\setlength{\unitlength}{1657sp}%
\begingroup\makeatletter\ifx\SetFigFont\undefined%
\gdef\SetFigFont#1#2#3#4#5{%
  \reset@font\fontsize{#1}{#2pt}%
  \fontfamily{#3}\fontseries{#4}\fontshape{#5}%
  \selectfont}%
\fi\endgroup%
\begin{picture}(6888,5270)(521,-4949)
\put(536,-731){\makebox(0,0)[lb]{\smash{{\SetFigFont{9}{10.8}{\rmdefault}{\mddefault}{\updefault}{\color[rgb]{0,0,.56}$r^-=r_0$}%
}}}}
\put(4313,-3823){\makebox(0,0)[lb]{\smash{{\SetFigFont{9}{10.8}{\rmdefault}{\mddefault}{\updefault}{\color[rgb]{0,0,0}$\zeta$}%
}}}}
\put(3700,-589){\makebox(0,0)[rb]{\smash{{\SetFigFont{9}{10.8}{\rmdefault}{\mddefault}{\updefault}{\color[rgb]{.56,0,.56}$Q[\eta]$}%
}}}}
\put(4117,-2300){\makebox(0,0)[lb]{\smash{{\SetFigFont{8}{9.6}{\rmdefault}{\mddefault}{\updefault}{\color[rgb]{.56,0,.56}$\eta$}%
}}}}
\put(4640,-4465){\makebox(0,0)[b]{\smash{{\SetFigFont{9}{10.8}{\rmdefault}{\mddefault}{\updefault}{\color[rgb]{0,0,0}$p$}%
}}}}
\put(3872,-4450){\makebox(0,0)[b]{\smash{{\SetFigFont{9}{10.8}{\rmdefault}{\mddefault}{\updefault}{\color[rgb]{0,0,0}$q$}%
}}}}
\put(4586,-3539){\makebox(0,0)[lb]{\smash{{\SetFigFont{9}{10.8}{\rmdefault}{\mddefault}{\updefault}{\color[rgb]{1,0,0}$\rho^+$}%
}}}}
\put(3574,-3421){\makebox(0,0)[rb]{\smash{{\SetFigFont{9}{10.8}{\rmdefault}{\mddefault}{\updefault}{\color[rgb]{0,0,.56}$\rho^-$}%
}}}}
\put(3109,-2257){\makebox(0,0)[rb]{\smash{{\SetFigFont{9}{10.8}{\rmdefault}{\mddefault}{\updefault}{\color[rgb]{0,0,.56}$\eta^-$}%
}}}}
\put(5121,-2886){\makebox(0,0)[lb]{\smash{{\SetFigFont{9}{10.8}{\rmdefault}{\mddefault}{\updefault}{\color[rgb]{1,0,0}$\eta^+$}%
}}}}
\put(6226,-1766){\makebox(0,0)[lb]{\smash{{\SetFigFont{9}{10.8}{\rmdefault}{\mddefault}{\updefault}{\color[rgb]{1,0,0}$r^+=r_5$}%
}}}}
\put(5877,-1235){\makebox(0,0)[lb]{\smash{{\SetFigFont{9}{10.8}{\rmdefault}{\mddefault}{\updefault}{\color[rgb]{.56,0,.56}$r_4$}%
}}}}
\put(5815,-641){\makebox(0,0)[lb]{\smash{{\SetFigFont{9}{10.8}{\rmdefault}{\mddefault}{\updefault}{\color[rgb]{.56,0,.56}$r_3$}%
}}}}
\put(5541,-261){\makebox(0,0)[lb]{\smash{{\SetFigFont{9}{10.8}{\rmdefault}{\mddefault}{\updefault}{\color[rgb]{0,0,.56}$r_2$}%
}}}}
\put(5226, 54){\makebox(0,0)[lb]{\smash{{\SetFigFont{9}{10.8}{\rmdefault}{\mddefault}{\updefault}{\color[rgb]{0,0,.56}$r_1$}%
}}}}
\put(4329,-2528){\makebox(0,0)[lb]{\smash{{\SetFigFont{7}{8.4}{\rmdefault}{\mddefault}{\updefault}{\color[rgb]{1,0,0}$e_{pr_4}$}%
}}}}
\put(2746,-4831){\makebox(0,0)[lb]{\smash{{\SetFigFont{9}{10.8}{\rmdefault}{\mddefault}{\updefault}{\color[rgb]{0,0,0}$\Vor(q)$}%
}}}}
\put(5716,-4831){\makebox(0,0)[b]{\smash{{\SetFigFont{9}{10.8}{\rmdefault}{\mddefault}{\updefault}{\color[rgb]{0,0,0}$\Vor(p)$}%
}}}}
\put(4096,-2806){\makebox(0,0)[rb]{\smash{{\SetFigFont{7}{8.4}{\rmdefault}{\mddefault}{\updefault}{\color[rgb]{0,0,.56}$e_{qr_3}$}%
}}}}
\end{picture}%

%% file: degenerateFlipDelaunay-new.pspdftex
\begin{picture}(0,0)%
\includegraphics{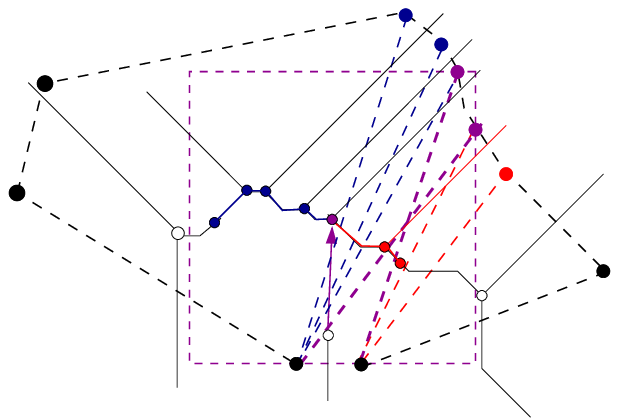}%
\end{picture}%
\setlength{\unitlength}{1657sp}%
\begingroup\makeatletter\ifx\SetFigFont\undefined%
\gdef\SetFigFont#1#2#3#4#5{%
  \reset@font\fontsize{#1}{#2pt}%
  \fontfamily{#3}\fontseries{#4}\fontshape{#5}%
  \selectfont}%
\fi\endgroup%
\begin{picture}(7023,5015)(386,-4694)
\put(4640,-4465){\makebox(0,0)[b]{\smash{{\SetFigFont{9}{10.8}{\rmdefault}{\mddefault}{\updefault}{\color[rgb]{0,0,0}$p$}%
}}}}
\put(3872,-4450){\makebox(0,0)[b]{\smash{{\SetFigFont{9}{10.8}{\rmdefault}{\mddefault}{\updefault}{\color[rgb]{0,0,0}$q$}%
}}}}
\put(6226,-1766){\makebox(0,0)[lb]{\smash{{\SetFigFont{9}{10.8}{\rmdefault}{\mddefault}{\updefault}{\color[rgb]{1,0,0}$r^+=r_5$}%
}}}}
\put(5877,-1235){\makebox(0,0)[lb]{\smash{{\SetFigFont{9}{10.8}{\rmdefault}{\mddefault}{\updefault}{\color[rgb]{.56,0,.56}$r_4$}%
}}}}
\put(5815,-641){\makebox(0,0)[lb]{\smash{{\SetFigFont{9}{10.8}{\rmdefault}{\mddefault}{\updefault}{\color[rgb]{.56,0,.56}$r_3$}%
}}}}
\put(5541,-261){\makebox(0,0)[lb]{\smash{{\SetFigFont{9}{10.8}{\rmdefault}{\mddefault}{\updefault}{\color[rgb]{0,0,.56}$r_2$}%
}}}}
\put(5226, 54){\makebox(0,0)[lb]{\smash{{\SetFigFont{9}{10.8}{\rmdefault}{\mddefault}{\updefault}{\color[rgb]{0,0,.56}$r_1$}%
}}}}
\put(401,-596){\makebox(0,0)[lb]{\smash{{\SetFigFont{9}{10.8}{\rmdefault}{\mddefault}{\updefault}{\color[rgb]{0,0,.56}$r^-=r_0$}%
}}}}
\end{picture}%

%% file: degenerate-flip-count.pspdftex
\begin{picture}(0,0)%
\includegraphics{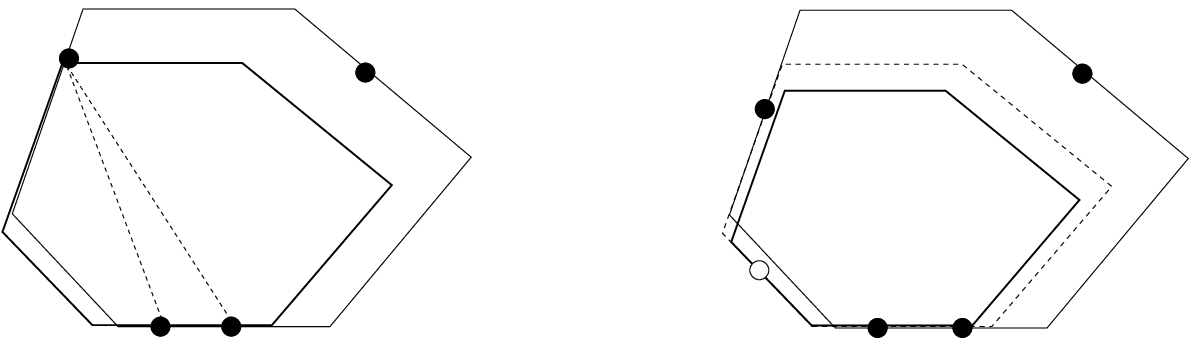}%
\end{picture}%
\setlength{\unitlength}{2072sp}%
\begingroup\makeatletter\ifx\SetFigFont\undefined%
\gdef\SetFigFont#1#2#3#4#5{%
  \reset@font\fontsize{#1}{#2pt}%
  \fontfamily{#3}\fontseries{#4}\fontshape{#5}%
  \selectfont}%
\fi\endgroup%
\begin{picture}(10876,4298)(1360,-6733)
\put(8202,-3727){\makebox(0,0)[rb]{\smash{{\SetFigFont{10}{12.0}{\rmdefault}{\mddefault}{\updefault}{\color[rgb]{0,0,0}$r$}%
}}}}
\put(11407,-3469){\makebox(0,0)[lb]{\smash{{\SetFigFont{10}{12.0}{\rmdefault}{\mddefault}{\updefault}{\color[rgb]{0,0,0}$s$}%
}}}}
\put(3018,-6633){\makebox(0,0)[lb]{\smash{{\SetFigFont{10}{12.0}{\rmdefault}{\mddefault}{\updefault}{\color[rgb]{0,0,0}(a)}%
}}}}
\put(9582,-6633){\makebox(0,0)[lb]{\smash{{\SetFigFont{10}{12.0}{\rmdefault}{\mddefault}{\updefault}{\color[rgb]{0,0,0}(b)}%
}}}}
\put(1800,-4300){\makebox(0,0)[lb]{\smash{{\SetFigFont{10}{12.0}{\rmdefault}{\mddefault}{\updefault}{\color[rgb]{0,0,0}$e_2$}%
}}}}
\put(3102,-5511){\makebox(0,0)[lb]{\smash{{\SetFigFont{10}{12.0}{\rmdefault}{\mddefault}{\updefault}{\color[rgb]{0,0,0}$e_1$}%
}}}}
\put(5435,-3866){\makebox(0,0)[lb]{\smash{{\SetFigFont{10}{12.0}{\rmdefault}{\mddefault}{\updefault}{\color[rgb]{0,0,0}$e_3$}%
}}}}
\put(9540,-5493){\makebox(0,0)[lb]{\smash{{\SetFigFont{10}{12.0}{\rmdefault}{\mddefault}{\updefault}{\color[rgb]{0,0,0}$e_1$}%
}}}}
\put(11927,-3884){\makebox(0,0)[lb]{\smash{{\SetFigFont{10}{12.0}{\rmdefault}{\mddefault}{\updefault}{\color[rgb]{0,0,0}$e_3$}%
}}}}
\put(8129,-4390){\makebox(0,0)[rb]{\smash{{\SetFigFont{10}{12.0}{\rmdefault}{\mddefault}{\updefault}{\color[rgb]{0,0,0}$e_2$}%
}}}}
\put(3048,-2690){\makebox(0,0)[b]{\smash{{\SetFigFont{10}{12.0}{\rmdefault}{\mddefault}{\updefault}{\color[rgb]{0,0,0}$\tilde Q$}%
}}}}
\put(10028,-2745){\makebox(0,0)[b]{\smash{{\SetFigFont{10}{12.0}{\rmdefault}{\mddefault}{\updefault}{\color[rgb]{0,0,0}$\tilde Q$}%
}}}}
\put(8479,-5204){\makebox(0,0)[lb]{\smash{{\SetFigFont{10}{12.0}{\rmdefault}{\mddefault}{\updefault}{\color[rgb]{0,0,0}$e_{12}$}%
}}}}
\put(1863,-3327){\makebox(0,0)[rb]{\smash{{\SetFigFont{10}{12.0}{\rmdefault}{\mddefault}{\updefault}{\color[rgb]{0,0,0}$r$}%
}}}}
\put(2752,-6124){\makebox(0,0)[rb]{\smash{{\SetFigFont{10}{12.0}{\rmdefault}{\mddefault}{\updefault}{\color[rgb]{0,0,0}$p$}%
}}}}
\put(3648,-6092){\makebox(0,0)[lb]{\smash{{\SetFigFont{10}{12.0}{\rmdefault}{\mddefault}{\updefault}{\color[rgb]{0,0,0}$q$}%
}}}}
\put(4794,-3355){\makebox(0,0)[lb]{\smash{{\SetFigFont{10}{12.0}{\rmdefault}{\mddefault}{\updefault}{\color[rgb]{0,0,0}$s$}%
}}}}
\put(9322,-6097){\makebox(0,0)[rb]{\smash{{\SetFigFont{10}{12.0}{\rmdefault}{\mddefault}{\updefault}{\color[rgb]{0,0,0}$p$}%
}}}}
\put(10252,-6127){\makebox(0,0)[lb]{\smash{{\SetFigFont{10}{12.0}{\rmdefault}{\mddefault}{\updefault}{\color[rgb]{0,0,0}$q$}%
}}}}
\put(8192,-5487){\makebox(0,0)[rb]{\smash{{\SetFigFont{10}{12.0}{\rmdefault}{\mddefault}{\updefault}{\color[rgb]{0,0,0}$w$}%
}}}}
\end{picture}%

%% file: ChewProof.pspdftex
\begin{picture}(0,0)%
\includegraphics{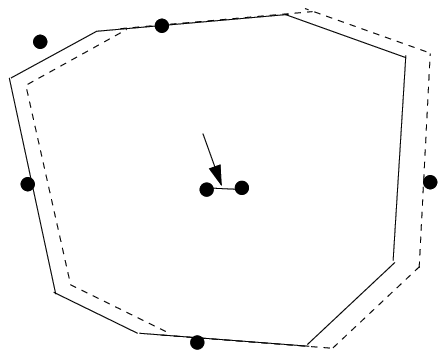}%
\end{picture}%
\setlength{\unitlength}{1973sp}%
\begingroup\makeatletter\ifx\SetFigFont\undefined%
\gdef\SetFigFont#1#2#3#4#5{%
  \reset@font\fontsize{#1}{#2pt}%
  \fontfamily{#3}\fontseries{#4}\fontshape{#5}%
  \selectfont}%
\fi\endgroup%
\begin{picture}(4519,4077)(436,-4145)
\put(1846,-359){\makebox(0,0)[lb]{\smash{{\SetFigFont{11}{13.2}{\rmdefault}{\mddefault}{\updefault}{\color[rgb]{0,0,0}$p_1$}%
}}}}
\put(451,-1912){\makebox(0,0)[lb]{\smash{{\SetFigFont{11}{13.2}{\rmdefault}{\mddefault}{\updefault}{\color[rgb]{0,0,0}$p_4$}%
}}}}
\put(2075,-1503){\makebox(0,0)[lb]{\smash{{\SetFigFont{11}{13.2}{\rmdefault}{\mddefault}{\updefault}{\color[rgb]{0,0,0}$e_{13}$}%
}}}}
\put(4027,-1587){\makebox(0,0)[lb]{\smash{{\SetFigFont{11}{13.2}{\rmdefault}{\mddefault}{\updefault}{\color[rgb]{0,0,0}$e_2$}%
}}}}
\put(665,-2865){\makebox(0,0)[lb]{\smash{{\SetFigFont{11}{13.2}{\rmdefault}{\mddefault}{\updefault}{\color[rgb]{0,0,0}$e_4$}%
}}}}
\put(2702,-3526){\makebox(0,0)[lb]{\smash{{\SetFigFont{11}{13.2}{\rmdefault}{\mddefault}{\updefault}{\color[rgb]{0,0,0}$e_3$}%
}}}}
\put(2702,-877){\makebox(0,0)[lb]{\smash{{\SetFigFont{11}{13.2}{\rmdefault}{\mddefault}{\updefault}{\color[rgb]{0,0,0}$e_1$}%
}}}}
\put(592,-556){\makebox(0,0)[lb]{\smash{{\SetFigFont{11}{13.2}{\rmdefault}{\mddefault}{\updefault}{\color[rgb]{0,0,0}$p_{5}$}%
}}}}
\put(4940,-2183){\makebox(0,0)[lb]{\smash{{\SetFigFont{11}{13.2}{\rmdefault}{\mddefault}{\updefault}{\color[rgb]{0,0,0}$p_2$}%
}}}}
\put(3056,-2230){\makebox(0,0)[lb]{\smash{{\SetFigFont{11}{13.2}{\rmdefault}{\mddefault}{\updefault}{\color[rgb]{0,0,0}$\nu_{123}$}%
}}}}
\put(1809,-2230){\makebox(0,0)[lb]{\smash{{\SetFigFont{11}{13.2}{\rmdefault}{\mddefault}{\updefault}{\color[rgb]{0,0,0}$\nu_{143}$}%
}}}}
\put(2201,-4027){\makebox(0,0)[lb]{\smash{{\SetFigFont{11}{13.2}{\rmdefault}{\mddefault}{\updefault}{\color[rgb]{0,0,0}$p_3$}%
}}}}
\end{picture}%

%% file: ChewProof1.pspdftex
\begin{picture}(0,0)%
\includegraphics{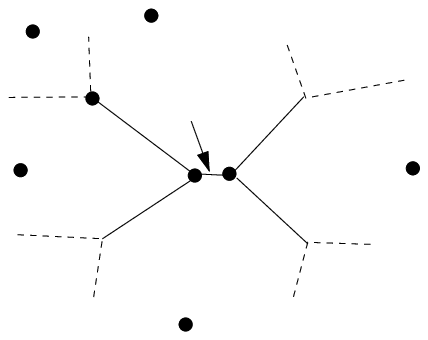}%
\end{picture}%
\setlength{\unitlength}{1973sp}%
\begingroup\makeatletter\ifx\SetFigFont\undefined%
\gdef\SetFigFont#1#2#3#4#5{%
  \reset@font\fontsize{#1}{#2pt}%
  \fontfamily{#3}\fontseries{#4}\fontshape{#5}%
  \selectfont}%
\fi\endgroup%
\begin{picture}(4404,3984)(436,-4041)
\put(1810,-348){\makebox(0,0)[lb]{\smash{{\SetFigFont{11}{13.2}{\rmdefault}{\mddefault}{\updefault}{\color[rgb]{0,0,0}$p_1$}%
}}}}
\put(451,-1863){\makebox(0,0)[lb]{\smash{{\SetFigFont{11}{13.2}{\rmdefault}{\mddefault}{\updefault}{\color[rgb]{0,0,0}$p_4$}%
}}}}
\put(2034,-1464){\makebox(0,0)[lb]{\smash{{\SetFigFont{11}{13.2}{\rmdefault}{\mddefault}{\updefault}{\color[rgb]{0,0,0}$e_{13}$}%
}}}}
\put(4825,-2125){\makebox(0,0)[lb]{\smash{{\SetFigFont{11}{13.2}{\rmdefault}{\mddefault}{\updefault}{\color[rgb]{0,0,0}$p_2$}%
}}}}
\put(588,-543){\makebox(0,0)[lb]{\smash{{\SetFigFont{11}{13.2}{\rmdefault}{\mddefault}{\updefault}{\color[rgb]{0,0,0}$p_{5}$}%
}}}}
\put(3030,-2209){\makebox(0,0)[lb]{\smash{{\SetFigFont{11}{13.2}{\rmdefault}{\mddefault}{\updefault}{\color[rgb]{0,0,0}$\nu_{123}$}%
}}}}
\put(1575,-2160){\makebox(0,0)[lb]{\smash{{\SetFigFont{11}{13.2}{\rmdefault}{\mddefault}{\updefault}{\color[rgb]{0,0,0}$\nu_{143}$}%
}}}}
\put(2158,-3923){\makebox(0,0)[lb]{\smash{{\SetFigFont{11}{13.2}{\rmdefault}{\mddefault}{\updefault}{\color[rgb]{0,0,0}$p_3$}%
}}}}
\end{picture}%

%% file: edgelet.pspdftex
\begin{picture}(0,0)%
\includegraphics{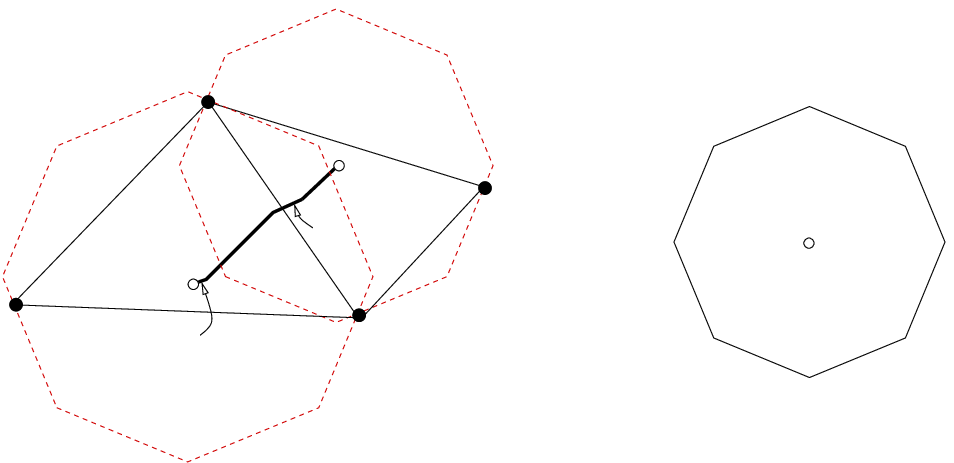}%
\end{picture}%
\setlength{\unitlength}{1973sp}%
\begingroup\makeatletter\ifx\SetFigFont\undefined%
\gdef\SetFigFont#1#2#3#4#5{%
  \reset@font\fontsize{#1}{#2pt}%
  \fontfamily{#3}\fontseries{#4}\fontshape{#5}%
  \selectfont}%
\fi\endgroup%
\begin{picture}(9084,4372)(761,-3566)
\put(3090,-1701){\makebox(0,0)[lb]{\smash{{\SetFigFont{9}{10.8}{\rmdefault}{\mddefault}{\updefault}{\color[rgb]{0,0,0}(2,7)}%
}}}}
\put(3929,-1015){\makebox(0,0)[lb]{\smash{{\SetFigFont{9}{10.8}{\rmdefault}{\mddefault}{\updefault}{\color[rgb]{0,0,0}(3,6)}%
}}}}
\put(3752,-1409){\makebox(0,0)[lb]{\smash{{\SetFigFont{9}{10.8}{\rmdefault}{\mddefault}{\updefault}{\color[rgb]{0,0,0}(3,7)}%
}}}}
\put(9516,-1091){\makebox(0,0)[rb]{\smash{{\SetFigFont{9}{10.8}{\rmdefault}{\mddefault}{\updefault}{\color[rgb]{0,0,0}1}%
}}}}
\put(9528,-1905){\makebox(0,0)[rb]{\smash{{\SetFigFont{9}{10.8}{\rmdefault}{\mddefault}{\updefault}{\color[rgb]{0,0,0}2}%
}}}}
\put(8969,-2413){\makebox(0,0)[rb]{\smash{{\SetFigFont{9}{10.8}{\rmdefault}{\mddefault}{\updefault}{\color[rgb]{0,0,0}3}%
}}}}
\put(8091,-2426){\makebox(0,0)[lb]{\smash{{\SetFigFont{9}{10.8}{\rmdefault}{\mddefault}{\updefault}{\color[rgb]{0,0,0}4}%
}}}}
\put(7532,-1866){\makebox(0,0)[lb]{\smash{{\SetFigFont{9}{10.8}{\rmdefault}{\mddefault}{\updefault}{\color[rgb]{0,0,0}5}%
}}}}
\put(7532,-1040){\makebox(0,0)[lb]{\smash{{\SetFigFont{9}{10.8}{\rmdefault}{\mddefault}{\updefault}{\color[rgb]{0,0,0}6}%
}}}}
\put(8130,-532){\makebox(0,0)[lb]{\smash{{\SetFigFont{9}{10.8}{\rmdefault}{\mddefault}{\updefault}{\color[rgb]{0,0,0}7}%
}}}}
\put(4294,-2431){\makebox(0,0)[lb]{\smash{{\SetFigFont{11}{13.2}{\rmdefault}{\mddefault}{\updefault}{\color[rgb]{0,0,0}$p$}%
}}}}
\put(2625, 88){\makebox(0,0)[lb]{\smash{{\SetFigFont{11}{13.2}{\rmdefault}{\mddefault}{\updefault}{\color[rgb]{0,0,0}$q$}%
}}}}
\put(9019,-532){\makebox(0,0)[rb]{\smash{{\SetFigFont{9}{10.8}{\rmdefault}{\mddefault}{\updefault}{\color[rgb]{0,0,0}0}%
}}}}
\put(2455,-2604){\makebox(0,0)[lb]{\smash{{\SetFigFont{9}{10.8}{\rmdefault}{\mddefault}{\updefault}{\color[rgb]{0,0,0}(2,0)}%
}}}}
\put(776,-2082){\makebox(0,0)[rb]{\smash{{\SetFigFont{11}{13.2}{\rmdefault}{\mddefault}{\updefault}{\color[rgb]{0,0,0}$r_1$}%
}}}}
\put(5569,-964){\makebox(0,0)[lb]{\smash{{\SetFigFont{9}{10.8}{\rmdefault}{\mddefault}{\updefault}{\color[rgb]{0,0,0}$r_2$}%
}}}}
\put(1759,-1816){\makebox(0,0)[lb]{\smash{{\SetFigFont{11}{13.2}{\rmdefault}{\mddefault}{\updefault}{\color[rgb]{0,0,0}$\triangle_1$}%
}}}}
\put(4442,-1383){\makebox(0,0)[lb]{\smash{{\SetFigFont{11}{13.2}{\rmdefault}{\mddefault}{\updefault}{\color[rgb]{0,0,0}$\triangle_2$}%
}}}}
\put(3751,-586){\makebox(0,0)[lb]{\smash{{\SetFigFont{11}{13.2}{\rmdefault}{\mddefault}{\updefault}{\color[rgb]{0,0,0}$\nu_2$}%
}}}}
\put(2476,-1711){\makebox(0,0)[lb]{\smash{{\SetFigFont{11}{13.2}{\rmdefault}{\mddefault}{\updefault}{\color[rgb]{0,0,0}$\nu_1$}%
}}}}
\end{picture}%

%% file: bisector-certificate.pspdftex
\begin{picture}(0,0)%
\includegraphics{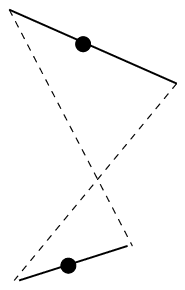}%
\end{picture}%
\setlength{\unitlength}{2072sp}%
\begingroup\makeatletter\ifx\SetFigFont\undefined%
\gdef\SetFigFont#1#2#3#4#5{%
  \reset@font\fontsize{#1}{#2pt}%
  \fontfamily{#3}\fontseries{#4}\fontshape{#5}%
  \selectfont}%
\fi\endgroup%
\begin{picture}(1837,3055)(796,-2987)
\put(1396,-2851){\makebox(0,0)[lb]{\smash{{\SetFigFont{10}{12.0}{\rmdefault}{\mddefault}{\updefault}{\color[rgb]{0,0,0}$p$}%
}}}}
\put(1666,-466){\makebox(0,0)[lb]{\smash{{\SetFigFont{10}{12.0}{\rmdefault}{\mddefault}{\updefault}{\color[rgb]{0,0,0}$q$}%
}}}}
\put(2161,-2581){\makebox(0,0)[lb]{\smash{{\SetFigFont{10}{12.0}{\rmdefault}{\mddefault}{\updefault}{\color[rgb]{0,0,0}$v_i$}%
}}}}
\put(901,-2896){\makebox(0,0)[lb]{\smash{{\SetFigFont{10}{12.0}{\rmdefault}{\mddefault}{\updefault}{\color[rgb]{0,0,0}$v_{i+1}$}%
}}}}
\put(811,-151){\makebox(0,0)[lb]{\smash{{\SetFigFont{10}{12.0}{\rmdefault}{\mddefault}{\updefault}{\color[rgb]{0,0,0}$v_j$}%
}}}}
\put(2521,-826){\makebox(0,0)[lb]{\smash{{\SetFigFont{10}{12.0}{\rmdefault}{\mddefault}{\updefault}{\color[rgb]{0,0,0}$v_{j+1}$}%
}}}}
\end{picture}%